\documentclass[aps,prx,reprint,superscriptaddress,amsmath,amssymb,longbibliography,floatfix]{revtex4-2}

\makeatletter
\def\input@path{{icml/}}
\makeatother

\usepackage[utf8]{inputenc}
\usepackage[T1]{fontenc}
\usepackage{url}
\usepackage{booktabs}
\usepackage{amsfonts}
\usepackage{nicefrac}
\usepackage{microtype}
\usepackage{xcolor}

\newcommand{\answerTODO}[1][]{\textcolor{red}{\bf [TODO]}}
\newcommand{\justificationTODO}[1][]{\textcolor{red}{\bf [TODO]}}

\usepackage{graphicx}
\graphicspath{{icml/}}
\usepackage{comment}
\usepackage{enumitem}
\usepackage{amsmath}
\usepackage{amssymb}
\usepackage{mathtools}

\usepackage{amsthm}
\usepackage{quantikz}
\usepackage{bm}
\usepackage{float}
\makeatletter
\let\prxoriginalnewfloat\newfloat
\let\newfloat\newfloat@ltx
\makeatother
\usepackage{algorithm}
\makeatletter
\let\newfloat\prxoriginalnewfloat
\makeatother
\usepackage{algpseudocode}
\usepackage[hidelinks]{hyperref}
\usepackage[capitalize,noabbrev]{cleveref}

\newcommand{\norm}[1]{\left \lVert #1 \right \rVert}
\newcommand{\abs}[1]{\lvert #1\rvert}
\newcommand{\Abs}[1]{\left\lvert #1\right\rvert}

\providecommand{\bra}[1]{\langle #1|}
\providecommand{\ket}[1]{|#1 \rangle}
\providecommand{\braket}[2]{\langle #1|#2 \rangle}
\providecommand{\ketbra}[2]{|#1\rangle \langle #2|}
\providecommand{\expval}[1]{\langle #1 \rangle}

\newcommand{\dsC}{\mathbb{C}}

\newcommand{\dsE}{\mathbb{E}}

\newcommand{\dsR}{\mathbb{R}}

\newcommand{\dsZ}{\mathbb{Z}}
\newcommand{\scA}{\mathcal{A}}
\newcommand{\scB}{\mathcal{B}}
\newcommand{\scC}{\mathcal{C}}
\newcommand{\scD}{\mathcal{D}}
\newcommand{\scE}{\mathcal{E}}
\newcommand{\scF}{\mathcal{F}}

\newcommand{\scH}{\mathcal{H}}

\newcommand{\scL}{\mathcal{L}}

\newcommand{\scN}{\mathcal{N}}
\newcommand{\scO}{\mathcal{O}}
\newcommand{\scP}{\mathcal{P}}

\newcommand{\scR}{\mathcal{R}}
\newcommand{\scS}{\mathcal{S}}

\newcommand{\scU}{\mathcal{U}}
\newcommand{\scV}{\mathcal{V}}

\newcommand{\scY}{\mathcal{Y}}

\newcommand{\Tr}{\operatorname{Tr}}

\renewcommand{\vr}{\mathbf{r}}

\DeclareMathOperator{\TV}{TV}

\theoremstyle{plain}
\newtheorem{theorem}{Theorem}[section]
\newtheorem{proposition}[theorem]{Proposition}
\newtheorem{lemma}[theorem]{Lemma}

\theoremstyle{definition}
\newtheorem{definition}[theorem]{Definition}

\theoremstyle{remark}

\usepackage[textsize=tiny]{todonotes}

\begin{document}

\title{Hierarchy of Discriminative Power and Complexity in Learning Quantum Ensembles}

\author{Jian Yao}
\affiliation{Ming Hsieh Department of Electrical and Computer Engineering, University of Southern California, Los Angeles, CA 90089}
\author{Pengtao Li}
\affiliation{Department of Mathematics, University of Southern California, Los Angeles, CA 90089}
\author{Xiaohui Chen}
\affiliation{Department of Mathematics, University of Southern California, Los Angeles, CA 90089}
\affiliation{Thomas Lord Department of Computer Science, University of Southern California, Los Angeles, CA 90089}
\affiliation{Computing and Mathematical Sciences, California Institute of Technology, Pasadena, CA 91125}
\author{Quntao Zhuang}
\email{qzhuang@usc.edu}
\affiliation{Ming Hsieh Department of Electrical and Computer Engineering, University of Southern California, Los Angeles, CA 90089}
\affiliation{Department of Physics and Astronomy, University of Southern California, Los Angeles, CA 90089}

\begin{abstract}




Distance metrics are central to machine learning, yet distances between ensembles of quantum states remain poorly understood due to fundamental quantum measurement constraints. We introduce a hierarchy of integral probability metrics, termed MMD-$k$, which generalizes the maximum mean discrepancy to quantum ensembles and exhibits a strict trade-off between discriminative power and statistical efficiency as the moment order $k$ increases. For pure-state ensembles of size $N$, estimating MMD-$k$ with arbitrary measurement schemes requires $\Theta(N^{1-1/k})$ samples for constant $k$. At the same time, we prove that any stable distance metric with full discriminative power admits an $O(N\log N)$ upper bound and an $\Omega(N)$ instance-gap lower bound. For quantum Wasserstein distance, with sufficiently large fixed state dimension, we establish a nearly linear lower bound in the ensemble size at constant additive accuracy, together with an $O(N\log N)$ upper bound. These results provide principled guidance for the design of loss functions in quantum machine learning, as we illustrate in training quantum denoising diffusion probabilistic models.

\end{abstract}

\maketitle

\section{Introduction}

Distance metrics are fundamental in modern statistics and machine learning. They shape how we define learning objectives, compare distributions, generalize models and reason about robustness.
In particular, estimating the distance between two unknown data ensembles from finite samples is a key task in generative learning, where the cost functions must be evaluated without direct access to the underlying probability distributions. Widely used distances such as the Fisher information metric~\citep{amari2016information}, maximum mean discrepancy (MMD)~\citep{MMD_JMLR:v13:gretton12a} and Wasserstein distance (i.e. earth mover's distance)~\citep{Wasserstein_villani2003topics} play a crucial role in enabling efficient learning. 

Quantum machine learning leverages quantum computation not only to accelerate learning tasks with classical data, but also to provide an effective way to learn quantum systems and dynamics where the data is inherently quantum~\cite{biamonte2017quantum,abbas2021power,cerezo2021variational}. 
Generative learning of general quantum state ensembles has been recently explored in applications ranging from device characterization to the study of many-body physics~\cite{QuDDPM_PhysRevLett.132.100602,tezuka2024generative}. In a generative quantum machine learning model (see Fig.~\ref{fig: sc and dp}), quantum measurements on subsystems generate a random label, conditioned on which the model outputs a quantum state, forming an output ensemble of labeled quantum states. To learn an ensemble of states generatively, one compares the output ensemble with samples from the target ensemble. In such cases where the data consists of an ensemble of quantum states, identifying a suitable cost function becomes substantially more challenging, due to the intricacy of quantum physics.

While cost functions based on the quantum Wasserstein distance and fidelity-based MMD distance have been utilized in training these generative quantum models~\cite{QuDDPM_PhysRevLett.132.100602,tezuka2024generative}, the discriminative power and sample complexity of distance metrics between quantum state ensembles remain poorly understood. Although distance measures between individual quantum states are well established, extending these notions to ensembles introduces new challenges: ensemble descriptions must be invariant under relabeling, and distances must be estimated from finite copies of quantum states subject to measurement uncertainty. Understanding how these constraints fundamentally limit what ensemble properties a distance can detect, and how efficiently such distances can be learned from data, is the central focus of this work.

\begin{figure*}[t]
    \centering
    \includegraphics[width=\linewidth]{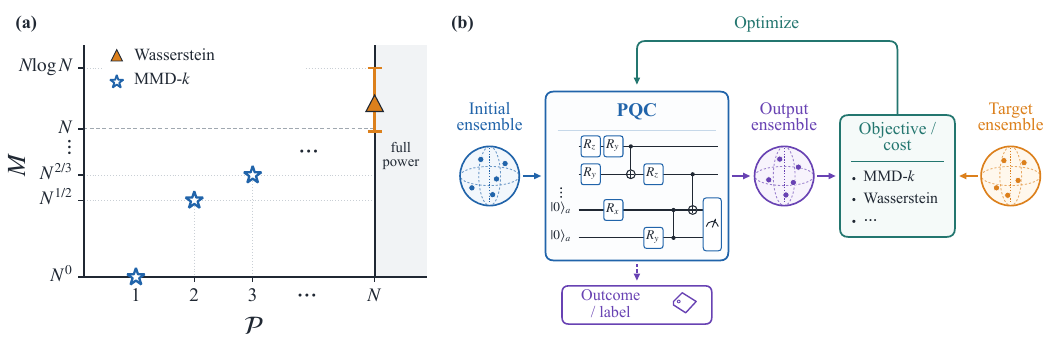}
    \caption{(a) Hierarchy of discriminative power and sample complexity for distances between quantum ensembles. As shown by points of MMD-$k$ (blue star), increasing the order $k$ yields strictly stronger discriminative power---allowing finer distinctions between ensembles---but leads to increasing sample complexity.
At the top of the hierarchy, the Wasserstein distance (orange triangle) achieves full discriminative power. The orange interval illustrates the lower bound $\Omega(N^{1-\gamma})$ (Theorem~\ref{thm:wasserstein-general-lower}) and the upper bound $O(N\log N)$ (Theorem~\ref{thm:general-sample-complexity-tomography}). The lower endpoint is drawn schematically just below $N$, reflecting an almost linear lower bound. This hierarchy formalizes a fundamental tradeoff between what properties of a quantum ensemble can be detected and how efficiently they can be learned from measurements. (b) Generative learning of quantum ensembles. A parameterized quantum circuit (PQC) produces an output ensemble, whose distance to the target ensemble defines the cost used to optimize circuit parameters. The choice of cost function determines which differences between ensembles can be distinguished and the sample complexity of estimating the cost.}
    \label{fig: sc and dp}
\end{figure*}

{\bf Our contributions.} We develop a theory of distance metrics for learning ensembles of quantum states, establishing a foundation for generative quantum machine learning. Our main result is a hierarchy of distance metrics for quantum ensembles that reveals a fundamental tradeoff between discriminative power and statistical sample complexity (Fig.~\ref{fig: sc and dp}).

At a general level, we prove lower and upper bounds for estimating broad classes of distance metrics between quantum ensembles in the fully general quantum measurement model. We then focus on integral probability metrics and show how their ability to distinguish ensembles is constrained by the number of quantum samples required for estimation. To connect the theory to practice, we also propose an experimentally feasible estimation scheme based on the SWAP test~\cite{barenco1997stabilization,buhrman2001quantum}.

More concretely, we introduce a family of (pseudo-) distances, termed MMD-\(k\), that generalizes classical MMD to quantum ensembles. While the case \(k=1\) has previously been used in quantum generative models~\cite{QuDDPM_PhysRevLett.132.100602}, we show that the discriminative power of MMD-\(k\) increases strictly with the moment order \(k\) and saturates at \(k\sim N\), where \(N\) is the number of states in each ensemble. This improved discriminative power comes at a necessary statistical cost: for constant \(k\), estimating MMD-\(k\) to fixed accuracy in the fully general measurement model requires \(\Theta(N^{1-1/k})\) samples. We complement these results with numerical simulations of the corresponding quantum measurement protocol. Furthermore, we show that any stable distance metric with full discriminative power has a sample complexity upper bound \(O(N\log N)\) and instance-gap lower bound of \(\Omega(N)\). For quantum Wasserstein distance, Theorem~\ref{thm:wasserstein-general-lower} gives a lower bound of \(\Omega(N^{1-\gamma})\) at constant additive accuracy for every fixed $\gamma\in(0,1)$, in any fixed dimension $d\ge d_0(\gamma,\epsilon_0)$; the dimension threshold is independent of $N$. Together with tomography, this places the sample complexity between \(\Omega(N^{1-\gamma})\) and \(O(N\log N)\).

Beyond the theory, our results provide a principled guideline for designing loss functions in quantum machine learning: one should use the lowest-order MMD-\(k\) that still discriminates the target ensemble, thereby balancing discriminative power against statistical efficiency. We illustrate this principle by training a quantum denoising diffusion probabilistic model (QuDDPM)~\cite{QuDDPM_PhysRevLett.132.100602}, where higher-order MMD-\(k\) succeeds in regimes where lower-order choices fail.

At a conceptual level, the hierarchy between discriminative power and sample complexity originates from measurement uncertainty in quantum mechanics. More broadly, our results suggest that hierarchies of loss functions may be unavoidable in learning settings where access to data is partial, noisy, or physically constrained.

{\bf Related works.}
{\em Quantum-state complexity.}
Haar-random ensembles and their approximations by \(t\)-designs are central to notions of quantum-state complexity~\cite{ambainis2007quantumtdesignstwiseindependence}. The expressive power of quantum models for such ensembles has been studied via frame potentials and moment operators~\cite{roberts2017chaos}, but the general problem of comparing arbitrary ensembles remains open.

{\em Classical-quantum state approach.}
A labeled ensemble \(\{(p_x,\rho_x)\}\) can be encoded as a classical-quantum state \(\sigma_{\rm CQ}=\sum_x p_x \ketbra{x}{x}\otimes \rho_x\), allowing distances between single quantum states, including quantum MMD~\cite{huang2021quantum} and quantum Wasserstein distances~\cite{chakrabarti2019quantum,de2021quantum,zhou2022quantum,kiani2022learning}, to be applied at the ensemble level. However, this representation is unsuitable for general unlabeled ensembles, where the labels are arbitrary indices: different relabelings produce different CQ states while leaving the underlying ensemble unchanged.

{\em Quantum generative learning.}
QuDDPM~\cite{QuDDPM_PhysRevLett.132.100602} uses MMD-\(1\) and Wasserstein distances to train generative models for quantum ensembles, and \citet{tezuka2024generative} studies Wasserstein-based training for implicit quantum generative models. These works demonstrate the practical relevance of such distances, but do not analyze their discriminative power or the sample complexity of estimating them.

\section{Preliminary: Quantum Ensembles and Distance Metrics}\label{Sec: ensembles and measures}

We first introduce the basics of quantum physics necessary for the discussion. A pure quantum state of $n$ qubits can be described by a vector $\ket{\phi}$ in the $2^n$-dimensional complex Hilbert space $\scH$. The basis of the Hilbert space can be formed by a tensor product of single qubit states $\otimes_{k=1}^n\ket{s_k} $, where $s_k=0,1$ corresponds to the $\ket{0}$ and $\ket{1}$ state of the $k$-th qubit. We denote the complex conjugate of vector as $\bra{\phi}=(\ket{\phi})^\dagger$ and the inner product between two vectors $\ket{\phi}, \ket{\psi}$ as $\braket{\phi}{\psi}$. A pure state is normalized, $\braket{\phi}{\phi}=1$. Closeness between two quantum states can be captured by fidelity $F(\ket{\phi},\ket{\psi})=|\braket{\phi}{\psi}|^2\in[0,1]$.

To describe noise in quantum states, we introduce the positive semi-definite density operator, $\rho\in \scB(\scH)$, where $\scB(\scH)$ is the set of linear operators on a finite Hilbert space $\scH$. Normalization of probability requires $\Tr(\rho) = 1$. We also introduce the notation $\rho^{\otimes k}=\rho\otimes \cdots\otimes \rho$ to represent the tensor product of $k$ copies of $\rho$. Generally, the purity $\Tr(\rho^2)\in[1/2^n, 1]$; $\Tr(\rho^2)=1$ if and only if $\rho=\ketbra{\phi}{\phi}$ is a pure state. A nice property of density operator is that $\sum_x p_x \rho_x$ directly represents the probability mixture of mixed states $\rho_x$ according to the probability $p_x$.

The focus of the study is on ensembles of quantum states, which we define below.

\begin{definition}[\textbf{Ensemble of quantum states}]
    A quantum ensemble $\scE$ is a set of $N$ pairs $\{(p_x,\rho_x)\}_{x=1}^N$ of probabilities $p_x$ and distinct density operators $\rho_x\in \scB(\scH)$. Each index $[x]$ is assigned to a state. When we sample a state from $\scE$, we obtain a physical copy of state $\rho_x$ and the corresponding index $x$ with probability $p_x$. \label{def: quantum ensemble}
\end{definition}
In Definition~\ref{def: quantum ensemble}, by sampling from $\scE$ to obtain $\rho_x$, the state is stored in a quantum memory as a physical copy---the classical description of $\rho_x$ is unknown. To extract information from the sample $\rho_x$, measurements or other quantum operations are performed. The problem becomes interesting as quantum measurements are destructive and unknown quantum states cannot be cloned perfectly.

Now we proceed to address the discrimination between ensembles.
\begin{definition}[\textbf{Equivalence of ensembles}]
    Given two ensembles $\scE_1 = \{(p_i,\ket{\psi_i})\}$ and $\scE_2 = \{(q_j,\ket{\phi_j})\}$, $\scE_1 = \scE_2$ if and only if $p_i=q_j, \ket{\psi_i} = \ket{\phi_j}$ for some map $[i]\rightarrow[j]$, otherwise we say $\scE_1\ne \scE_2$. \label{def: equivalence}
\end{definition}

As a result, any distance between ensembles must be invariant under permutations of ensemble elements, a property not satisfied by distances defined on CQ states discussed in the introduction. Indeed, Definition \ref{def: equivalence} avoids the shortcoming of the CQ state approach, and is applicable to describe quantum state ensemble such as the Haar random ensemble and t-design ensemble~\cite{ambainis2007quantumtdesignstwiseindependence}.

To quantify the distance between two ensembles, we need to introduce a distance metric. For example, one can adopt the Wasserstein distance \cite{Wasserstein_PhysRevA.79.032336,Wasserstein_MAL-073,Wasserstein_villani2003topics}; let $p$ and $q$ be histograms representing $\scE_1$ and $\scE_2$ and assume they have $N_1$ and $N_2$ states respectively, define the cost matrix $C$ by $C_{ij} = 1-\abs{\braket{\psi_i}{\phi_j}}^2$, we have the Wasserstein distance
\begin{equation}
\centering
\begin{aligned}
W(\scE_1,\scE_2) &= \min_{P}\; \langle P, C\rangle \\
\text{s.t.}\qquad & P\,\mathbf{1}_{N_1} = p, \quad P^{\top}\mathbf{1}_{N_2} = q, \quad P \ge 0 .
\end{aligned}\label{eq: calculation of Wasserstein}
\end{equation}
We say Wasserstein distance has full discriminative power as a distance metric of two ensembles, since $W(\scE_1,\scE_2) = 0$ if and only if $\scE_1=\scE_2$. However, not all pseudo distance metrics have the full discriminative power. Consider the maximum mean discrepancy (MMD) \cite{MMD_JMLR:v13:gretton12a} used in \cite{QuDDPM_PhysRevLett.132.100602}, which is defined as
\begin{equation}
    \scD_{\text{MMD}}(\scE_1,\scE_2) = \bar{F}(\scE_1, \scE_1)+\bar{F}(\scE_2, \scE_2) - 2\bar{F}(\scE_1, \scE_2), \label{eq:MMD}
\end{equation}
where $\bar{F}(\scE_a, \scE_b) = \dsE_{\ket{\psi}\sim\scE_a, \ket{\phi}\sim\scE_b} [\abs{\braket{\psi}{\phi}}^2]$. It turns out that $\scD_{\text{MMD}}(\scE_1,\scE_2) =0$ as long as $\scE_1$ and $\scE_2$ have the same average state, $\sum p_i\ketbra{\psi_i}{\psi_i}=\sum q_j \ketbra{\phi_j}{\phi_j}$.

It is obvious that the above MMD has less discriminative power than Wasserstein distance. To enable a comparison, we need to formulate the discriminative power. Given a distance metric $\scD$ and a pair of ensembles $\scE_1 \neq \scE_2$, we say $\scD$ can \emph{discriminate} $(\scE_1, \scE_2)$ if $\scD(\scE_1,\scE_2)\ne0$. If $\scD(\scE_1,\scE_2) = 0 \Longleftrightarrow \scE_1=\scE_2$, we say $\scD$ has the full discriminative power. To provide a more concrete measure for discriminative power, we define a discriminative power based on moments.



\begin{definition}[\textbf{Discriminative power based on moments}]
    Suppose a metric $\scD$ can discriminate all pairs of ensembles $(\scE_1, \scE_2)$ with $\dsE_{\rho\sim\scE_1}[\rho^{\otimes k}] \ne \dsE_{\sigma\sim\scE_2}[\sigma^{\otimes k}]$, but can not discriminate any pair with $\dsE_{\rho\sim\scE_1}[\rho^{\otimes k}] = \dsE_{\sigma\sim\scE_2}[\sigma^{\otimes k}]$, we say the discriminative power based on moments of $\scD$ is $k$, denoted by $\scP(\scD)=k$.\label{def: discriminative power moments}
\end{definition}

Intuitively, $\scP(\scD)=k$ means that the distance $\scD$ can detect discrepancies between two ensembles only through differences in their $k$-th ensemble moments, while all lower-order moments coincide. For $\scD$ with full discriminative power, then we have $\scP(\scD)=N$ achieves the maximum for describing ensembles with $N$ elements.

\section{MMD-$k$: distance metrics with a hierarchy of discriminative power}
\label{MMDk:dis power}

In this section, we give the definition of a family of pseudo distance metrics MMD-$k$ and its properties (with proofs in Appendix~\ref{appsub: proofs of MMD-k's properties}).

\begin{definition}[\textbf{MMD-$k$ pseudo distance metrics}]
    The MMD-$k$ of two quantum ensembles $\scE_1$ and $\scE_2$ is
    \begin{equation}
    \scD^{(k)}(\scE_1,\scE_2) = \bar{F}^{(k)}(\scE_1, \scE_1)+\bar{F}^{(k)}(\scE_2, \scE_2) - 2\bar{F}^{(k)}(\scE_1, \scE_2), \label{eq: k-MMD}
\end{equation}
where $\bar{F}^{(k)}(\scE_a, \scE_b) = \dsE_{\ket{\psi}\sim\scE_a, \ket{\phi}\sim\scE_b} [\abs{\braket{\psi}{\phi}}^{2k}]$ with a tunable parameter $k$.\label{def:k-MMD}
\end{definition}
It is clear that the MMD used in \cite{QuDDPM_PhysRevLett.132.100602} is MMD-1. And by Definition \ref{def:k-MMD}, we have
\begin{proposition}
\label{prop:DK}
    \begin{equation}
        \scD^{(k)}(\scE_1,\scE_2) = \Tr[(\dsE_{\rho\sim\scE_1}[\rho^{\otimes k}]- \dsE_{\sigma\sim\scE_2}[\sigma^{\otimes k}])^2],
    \end{equation} where $\dsE_{\rho\sim\scE_1}[\rho^{\otimes k}]$ and $\dsE_{\sigma\sim\scE_2}[\sigma^{\otimes k}]$ are the $k$-th moment operator of two ensembles, that is, the average of the $k$-th fold of the density operators. 
\end{proposition}
Proposition~\ref{prop:DK} indicates that the MMD-$k$ distance can be interpreted as the Hilbert-Schmidt distance between the two moment operators of the ensembles. Related to our definition of MMD-$k$, the $k$-th-order frame potential~\cite{roberts2017chaos} widely adopted in characterizing an ensemble of random quantum states is the trace of the moment operator.
\begin{lemma}[\textbf{The hierarchy of discriminative power of MMD-$k$}]
    $D^{(k)}(\scE_1,\scE_2)=0$ if and only if $\dsE_{\rho\sim\scE_1}[\rho^{\otimes k}] = \dsE_{\sigma\sim\scE_2}[\sigma^{\otimes k}]$; if $D^{(k)}(\scE_1,\scE_2)=0$, $D^{(k')}(\scE_1,\scE_2)=0$ and $\dsE_{\rho\sim\scE_1}[\rho^{\otimes k'}] = \dsE_{\sigma\sim\scE_2}[\sigma^{\otimes k'}]$ for $k \ge k'$. Then, $\scP(\scD^{(k)})\ge\scP(\scD^{(k')})$ for $k\ge k'$. \label{thm: hierarachy of MMD-k}
\end{lemma}
Lemma \ref{thm: hierarachy of MMD-k} shows that $\scP(D^{(k)})$ increases as $k$ increases, and we care about the threshold of $k$ that gives the full discriminative power.
\begin{lemma}[\textbf{Full discriminative power threshold at $k=N$}]
    Suppose the ensembles we care about each have at most $N$ states with finite dimension $d\ge2$, then $\scP(D^{(k)})=N$ attains the maximum (indicating full discriminative power) when $k\ge N$. \label{thm: k to reach full disriminative power}
\end{lemma}

It is worth noting that in classical statistics, an ensemble of $N$ elements with equal probability also requires up to the $N$-th moment to fully specify. However, in the quantum case, knowing the $k$-th moment operator automatically gives all lower order moment operators by tracing out, while this is not the case for the classical case.

In later sections, we will show there is also a hierarchy of sample complexity for MMD-$k$.

\section{Sampling and measurement setup}
In this section, we present the general process to estimate the values of distance metrics. While in the classical setting the evaluation of distance can be entirely completed by classical computation, in the quantum setting, quantum measurements are necessary to estimate such distance metrics. Moreover, due to the random and destructive nature of quantum measurements, multiple copies of quantum samples are generally needed, even just to estimate a property of a single quantum state $\ket{\phi}$. 

Here, we focus on the most general measurement setting, empowered by unlimited quantum memory and quantum-gate processing, to reveal fundamental limits. In Appendix~\ref{app:local swap}, we also consider more practical measurement settings without quantum memory. In a general measurement setting, we sample data from each ensemble, store them in a quantum memory and perform possibly joint and adaptive measurement on all quantum data. Generally, we sample $M$ times to obtain label-state pairs $\{(x_{1,\ell},\sigma_{1,\ell})\}_{\ell=1}^{M}$ from $\scE_1$, and $\{(x_{2,\ell},\sigma_{2,\ell})\}_{\ell=1}^{M}$ from $\scE_2$. This leads to a sequence of quantum states $\rho_L=\otimes_{\ell=1}^{M} (\sigma_{1,\ell}\otimes\sigma_{2,\ell})$ and classical labels $L=\{(x_{1,\ell},x_{2,\ell})\}_{\ell=1}^M$. Given the label information, one performs a general measurement on $\rho_L$ to obtain an estimate of the distance metric. 

Mathematically, such a general real-valued estimator is represented by a positive operator valued measure (POVM) on $\dsR$: for every Borel set $S\subseteq\dsR$, there is a POVM element $E_S^L$ such that the probability of the estimator value in $S$ is given by
\begin{equation}
\Pr\!\left[\widehat{\scD}\in S\,\middle|\,L,\rho_L\right]=\Tr(E_S^L\rho_L).
\end{equation}
This includes finite-outcome POVMs followed by arbitrary post-processing: if the protocol measures $\{E_y^L\}_{y\in\scY}$ and outputs the estimation value $\widehat{\scD}=g(y,L)$, then
$
E_S^L=\sum_{y:\,g(y,L)\in S}E_y^L.
$ 
This also includes the case where one reconstructs the states from the ensemble via tomography, and then computes the distance metrics via post-processing.

\section{Sample Complexity of Estimating Distance Metrics}\label{sec: sample complexity}

We define sample complexity based on arbitrary quantum measurement on the collected samples.
\begin{definition}[\textbf{Sample complexity}]
    Given a distance metric $\scD$ and two $N$-state quantum ensembles $\scE_1$ and $\scE_2$, suppose that we can randomly sample up to $M$ times from each of the ensembles. Given fully general measurement on the samples, the sample complexity $M(N,\epsilon,\delta)$ is the number of samples needed to estimate $\scD(\scE_1, \scE_2)$ up to an additive error $\epsilon$ and failure probability $\delta$.\label{def: sample complexity general}
\end{definition}
Unless stated otherwise, we consider uniform pure ensembles in this section; the extension to general nonuniform pure ensembles can be found in Appendix~\ref{app: general case of sample complexity}.

\subsection{Sample complexity of MMD-$k$}
\label{secsub: sc of MMDk general}

Below, we show that the sample complexity of estimating MMD-$k$ has the sharp scaling
\begin{align}
M=\Theta\!\left(N^{1-1/k}\right).
\label{M_scaling_tight}
\end{align}
This is proven by a lower bound in Theorem~\ref{thm:mmdk-lower-fullygeneral} and an upper bound via an achievable protocol in Theorem~\ref{thm:mmdk-upper}.

\begin{theorem}[\textbf{General lower bound for MMD-$k$}]\label{thm:mmdk-lower-fullygeneral}
Fix $k\ge 2$. There exist constants $\epsilon_0>0$, $\delta_0\in(0,1/2)$, and $c_k>0$, depending only on $k$, and there exists a hard family of pairs of $N$-state uniform pure ensembles such that the following holds for all sufficiently large $N$. Any fully general estimator that estimates $D^{(k)}$ with additive error at most $\epsilon_0$ and failure probability at most $\delta_0$ uniformly over this hard family must use
$$
M\ge c_k N^{1-1/k}
$$
samples. Consequently, the same lower bound holds for worst-case estimation over all $N$-state uniform pure ensembles.
\end{theorem}

The lower bound applies to arbitrary adaptive protocols, arbitrary quantum memory, arbitrary coherent measurements across labels, and arbitrary estimators. The bound thus shows an intrinsic hierarchy of sample complexity. We present the proof in Appendix~\ref{app: proof lower general}.

On the other hand, we are able to design a measurement protocol that attains the above lower bound:
\begin{theorem}[\textbf{Upper bound for MMD-$k$}]\label{thm:mmdk-upper}
Fix $k\ge 2$ and set $s:=\epsilon^{-2}\log(1/\delta)$. There exists a sampling protocol which estimates $D^{(k)}(\scE_1,\scE_2)$ for any two $N$-state uniform pure ensembles with additive error $\epsilon$ and failure probability $\delta$ using
\begin{align}
M=
\begin{cases}
\scO\!\left((k!s)^{1/k}N^{1-1/k}\right),
& k!s\lesssim N,\\[1mm]
\scO\!\left(\max\{ks,\,N(\log(N/\delta)+k)\}\right),
& k!s\gtrsim N,
\end{cases}
\end{align}
samples, where the hidden constants may depend on $k$.
\end{theorem}
For the asymptotic scaling in Eq.~\eqref{M_scaling_tight}, we focus on the scaling with $N$---it is in the large-$N$ regime, in which the sparse-collision case $k!s\lesssim N$ applies.

The upper bound is achieved by a coherent measurement protocol that uses quantum memory (with a size at most $M$) to store the collected samples and performs coherent measurements across multiple samples (see Appendix~\ref{app:mmdk-upper-bound}). Noticing the specific form in Eq.~\eqref{eq: k-MMD}, we only need to describe the estimator for the cross term average fidelity $\bar F^{(k)}(\scE_1,\scE_2)$; the two self terms $\bar{F}^{(k)}(\scE_1, \scE_1)$ and $\bar{F}^{(k)}(\scE_2, \scE_2)$ are estimated analogously. 

We will sample equal number of states from both $\scE_1,\scE_2$. Take a batch of $m$ samples from $(\scE_1,\scE_2)$, with labels $(I_t,J_t)$ and registers $A_tB_t$, $t=1,\ldots,m$. Let $\scR(I)$ be the set of $k$-subsets $S\subseteq[m]$ on which the row labels $I_t$ are all equal, and let $\scC(J)$ be the set of $k$-subsets $T\subseteq[m]$ on which the column labels $J_t$ are all equal. Write $R=|\scR(I)|$ and $C=|\scC(J)|$.

If $R,C>0$, choose $S=\{s_1,\ldots,s_k\}\in\scR(I)$ and $T=\{t_1,\ldots,t_k\}\in\scC(J)$ uniformly at random, and perform the two-outcome multi-register SWAP measurement (see Appendix~\ref{app:swap_details} for details of SWAP) with observable
$$
W_{S,T}^{(k)}=\bigotimes_{r=1}^k \mathrm{SWAP}(A_{s_r},B_{t_r}).
$$
Equivalently, the measured projectors are $(I\pm W_{S,T}^{(k)})/2$. This is a single physical measurement on the batch; the random choice of $(S,T)$ is classical randomness. If $R=0$ or $C=0$, the batch outputs $0$. For $m=\Theta_k(N^{1-1/k})$, the probabilities $\Pr[R>0]$ and $\Pr[C>0]$ are bounded below by constants depending only on $k$. With the appropriate constant normalization, the batch output is an unbiased bounded estimator of $\bar F^{(k)}(\scE_1,\scE_2)$. 

\paragraph{Large moment orders.}
The preceding analysis assumes that $k$ is fixed independently of the ensemble size $N$, as is often relevant in training generative models. The sparse-collision analysis is not uniform when $k$ grows with $N$. In the full-power regime $k=\Theta(N)$, the across-label SWAP protocol instead enters a dense row/column-coverage regime.

\begin{theorem}[\textbf{Large-order across-label SWAP protocol}]
\label{thm:large-order-across-swap}
Let $k=\Theta(N)$. For fixed additive error $\epsilon>0$ and failure probability $\delta\in(0,1)$, the across-label block-SWAP protocol estimates $\scD^{(k)}(\scE_1,\scE_2)$ for any two $N$-state uniform pure ensembles using
\[
    M_{\rm across}=\scO(Nk)=\scO(N^2)
\]
samples. For fixed $\epsilon_0\in(0,1/4)$ and $\delta_0\in(0,1/2)$, any across-label block-SWAP estimator that estimates $\bar F^{(k)}$ only through disjoint row and column $k$-collision blocks must use $\Omega(Nk)$ samples to achieve a uniform $(\epsilon_0,\delta_0)$-guarantee. Thus $M_{\rm across}=\Theta(Nk)=\Theta(N^2)$ within this protocol class.
\end{theorem}
The lower bound is specific to this protocol class and does not rule out a different fully general quantum protocol with smaller sample complexity. The proof is given in Appendix~\ref{app:large-order-across-swap-proof}.

\paragraph{Classical analogy.}

To reveal the source of the hierarchy of sample complexity, we consider the classical limit, where the ensembles \(\scE_1\) and \(\scE_2\) consist only of computational-basis states
\[
    \{\ket{0},\ket{1},\ldots,\ket{d-1}\}\subset \dsC^d .
\]
For a probability vector \(p\) on \(\{0,1,\ldots,d-1\}\), write
\[
    \scE_p:=\{(p_x,\ket{x})\}_{x:p_x>0}.
\]
For computational-basis states,
\begin{equation}
    \abs{\braket{x}{y}}^{2k}
    =
    \mathbf 1\{x=y\}
\end{equation}
for every \(k\ge1\). Therefore, for any two classical ensembles \(\scE_p\) and \(\scE_q\),
\begin{equation}
    \bar F^{(k)}(\scE_p,\scE_q)
    =
    \sum_{x=0}^{d-1}p_xq_x,
    \qquad
    \scD^{(k)}(\scE_p,\scE_q)
    =
    \sum_{x=0}^{d-1}(p_x-q_x)^2 .
    \label{eq: classical MMD-k collapse}
\end{equation}
Thus all orders \(k\) coincide in the computational-basis limit. Formally, we have the following theorem (see Appendix~\ref{app:Proof of classical analog} for a proof).

\begin{theorem}[\textbf{Sample complexity of MMD-\(k\), classical limit}]
    \label{thm: sample complexity MMD-k, classical}
    Fix \(d\ge2\) and \(k\ge1\). In the computational-basis classical limit, the sample complexity to estimate MMD-\(k\) up to additive error \(\epsilon\) and failure probability \(\delta\) is
    \begin{equation}
        M
        =
        \Theta\!\left(
            \frac{1}{\epsilon^2}\log\frac1\delta
        \right),
    \end{equation}
    for sufficiently small \(\epsilon>0\) and \(\delta\in(0,1/4)\). The upper bound is achieved by the ordinary SWAP-test protocol, and the lower bound holds for arbitrary protocols in the general measurement model, including arbitrary quantum memory and arbitrary joint/adaptive POVMs.
\end{theorem}

Theorem~\ref{thm: sample complexity MMD-k, classical} shows that the hierarchy of sample complexity is absent in the computational-basis classical limit. In this limit, \(X^k=X\) for every \(k\ge1\), so both the value of MMD-\(k\) and its sample complexity are independent of \(k\). The hierarchy is therefore a genuinely quantum effect, caused by measurement uncertainty for unknown noncommuting states: for a general pair of unknown quantum states, a SWAP test gives only a noisy unbiased sample of the fidelity, and higher powers require higher-order measurement structure.

In the classical setting, there is also a lack of hierarchy of sample complexity as the discriminative power of MMD increases. For MMD in a reproducing kernel Hilbert space (RKHS), the representation power depends on the choice of kernel. For instance, MMD with the Gaussian radial basis function kernel is a metric on the probability space, while polynomial kernels with finite order do not result in an MMD metric. Nevertheless, they have the same learning sample complexity, with unbounded polynomial kernels subject to certain extra regularity conditions~\citep{MMD_JMLR:v13:gretton12a,TahmasebiJegelka2024ICML}. This is in sharp contrast with the quantum-ensemble setting, where increasing the moment order changes the scaling with the ensemble size.

While our results are specific to quantum data, they illustrate a broader phenomenon: when access to data is constrained by physical or algorithmic limitations, increased discriminative power may necessarily incur higher statistical cost.

\subsection{Sample complexity of general distance metrics}

Here, we present some general insights about sample complexity. We first have a general theorem for all metrics with full discriminative power (with proof in Appendix~\ref{app:lower-bound-full-discriminative}).

\begin{theorem}[\textbf{Instance-gap lower bound for full discriminative power}]
\label{thm:lower-bound-full-discriminative}
Fix \(d\ge 2\) and \(\delta_0\in(0,1/2)\). Let \(\scD_N\) be a permutation-invariant
distance on \(N\)-state uniform pure ensembles with full discriminative power. Then, for
all sufficiently large \(N\), there exists a hard instance of $\scE_1, \scE_2$
with $\Delta_N:=\scD_N(\scE_1,\scE_2)>0$, such that any protocol in the general sampling model that estimates $\scD_N(\scE_1,\scE_2)$
to additive error at most
\(\Delta_N/4\) with failure probability at most \(\delta_0\) must use
\[
    M\ge (1-2\delta_0)N=\Omega(N).
\]
\end{theorem}

To upper bound sample complexity via tomography, we need a stability condition ensuring that finite-accuracy state reconstruction changes the distance value only by a controlled amount. (see Appendix~\ref{app:stability} for details).
We call a distance \(\scD_N\) tomography-stable at accuracy \(\epsilon_0\) if there exists
a constant \(\eta_0>0\), independent of \(N\), such that replacing every component state
of \(\scE_1,\scE_2\) by a state with infidelity at most \(\eta_0\) changes
\(\scD_N(\scE_1,\scE_2)\) by at most \(\epsilon_0\).
We prove the stability property for Wasserstein distance and MMD-$k$ distance (see Appendix~\ref{app:stability} for a proof).
\begin{lemma}\label{lemma:stability}
The Wasserstein distance and MMD-$k$ distance with a fixed $k$ are both stable under state perturbations.
\end{lemma}

For the above stable distance metrics, we have (with proof in Appendix~\ref{app:general-sample-complexity-tomography})
\begin{theorem}[\textbf{General sample complexity using tomography evaluation}]
\label{thm:general-sample-complexity-tomography}
Fix \(d\ge 2\), a constant additive error \(\epsilon_0>0\), and a constant failure
probability \(\delta_0\in(0,1)\). Let \(\scD_N\) be a stable,
permutation-invariant distance on \(N\)-state uniform pure ensembles in \(\mathbb C^d\).
Then the tomography-evaluation strategy, which first reconstructs all component states of
\(\scE_1\) and \(\scE_2\) and then computes \(\scD_N\) classically, has sample complexity
\[
    M_{\rm tomo}=\Theta(N\log N)
\]
in the scaling with \(N\). 
The upper bound allows collective measurements in the
tomography step, e.g. using sample-optimal collective tomography~\cite{sample_optimal_tomography_7956181};
the lower bound applies to any tomography-evaluation procedure that reconstructs every
component state to constant infidelity.

\end{theorem}


The qualifications in Theorems~\ref{thm:lower-bound-full-discriminative} and
\ref{thm:general-sample-complexity-tomography} are necessary. Full discriminative power
only determines when the distance is zero; it is unchanged by an \(N\)-dependent rescaling
\(\scD_N\mapsto a_N\scD_N\). Hence a constant-error lower bound requires a constant-gap
normalization condition. Similarly, tomography evaluation gives an \(O(N\log N)\) upper
bound only for distances that are stable under finite-accuracy state reconstruction, and
the \(\Omega(N\log N)\) lower bound is a lower bound for tomography evaluation rather
than for all possible direct estimators.


\subsection{Sample complexity of Wasserstein distance}
Defined in Eq.~\eqref{eq: calculation of Wasserstein}, the Wasserstein distance has full discriminative power and is stable (see Lemma~\ref{lemma:stability}), and therefore admits an \(O(N\log N)\) upper bound. On the lower-bound side, Theorem~\ref{thm:lower-bound-full-discriminative} yields an instance-gap bound: there exists a hard pair of ensembles with gap \(\Delta_N>0\) such that estimating the Wasserstein distance to additive error at most \(\Delta_N/4\) requires \(\Omega(N)\) samples. Moreover, we can prove a sample complexity lower bound for a fixed constant additive error (see proof in Appendix~\ref{app:wasserstein-general-lower}):
\begin{theorem}[\textbf{General lower bound for Wasserstein distance}]
\label{thm:wasserstein-general-lower}
Fix $\gamma\in(0,1)$ and $\epsilon_0,\delta_0\in(0,1/2)$. There exist a dimension $d_0=O_{\epsilon_0}(\gamma^{-1})$ and a constant $c>0$, independent of $N$, and a finite hard family of pairs of $N$-state uniform pure ensembles in $\mathbb C^{d_0}$ such that the following holds for all sufficiently large $N$. Any fully general estimator that estimates $W$ with additive error at most $\epsilon_0$ and failure probability at most $\delta_0$ uniformly over this hard family must use
\[
M\ge cN^{1-\gamma}
\]
samples. Consequently, the same lower bound holds for worst-case estimation in $\mathbb C^d$ for every $d\ge d_0$.
\end{theorem}


\section{Numerical Simulations and Applications}
\label{sec:numerical and applications}
\subsection{Numerical simulations of sample complexity}
\label{sec: simulation complexity}

\begin{figure*}[t]
\centering
\includegraphics[width=0.9\linewidth]{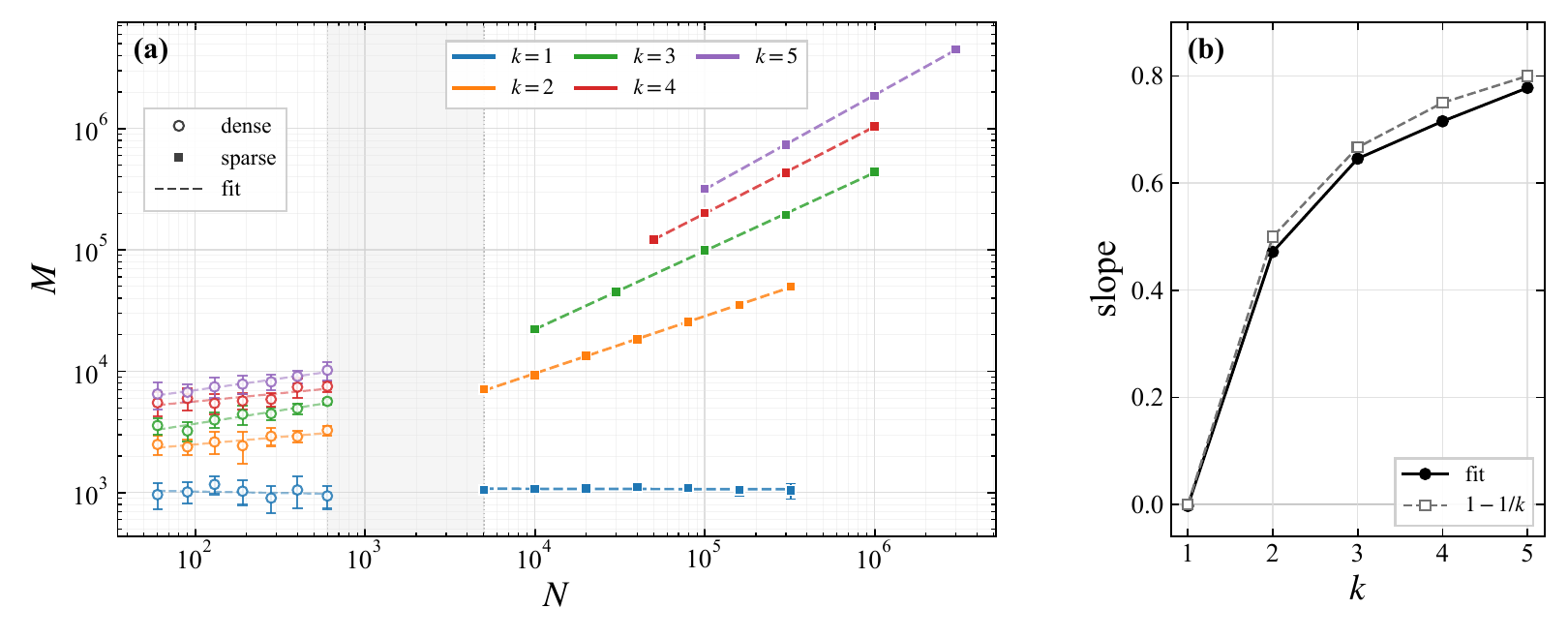}
\caption{
Numerical sample complexity of the across-label multi-copy SWAP estimator. The plot shows the number of samples $M$ required to estimate MMD-$k$ with $\epsilon=0.1$ and $\delta=1/3$ for the cluster and circular one-qubit ensembles. Both axes are logarithmic. Open markers denote the dense full-coverage regime, and filled markers denote the sparse row/column-collision regime. Dashed lines are log-log fits in the sparse regime. The large-$N$ behavior follows the predicted scaling $M\sim N^{1-1/k}$, while the small-$N$ region displays the dense-regime behavior described by Theorem~\ref{thm:mmdk-upper}.
}
\label{fig:numerical-across-label-scaling}
\end{figure*}

To verify the theoretical results in Section \ref{sec: sample complexity}, we conduct numerical simulations to verify the scaling of the number of samples $M(N,\epsilon,\delta)$ needed to estimate MMD-$k$. We care about how $M$ scales with $N$. Suppose we have two infinite-state ensembles $\scE_1$ and $\scE_2$. For each $N$, we randomly sample $N$ states from $\scE_1$ and $\scE_2$ to form two $N$-state ensembles $S_1$ and $S_2$, and then calculate $M$ needed to estimate $\scD(S_1, S_2)$ up to a fixed additive error $\epsilon$ and a fixed failure probability $\delta$. The algorithm to calculate $M(N,\epsilon,\delta)$ for a distance metric $\scD$ and two ensembles $S_1$ and $S_2$ is Algorithm \ref{alg:simulation-upper-bound} in Appendix \ref{app: algorithms}. We choose $\scE_1$ and $\scE_2$ to be two simple ensembles used in the work of QuDDPM \cite{QuDDPM_PhysRevLett.132.100602}. Specifically, $\scE_1$ is the cluster ensemble, consisting of one qubit states $\ket{\psi}\sim \ket{0}+sc\ket{1}$ up to a normalization factor where $\Re (c), \Im(c) \sim \scN(0,1)$ is Gaussian distributed, with the scale factor chosen as $s=0.08$. $\scE_2$ is the circular state ensemble $\scE_{\rm cir}$, consisting of one qubit states $\ket{\psi}\sim e^{-i\theta Y}$, where $Y$ is the Pauli $Y$ and $\theta$ is generated by a uniform distribution on $[0,2\pi)$.

For fixed additive error $\epsilon=0.1$ and failure probability $\delta=1/3$, we estimate the minimum number of samples $M$ required by the coherent measurement protocol. The estimator forms disjoint row and column $k$-collision blocks, pairs them, and performs one $k$-copy SWAP measurement on each block pair. For each $(N,k)$, the reported sample complexity is obtained by bisection over $M$ using repeated measurement trials. Error bars show the standard deviation over independent ensemble realizations.

Figure~\ref{fig:numerical-across-label-scaling} shows the expected two-regime behavior. For small $N$, row and column labels are repeatedly covered, and the sample cost is governed by the dense full-coverage regime. For large $N$, the useful events are sparse row and column $k$-collisions, and the fitted log-log slopes agree with the exponent $1-1/k$. In particular, the $k=1$ curve is approximately flat in $N$, while the slopes increase with $k$, as predicted by the across-label theory.

\begin{figure*}[t]
    \centering
    \includegraphics[width=0.9\linewidth]{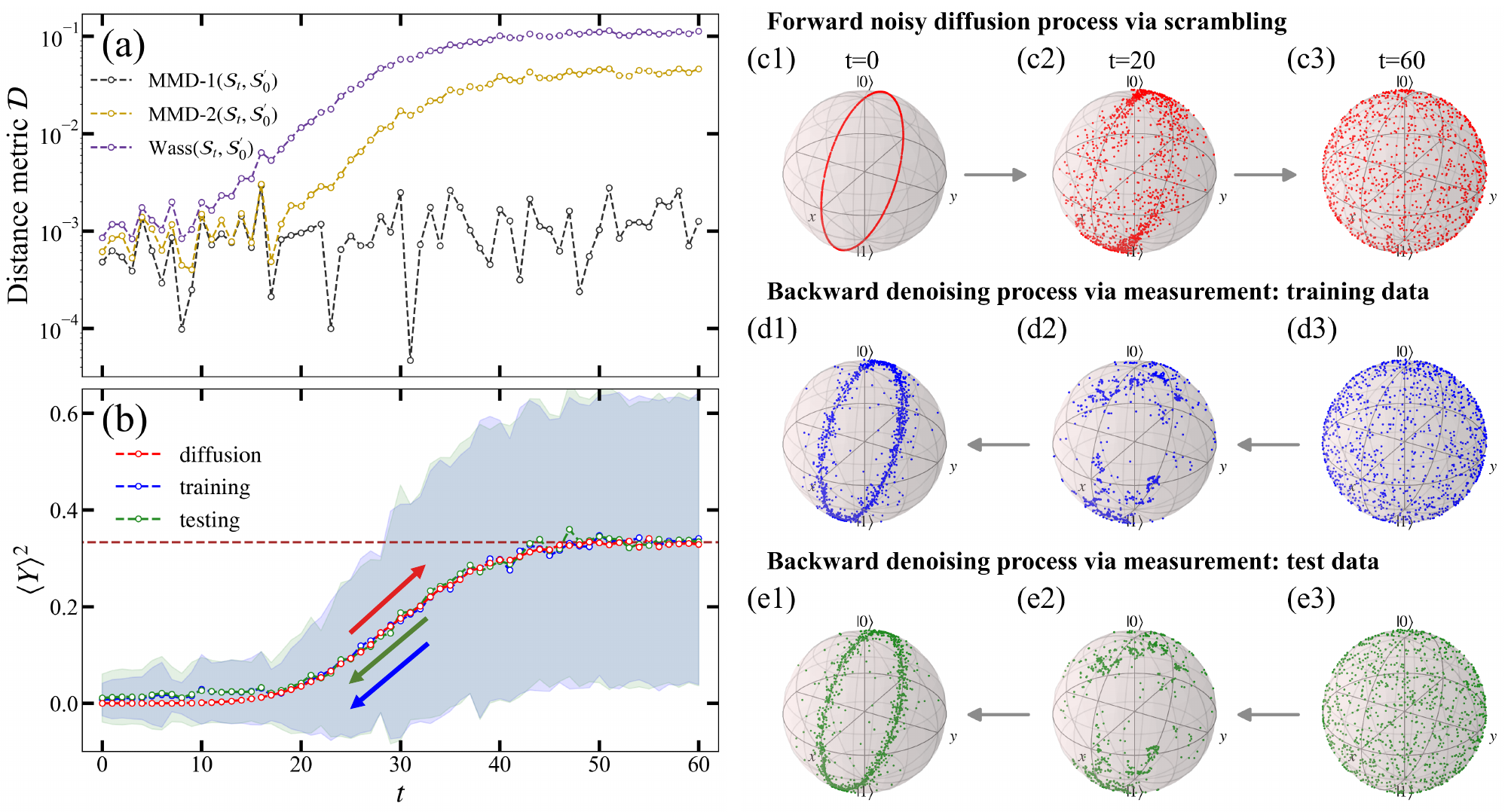}
    \caption{Analysis of effect of loss functions and training performance of QuDDPM. (a) Distance metrics (scaled by log) between ensemble $S_0$ and the ensemble through diffusion process at step $t$, $S_t$ in generation of the circular state ensemble. The data set size of the ensemble is $\abs{\scS}=1000$, and the total number of steps is $T=60$. Due to finite samples, the vanishment is not exact. (b) deviation of generated states from unit circle in X-Z plane. The deviation $\expval{Y}^2$ for forward diffusion (red), backward training (blue), and backward test (green) are plotted. The shaded area shows the sample standard deviation.(c)(d)(e) Bloch visualization of the forward (c1)–(c3) and backward (d1)–(d3),(e1)–(e3) process.}
    \label{fig: loss functions and QUDDPM performance}
\end{figure*}

\subsection{Application: MMD-$2$ as the loss function to learn circular state ensemble}
As an application, we train a QuDDPM \cite{QuDDPM_PhysRevLett.132.100602} to learn the circular state ensemble using MMD-$2$ as the loss function. The goal of QuDDPM is to generate new elements from an unknown ensemble $\scE_0$, given only a finite uniform pure state ensemble $\scS_0=\{\ket{\psi_i}\}\sim\scE_0$. QuDDPM casts this task as a denoising diffusion process on quantum states: a forward noisy process progressively transforms $\scE_0$ into an analytically tractable noise ensemble $\scE_T$, while a parameterized reverse process learns to approximately invert each step of the forward dynamics and thus map samples from $\scE_T$ back to $\scE_0$. In practice, each reverse step is implemented by a trainable quantum circuit $U_{\theta}(t)$ acting on the system and possible ancillas, together with mid-circuit measurements and classical feedforward, and the parameters are optimized by minimizing a statistical discrepancy between the predicted ensemble at each step and the corresponding training ensemble, estimated from finitely many samples. More details about QuDDPM can be found in Appendix~\ref{app:quddpm}.

In QuDDPM, the loss function to estimate the discrepancy should at least be able to discriminate the target ensemble $\scE_0$ and the noise ensemble $\scE_T$. It turns out that MMD-$1$ cannot discriminate the circular state ensemble and the noise ensemble, which is close to the Haar random ensemble. As shown by Fig.~\ref{fig: loss functions and QUDDPM performance} (a), for circular state ensemble, MMD-$1$ (purple) vanishes, while MMD-$2$ (orange) and Wasserstein (blue) can characterize the diffusion of distribution. 

From above, we see that MMD-$2$ is the minimal statistically efficient choice for the task. Therefore, we use MMD-$2$ as the loss function to train QuDDPM to learn the circular state ensemble $\scE_{\text{cir}}$. As shown in Fig.~\ref{fig: loss functions and QUDDPM performance} (c)(d)(e), QuDDPM trained by MMD-$2$ successfully generates elements following the distribution of $\scE_{\text{cir}}$ by denoising process, both on training data and test data. To quantify the performance of QuDDPM, in Fig.~\ref{fig: loss functions and QUDDPM performance} (b), we evaluate the deviation by Pauli $Y$ expectation $\expval{Y}^2$, the values of $\expval{Y}^2$ at $t=0$ are $\expval{Y}^2_{\text{train}}=0.00807 \pm 0.0337$ and $\expval{Y}^2_{\text{test}}=0.01104 \pm 0.0496$ for training data and testing data, while the ground value is $\expval{Y}^2_{\text{data}}=0$. The configuration and performance of model is also presented in Table~\ref{tab:quddpm-training-performance} in Appendix~\ref{app:quddpm}. Our implementation of QuDDPM is a modification based on the code in \cite{QuDDPM_PhysRevLett.132.100602} that is publicly available~\cite{Github}.

To test whether stronger discriminative power improves performance in this task, we also train QuDDPM with MMD-\(3\) and Wasserstein using the same architecture and hyperparameters. As shown in Fig.~\ref{fig: performance-different-loss} in Appendix~\ref{app:quddpm}, the three losses give comparable results, with no clear advantage from the higher-discriminative-power metrics. Understanding when stronger losses improve optimization or generalization remains an open question.

Higher-order losses can nevertheless be necessary in principle. For example, the uniform Bell-state ensemble and the uniform two-qubit computational-basis ensemble both have average state \(I_4/4\), so MMD-\(1\) cannot distinguish them, while
\[
    \scD^{(2)}(\scE_{\mathrm{Bell}},\scE_{\mathrm{comp}})=\frac14,
    \qquad
    W(\scE_{\mathrm{Bell}},\scE_{\mathrm{comp}})=\frac12 .
\]
This illustrates that MMD-\(1\) may miss ensemble structure beyond the first moment, whereas MMD-\(2\) and Wasserstein detect it. Details of this example are in Appendix~\ref{app:quddpm}.

\section{Conclusion and Discussion}

In the main text, we have focused on general measurements with unlimited quantum memory. While revealing the fundamental limits, such measurements may be challenging. To connect to practical experiments, we consider a special class of measurements without quantum memory based on local SWAP tests. In Appendix~\ref{app:local swap}, we derive the sample complexity for both distances under such a measurement setting: $\scO({N^2}\log N)$ for Wasserstein distance and $\Omega(N^{2-2/k})$ for MMD-$k$. 

While our main text has focused on pure-state ensembles, the major results can be extended to mixed states. In Appendix~\ref{app: mixed states}, we consider a practical scenario where quantum state ensembles are subject to noise and generalize the same MMD-$k$ sample complexity scaling to weakly noisy and close-to-pure states.


More broadly, our work suggests that hierarchies of loss functions may be unavoidable in learning settings where only partial or noisy access to data is available.

\begin{acknowledgments}
X.C. acknowledges support from NSF DMS-2413404 and an unrestricted gift from the Simons Foundation. J.Y. and Q.Z. acknowledge support from Office of Naval Research (N00014-23-1-2296, MURI N000142612102), DARPA (HR0011-24-9-0362,HR00112490453,D24AC00153-02), NSF (OMA-2326746, 2350153, CCF-2240641), AFOSR MURI FA9550-24-1-0349, Halliburton Company and an unrestricted gift from Google.
\\

Data availability statement: 
The data and code that support the findings of this article are available from the corresponding author upon reasonable request.
\end{acknowledgments}

\newpage

\bibliographystyle{apsrev4-2}
\bibliography{example_paper}

\newpage
\twocolumngrid
\appendix
\setcounter{tocdepth}{1}
\tableofcontents 

\newpage

\section{Properties of MMD-$k$: Proofs for Section~\ref{MMDk:dis power}}

\label{appsub: proofs of MMD-k's properties}
\subsection{Proof of Proposition \ref{prop:DK}}
\begin{proof}
First, we have
\begin{equation*}
    \abs{\braket{\psi}{\phi}}^2=\braket{\psi}{\phi}\braket{\phi}{\psi}=\Tr(\rho\sigma),
\end{equation*}
where $\rho=\ketbra{\psi}{\psi}$ and $\sigma=\ketbra{\phi}{\phi}$, then we have
\begin{align}
&\bar{F}^{(k)}(\scE_a, \scE_b) 
    \notag\\
    &= \dsE_{\ket{\psi}\sim\scE_a, \ket{\phi}\sim\scE_b} [\abs{\braket{\psi}{\phi}}^{2k}]\notag\\
    &= \dsE_{\rho\sim\scE_a, \sigma\sim\scE_b}[\Tr^k(\rho\sigma)]\notag\\
    &= \dsE_{\rho\sim\scE_a, \sigma\sim\scE_b}[\Tr(\rho^{\otimes k}\sigma^{\otimes k})]\notag\\
    &= \Tr(\dsE_{\rho\sim\scE_a, \sigma\sim\scE_b}[\rho^{\otimes k}\sigma^{\otimes k}]) \notag\\
    &= \Tr(\dsE_{\rho\sim\scE_a}[\rho^{\otimes k}] \dsE_{\sigma\sim\scE_b}[\sigma^{\otimes k}]).
\label{eq:app-frame-potential-moments}
\end{align}
Hence,
\begin{align}
&
    \scD^{(k)}(\scE_1,\scE_2) 
    \notag\\
    &= \bar{F}^{(k)}(\scE_1, \scE_1)+\bar{F}^{(k)}(\scE_2, \scE_2) - 2\bar{F}^{(k)}(\scE_1, \scE_2)\notag\\
    &=\Tr(\dsE_{\rho\sim\scE_1}[\rho^{\otimes k}] \dsE_{\sigma\sim\scE_1}[\sigma^{\otimes k}])\notag\\
    &\quad+ \Tr(\dsE_{\rho\sim\scE_2}[\rho^{\otimes k}] \dsE_{\sigma\sim\scE_2}[\sigma^{\otimes k}])\notag\\
    &\quad-2 \Tr(\dsE_{\rho\sim\scE_1}[\rho^{\otimes k}] \dsE_{\sigma\sim\scE_2}[\sigma^{\otimes k}])\notag\\
    &= \Tr((\dsE_{\rho\sim\scE_1}[\rho^{\otimes k}]- \dsE_{\sigma\sim\scE_2}[\sigma^{\otimes k}])^2).
\label{eq:app-mmd-hilbert-schmidt}
\end{align}
\end{proof}
\subsection{Proof of Lemma \ref{thm: hierarachy of MMD-k}}
\begin{proof}
    First part:
    if $\dsE_{\rho\sim\scE_1}[\rho^{\otimes k}] = \dsE_{\sigma\sim\scE_2}[\sigma^{\otimes k}]$, obviously $\scD^{(k)}(\scE_1, \scE_2) = \Tr((\dsE_{\rho\sim\scE_1}[\rho^{\otimes k}]- \dsE_{\sigma\sim\scE_2}[\sigma^{\otimes k}])^2)=0$, if $\scD^{(k)}(\scE_1, \scE_2) =0$; we write the spectrum decomposition (obviously $\dsE_{\rho\sim\scE_1}[\rho^{\otimes k}]- \dsE_{\sigma\sim\scE_2}[\sigma^{\otimes k}]$ is a hermitian operator):
    \begin{equation*}
        \dsE_{\rho\sim\scE_1}[\rho^{\otimes k}]- \dsE_{\sigma\sim\scE_2}[\sigma^{\otimes k}] = \sum_i\lambda_i\ketbra{\lambda_i}{\lambda_i},
    \end{equation*}
    where $\{\lambda_i\}$ are real eigenvalues and $\{\ket{\lambda_i}\}$ are corresponding eigenvectors that can form an orthonormal basis. Then we have
    \begin{align*}
        \scD^{(k)}(\scE_1, \scE_2)
        &=\Tr((\dsE_{\rho\sim\scE_1}[\rho^{\otimes k}]- \dsE_{\sigma\sim\scE_2}[\sigma^{\otimes k}])^2)\\
        &=\Tr(\sum_i\lambda_i^2\ketbra{\lambda_i}{\lambda_i})=\sum_i\lambda_i^2=0.
    \end{align*}
    So every $\lambda_i=0$, $\dsE_{\rho\sim\scE_1}[\rho^{\otimes k}]- \dsE_{\sigma\sim\scE_2}[\sigma^{\otimes k}]=0$, $\dsE_{\rho\sim\scE_1}[\rho^{\otimes k}]=\dsE_{\sigma\sim\scE_2}[\sigma^{\otimes k}]$.

    Second part: when $\scD^{(k)}(\scE_1, \scE_2)=0$, we have $\dsE_{\rho\sim\scE_1}[\rho^{\otimes k}]=\dsE_{\sigma\sim\scE_2}[\sigma^{\otimes k}]$. Then we have $\dsE_{\rho\sim\scE_1}[\rho^{\otimes k'}]=\dsE_{\sigma\sim\scE_2}[\sigma^{\otimes k'}]$ for $k'\le k$ by partial trace and  we also have $\scD^{(k')}(\scE_1, \scE_2)=0$ for $k'\le k$. By definition (see around Definition \ref{def: discriminative power moments}), $\scP(\scD^{(k)})\ge\scP(\scD^{(k')})$ because that if the pair$(\scE_1,\scE_2)$ can not be discriminated by $\scD^{(k)}$ ($\scD^{(k)}(\scE_1, \scE_2)=0$), then it can not be discriminated by $\scD^{(k')}$, neither ($\scD^{(k')}(\scE_1, \scE_2)=0$), for $k'\le k$.
\end{proof}

\subsection{Proof of Lemma~\ref{thm: k to reach full disriminative power}}

We denote the $k$-th moment $\dsE_{\rho\sim\scE}[\rho^{\otimes k}]$ as $M_k(\scE)$. Notice that for $\rho^{\otimes k} = \ket{\psi}^{\otimes k} \bra{\psi}^{\otimes k}$, $\ket{\psi}^{\otimes k}$ lives in the space of symmetric $k$-fold tensor power $\mathrm{Sym}^k(\dsC^d) \subset (\dsC^d) ^{\otimes k}$.

\begin{lemma}[Homogeneous polynomials and symmetric tensors]
Let $f$ be a homogeneous polynomial of degree $t$ of the elements of $\ket{\psi}\in\mathbb C^d$.
Then there exists a vector $\lvert v_f\rangle \in \mathrm{Sym}^t(\mathbb C^d)$ such that
\begin{equation}
f(\ket{\psi}) = \braket{v_f}{\psi}^{\otimes t}.
\label{eq:app-polynomial-tensor-representation}
\end{equation}
Moreover, for any ensemble $\mathcal E$,
\begin{equation}
\dsE_{\ket{\psi}\sim\mathcal E}\bigl[\lvert f(\ket{\psi})\rvert^2\bigr]
=
\mathrm{Tr}\!\left( \lvert v_f\rangle\langle v_f\rvert \, M_t(\mathcal E)\right).
\label{eq:app-polynomial-moment-identity}
\end{equation}
Hence $M_t(\mathcal E_1)=M_t(\mathcal E_2)$ implies
$\dsE_{\mathcal E_1}[\lvert f\rvert^2]=\dsE_{\mathcal E_2}[\lvert f\rvert^2]$
for all homogeneous degree-$t$ polynomials $f$. \label{lemma: homogeneous and sym}
\end{lemma}

\begin{proof}
For the first claim, fix an orthonormal basis $\{\ket{1},\dots,\ket{d}\}$ of $\mathbb C^d$ and write
\[
\ket{\psi}=\sum_{r=1}^d \psi_r \ket{r},\qquad \psi_r\in\mathbb C.
\]
Then
\[
\begin{aligned}
\ket{\psi}^{\otimes t}
&=
\sum_{i_1,\dots,i_t=1}^d \psi_{i_1}\cdots \psi_{i_t}\,
\ket{i_1}\otimes\cdots\otimes \ket{i_t}\\
&=
\sum_{i_1,\dots,i_t=1}^d \psi_{i_1}\cdots \psi_{i_t}\,\ket{i_1,\dots,i_t},
\end{aligned}
\]
where $\ket{i_1,\dots,i_t}:=\ket{i_1}\otimes\cdots\otimes\ket{i_t}$.

Let $f$ be a homogeneous polynomial of degree $t$ in the coordinates $(\psi_1,\dots,\psi_d)$.
Then $f$ can be written (after reindexing monomials) as
\[
f(\ket{\psi})=\sum_{i_1,\dots,i_t=1}^d c_{i_1\cdots i_t}\,\psi_{i_1}\cdots \psi_{i_t}
\]
for some coefficients $c_{i_1\cdots i_t}\in\mathbb C$.
Define a vector $\ket{v}\in(\mathbb C^d)^{\otimes t}$ by
\[
\ket{v}:=\sum_{i_1,\dots,i_t=1}^d c_{i_1\cdots i_t}^*\ket{i_1,\dots,i_t}.
\]
We have
$\bra{v} i_1,\dots,i_t\rangle=c_{i_1\cdots i_t}$, and hence
\[
\braket{v}{\psi}^{\otimes t}
=
\sum_{i_1,\dots,i_t=1}^d c_{i_1\cdots i_t}\,\psi_{i_1}\cdots \psi_{i_t}
=
f(\ket{\psi}).
\]
Thus $f(\ket{\psi})$ can be written as a linear functional of $\ket{\psi}^{\otimes t}$.

It remains to ensure that the representing vector may be chosen symmetric. For each permutation
$\pi\in S_t$, let $U_\pi$ be the unitary that permutes tensor factors:
\[
U_\pi\big(\ket{x_1}\otimes\cdots\otimes\ket{x_t}\big)
=
\ket{x_{\pi^{-1}(1)}}\otimes\cdots\otimes\ket{x_{\pi^{-1}(t)}}.
\]
Define the symmetrizer
\[
\Pi_{\mathrm{sym}}:=\frac{1}{t!}\sum_{\pi\in S_t} U_\pi,
\]
whose range is $\mathrm{Sym}^t(\mathbb C^d)$. Since $\ket{\psi}^{\otimes t}$ is invariant under
permutations of the $t$ identical factors, we have $U_\pi\ket{\psi}^{\otimes t}=\ket{\psi}^{\otimes t}$
for all $\pi$, and therefore
\[
\Pi_{\mathrm{sym}}\ket{\psi}^{\otimes t}=\ket{\psi}^{\otimes t}.
\]
Using that $\Pi_{\mathrm{sym}}$ is self-adjoint, we obtain
\[
f(\ket{\psi})
=
\braket{v}{\psi}^{\otimes t}
=
\bra{v}\Pi_{\mathrm{sym}}\ket{\psi}^{\otimes t}
=
\braket{\Pi_{\mathrm{sym}}v}{\psi}^{\otimes t}.
\]
Hence, setting $\ket{v_f}:=\Pi_{\mathrm{sym}}\ket{v}\in \mathrm{Sym}^t(\mathbb C^d)$ yields
\[
f(\ket{\psi})=\braket{v_f}{\psi}^{\otimes t},
\]
which proves the first claim.

For the second claim,
\[
\begin{aligned}
\lvert f(\ket{\psi})\rvert^2
&=
\abs{\braket{v_f}{\psi}^{\otimes t}}^2\\
&=
\bra{\psi}^{\otimes t} \lvert v_f\rangle\langle v_f\rvert \ \ket{\psi}^{\otimes t}\\
&=
\mathrm{Tr}\!\left(\lvert v_f\rangle\langle v_f\rvert \, \rho_{\ket{\psi}}^{\otimes t}\right).
\end{aligned}
\]
Taking expectation over $\ket{\psi}\sim\mathcal E$ yields
$\dsE[\lvert f(\ket{\psi})\rvert^2]=\mathrm{Tr}(\lvert v_f\rangle\langle v_f\rvert\,M_t(\mathcal E))$.
\end{proof}

\begin{lemma}[Hyperplane separation for distinct pure states; $d\ge 2$]
Let $\ket{x},\ket{y}\in\mathbb C^d$ be two distinct unit vectors (pure quantum states), i.e.,
$\ket{x}\neq e^{i\theta}\ket{y}$ for all $\theta\in\mathbb R$. Then there exists a vector $\ket{a}\in\mathbb C^d$
such that
\begin{equation}
\braket{a}{y}=0
\quad\text{and}\quad
\braket{a}{x}\neq 0.
\label{eq:app-hyperplane-separation}
\end{equation}
Equivalently, the linear functional $\ell(\ket{\psi}) := \braket{a}{\psi}$ vanishes at $\ket{y}$ but not at $\ket{x}$. \label{lemma: hyperplane}
\end{lemma}

\begin{proof}
Define the orthogonal complement of $\ket{y}$,
\[
Y^\perp := \bigl\{\,\ket{a}\in\mathbb C^d : \braket{a}{y}=0\,\bigr\}.
\]
Since $d\ge 2$ and $\ket{y}\neq 0$, we have $\dim(Y^\perp)=d-1\ge 1$, so $Y^\perp$ contains nonzero vectors.

Consider the additional constraint $\braket{a}{x}=0$. The set
\[
Z := \bigl\{\,\ket{a}\in Y^\perp : \braket{a}{x}=0\,\bigr\}
\]
is a linear subspace of $Y^\perp$. We claim that $Z$ is a \emph{proper} subspace of $Y^\perp$.

Indeed, if $Z=Y^\perp$, then $\braket{a}{x}=0$ for all $\ket{a}\in Y^\perp$, which implies
$\ket{x}\in (Y^\perp)^\perp = \mathrm{span}\{\ket{y}\}$.
Thus $\ket{x}=c\ket{y}$ for some $c\in\mathbb C\setminus\{0\}$, and because both $\ket{x}$ and $\ket{y}$ are unit,
we must have $c=e^{i\theta}$ for some $\theta\in\mathbb R$, contradicting the assumption that
$\ket{x}\neq e^{i\theta}\ket{y}$.

Therefore $Z\subsetneq Y^\perp$, so there exists $\ket{a}\in Y^\perp\setminus Z$.
By construction, $\braket{a}{y}=0$ and $\braket{a}{x}\neq 0$, as desired.
\end{proof}

\begin{theorem}[Sufficiency: $k=N$ gives full discriminative power for ensembles with at most $N$ pure states]
Fix $d\ge 2$ and $N\in\mathbb N$.
Let
\[
\mathcal E_1=\{(p_i,\ket{\psi_i})\}_{i=1}^{N_1},
\qquad
\mathcal E_2=\{(q_j,\ket{\phi_j})\}_{j=1}^{N_2}
\]
be two quantum ensembles in $\mathbb C^d$, where $\ket{\psi_i},\ket{\phi_j}$ are unit vectors, all states in each
ensemble are distinct, and $(p_i)_{i=1}^{N_1}$, $(q_j)_{j=1}^{N_2}$ are probability distributions (in particular,
$p_i>0$, $q_j>0$, and $\sum_i p_i=\sum_j q_j=1$). Assume $N_1,N_2\le N$.
If
\[
M_N(\mathcal E_1)=M_N(\mathcal E_2),
\]
then $\mathcal E_1=\mathcal E_2$ up to a permutation of indices (i.e., the two ensembles have the same states with
the same weights).\label{thm: k=N full discriminative power}
\end{theorem}

\begin{proof}
Assume $M_N(\mathcal E_1)=M_N(\mathcal E_2)$. By Lemma~\ref{thm: hierarachy of MMD-k},we have $M_t(\mathcal E_1)=M_t(\mathcal E_2)$ for all integers $1\le t\le N$.

\medskip
\noindent\textit{Step 1: the sets of states coincide.}
Suppose, for contradiction, that there exists a state $\ket{\psi_\star}$ among $\{\ket{\psi_i}\}_{i=1}^{N_1}$
such that $\ket{\psi_\star}\neq e^{i\theta}\ket{\phi_j}$ for every $j\in\{1,\dots,N_2\}$ and every $\theta\in\mathbb R$.

For each $j\in\{1,\dots,N_2\}$, apply Lemma~\ref{lemma: hyperplane}
to the pair $(\ket{x},\ket{y})=(\ket{\psi_\star},\ket{\phi_j})$. This gives a vector $\ket{a_j}\in\mathbb C^d$ such that
\[
\braket{a_j}{\phi_j}=0
\quad\text{and}\quad
\braket{a_j}{\psi_\star}\neq 0.
\]
Define the linear functional $\ell_j(\ket{\psi}) := \braket{a_j}{\psi}$ and the function
\begin{equation}
f(\ket{\psi}) := \prod_{j=1}^{N_2} \ell_j(\ket{\psi})
= \prod_{j=1}^{N_2} \braket{a_j}{\psi}.
\label{eq:app-support-separating-polynomial}
\end{equation}
Each factor $\ell_j(\ket{\psi})$ is a homogeneous polynomial of degree $1$ in the coordinates of $\ket{\psi}$,
so $f$ is a homogeneous polynomial of degree $N_2$.

Now observe:
\begin{itemize}
\item For each $j$, since $\ell_j(\ket{\phi_j})=\braket{a_j}{\phi_j}=0$, we have $f(\ket{\phi_j})=0$.
Therefore,
\[ 
\dsE_{\ket{\phi}\sim\mathcal E_2}\!\big[\Abs{f(\ket{\phi})}^2\big]
= \sum_{j=1}^{N_2} q_j\,\Abs{f(\ket{\phi_j})}^2
=0.
\]
\item On the other hand, $\ell_j(\ket{\psi_\star})=\braket{a_j}{\psi_\star}\neq 0$ for every $j$, hence
$f(\ket{\psi_\star})\neq 0$. Since $p_\star>0$ is the weight of $\ket{\psi_\star}$ in $\mathcal E_1$,
\[
\begin{aligned}
\dsE_{\ket{\psi}\sim\mathcal E_1}\!\big[\Abs{f(\ket{\psi})}^2\big]
&= \sum_{i=1}^{N_1} p_i\,\Abs{f(\ket{\psi_i})}^2\\
&\ge p_\star\,\Abs{f(\ket{\psi_\star})}^2
>0.
\end{aligned}
\]
\end{itemize}

By Lemma~\ref{lemma: homogeneous and sym}, there exists $\ket{v_f}\in\mathrm{Sym}^{N_2}(\mathbb C^d)$
such that
\[
\begin{gathered}
\dsE_{\ket{\psi}\sim\mathcal E}\!\big[\Abs{f(\ket{\psi})}^2\big]
=
\mathrm{Tr}\!\Big(\ketbra{v_f}{v_f}\,M_{N_2}(\mathcal E)\Big)
\\
\text{for any ensemble }\mathcal E.
\end{gathered}
\]
Hence the strict inequality
\[
\dsE_{\mathcal E_1}\!\big[\Abs{f}^2\big] \neq \dsE_{\mathcal E_2}\!\big[\Abs{f}^2\big]
\]
implies
\[
\mathrm{Tr}\!\Big(\ketbra{v_f}{v_f}\,M_{N_2}(\mathcal E_1)\Big)
\neq
\mathrm{Tr}\!\Big(\ketbra{v_f}{v_f}\,M_{N_2}(\mathcal E_2)\Big),
\]
and therefore $M_{N_2}(\mathcal E_1)\neq M_{N_2}(\mathcal E_2)$.
But $N_2\le N$ and we already have $M_t(\mathcal E_1)=M_t(\mathcal E_2)$ for all $t\le N$, a contradiction.
Thus, every state in $\mathcal E_1$ must coincide with some state in $\mathcal E_2$ up to global phase.

By symmetry (swapping the roles of $\mathcal E_1$ and $\mathcal E_2$), every state in $\mathcal E_2$ must also coincide
with some state in $\mathcal E_1$ up to global phase. Therefore the two ensembles contain the same set of states (up to
a permutation of indices and global phases), and in particular $N_1=N_2$.

\medskip
\noindent\textit{Step 2: the weights coincide.}
After relabeling, assume $\ket{\psi_i}=e^{i\theta_i}\ket{\phi_i}$ for all $i\in\{1,\dots,N_1\}$.
Fix an index $i$. For each $j\neq i$, apply Lemma~\ref{lemma: hyperplane}
to the pair $(\ket{x},\ket{y})=(\ket{\psi_i},\ket{\psi_j})$ to obtain $\ket{a_{ij}}$ such that
\[
\braket{a_{ij}}{\psi_j}=0
\quad\text{and}\quad
\braket{a_{ij}}{\psi_i}\neq 0.
\]
Define $\ell_{ij}(\ket{\psi}) := \braket{a_{ij}}{\psi}$ and
\begin{equation}
f_i(\ket{\psi})
:=\prod_{j\neq i}\frac{\ell_{ij}(\ket{\psi})}{\ell_{ij}(\ket{\psi_i})}
=\prod_{j\neq i}\frac{\braket{a_{ij}}{\psi}}{\braket{a_{ij}}{\psi_i}}.
\label{eq:app-weight-interpolation-polynomial}
\end{equation}
Then $f_i(\ket{\psi_i})=1$ and $f_i(\ket{\psi_j})=0$ for all $j\neq i$.
Consequently,
\[
\dsE_{\ket{\psi}\sim\mathcal E_1}\!\big[\Abs{f_i(\ket{\psi})}^2\big]
= \sum_{r=1}^{N_1} p_r\,\Abs{f_i(\ket{\psi_r})}^2
= p_i,
\]
and similarly,
\[
\dsE_{\ket{\phi}\sim\mathcal E_2}\!\big[\Abs{f_i(\ket{\phi})}^2\big]
= q_i,
\]
since $\Abs{f_i(e^{i\theta}\ket{\psi})}=\Abs{f_i(\ket{\psi})}$ and the two ensembles share the same states up to phases.

Because $f_i$ is a homogeneous polynomial of degree $N_1-1$ divided by a nonzero constant,
it is (as a function of coordinates) homogeneous of degree $N_1-1\le N-1$.
Applying Lemma~\ref{lemma: homogeneous and sym} to $f_i$ gives
\[
\dsE_{\mathcal E}\!\big[\Abs{f_i}^2\big]
=
\mathrm{Tr}\!\Big(\ketbra{v_{f_i}}{v_{f_i}}\,M_{N_1-1}(\mathcal E)\Big).
\]
Since $N_1-1\le N$ and $M_{N_1-1}(\mathcal E_1)=M_{N_1-1}(\mathcal E_2)$, we conclude
\begin{equation}
p_i=\dsE_{\mathcal E_1}\!\big[\Abs{f_i}^2\big]
=\dsE_{\mathcal E_2}\!\big[\Abs{f_i}^2\big]
=q_i.
\label{eq:app-ensemble-weight-equality}
\end{equation}
As $i$ was arbitrary, $p_i=q_i$ for all $i$, proving that the two ensembles are identical up to permutation.
\end{proof}
To make the proof of Lemma~\ref{thm: k to reach full disriminative power} complete, we now construct a worst case, which shows two ensembles that can not discriminated by MMD-$k$ with $k<N$.
\begin{lemma}[Hard pair for $k<N$]\label{lem:hardpair}
Fix $d\ge 2$ and $N\ge 2$. Choose orthonormal $\ket{0},\ket{1}\in\mathbb C^d$. For $\theta\in\mathbb R$ and
$\ell\in\{0,\dots,N-1\}$ define
\begin{equation}
\begin{gathered}
\ket{\psi_\ell^{(\theta)}}:=\frac{1}{\sqrt2}\Big(\ket{0}+e^{i(\theta+2\pi\ell/N)}\ket{1}\Big),\\
\mathcal E_\theta:=\Big\{\big(1/N,\ket{\psi_\ell^{(\theta)}}\big)\Big\}_{\ell=0}^{N-1}.
\end{gathered}
\label{eq:app-phase-polygon-hard-pair}
\end{equation}
Then for every $t<N$, the moment operator $M_t(\mathcal E_\theta)$ is independent of $\theta$. Hence for any $k<N$,
\begin{equation}
M_k(\mathcal E_0)=M_k(\mathcal E_{\pi/N})\quad\text{but}\quad \mathcal E_0\neq \mathcal E_{\pi/N}.
\label{eq:app-hard-pair-moment-matching}
\end{equation}
Moreover $M_N(\mathcal E_0)\neq M_N(\mathcal E_{\pi/N})$.
\end{lemma}

\begin{proof}
Let $\varphi_\ell:=\theta+2\pi\ell/N$ and $\rho(\varphi_\ell):=\ketbra{\psi_\ell^{(\theta)}}{\psi_\ell^{(\theta)}}$.
In the $\{\ket{0},\ket{1}\}$ subspace,
\[
\rho(\varphi)=\tfrac12\big(\ketbra{0}{0}+\ketbra{1}{1}+e^{-i\varphi}\ketbra{0}{1}+e^{i\varphi}\ketbra{1}{0}\big),
\]
so each matrix element of $\rho(\varphi)^{\otimes t}$ is a linear combination of $e^{im\varphi}$ with $|m|\le t$.
Averaging over $\ell$ gives the discrete Fourier sum
\begin{equation}
\begin{aligned}
\frac1N\sum_{\ell=0}^{N-1}e^{im\varphi_\ell}
&=e^{im\theta}\cdot \frac1N\sum_{\ell=0}^{N-1}e^{i2\pi m\ell/N}\\
&=
\begin{cases}
e^{im\theta}, & \frac{m}{N}\in\dsZ,\\
0, & \text{otherwise}.
\end{cases}
\end{aligned}
\label{eq:app-discrete-phase-fourier-sum}
\end{equation}
If $t<N$, the only multiple of $N$ in $\{-t,\dots,t\}$ is $m=0$, so $\frac{m}{N}\in\dsZ$ only when $m=0$, hence all $\theta$-dependent terms vanish and
$M_t(\mathcal E_\theta)=\frac1N\sum_{\ell}\rho(\varphi_\ell)^{\otimes t}$ is independent of $\theta$.

The ensembles $\mathcal E_0$ and $\mathcal E_{\pi/N}$ have different phase sets
$\{2\pi\ell/N\}_\ell$ and $\{\pi/N+2\pi\ell/N\}_\ell$, hence $\mathcal E_0\neq \mathcal E_{\pi/N}$.
Finally, with $A:=\ketbra{0^{\otimes N}}{1^{\otimes N}}$,
\[
\begin{aligned}
\Tr\!\big(A\,\rho(\varphi)^{\otimes N}\big)&=\big(\bra{1}\rho(\varphi)\ket{0}\big)^N\\
&=(2^{-1}e^{i\varphi})^N=2^{-N}e^{iN\varphi},
\end{aligned}
\]
so
\begin{equation}
\Tr\!\big(A\,M_N(\mathcal E_\theta)\big)=\frac1N\sum_{\ell=0}^{N-1}2^{-N}e^{iN(\theta+2\pi\ell/N)}=2^{-N}e^{iN\theta}.
\label{eq:app-phase-polygon-nth-moment}
\end{equation}
Thus $\Tr(A\,M_N(\mathcal E_0))=2^{-N}$ while $\Tr(A\,M_N(\mathcal E_{\pi/N}))=-2^{-N}$, so $M_N$ differs.
\end{proof}
Combining Lemma~\ref{thm: hierarachy of MMD-k}, Theorem~\ref{thm: k=N full discriminative power}, and Lemma~\ref{lem:hardpair}, we conclude that the threshold to achieve full discrimiative power for MMD-$k$ is $k=N$, completing the proof of Lemma~\ref{thm: k to reach full disriminative power}.

\section{General sample complexity: Proofs for Section~\ref{sec: sample complexity}
}

\subsection{Proof of Theorem~\ref{thm:mmdk-lower-fullygeneral}: general lower bound for MMD-$k$}
\label{app: proof lower general}

\begin{proof}
We prove the lower bound by showing that an accurate MMD-$k$ estimator would solve a hard binary testing problem. We construct two probability distributions over ensemble pairs such that their MMD-$k$ values are separated by a constant with high probability, but their sampled quantum inputs have identical prior averages unless some label of the second ensemble appears at least $k$ times. Consequently, any measurement can distinguish the two hypotheses only through such repeated-label events. Bounding their probability then gives the required sample complexity. The random choice of an ensemble is only a proof device: each ensemble is drawn once and subsequently remains fixed, and the assumed estimation guarantee must hold for every possible realized pair.

\emph{Construction of the two hypotheses.} Fix an integer $k\ge2$. It suffices to work in $\mathbb C^2$, since the construction embeds into any $\mathbb C^d$ with $d\ge2$. Define the equatorial qubit states and small label-dependent offsets
\begin{equation}
\begin{gathered}
|\psi_\theta\rangle:=\frac{|0\rangle+e^{i\theta}|1\rangle}{\sqrt2},\\ \rho_\theta:=|\psi_\theta\rangle\langle\psi_\theta|,\\ \theta_j:=\frac{\pi j}{3k(N+1)},\quad j\in[N].
\end{gathered}
\label{eq:app-mmd-lower-equatorial-states}
\end{equation}
The first ensemble is fixed:
\[
\mathcal E_1:=\left\{\left(\frac1N,\rho_{\theta_i}\right)\right\}_{i=1}^N.
\]
Under hypothesis $H_b$, where $b\in\{0,1\}$, draw $\ell_1,\ldots,\ell_N$ independently and uniformly from $\{0,\ldots,k-1\}$, and define
\begin{equation}
\sigma_{b,j}:=\rho_{\theta_j+(2\ell_j+b)\pi/k},\qquad \mathcal E_{2,b}:=\left\{\left(\frac1N,\sigma_{b,j}\right)\right\}_{j=1}^N.
\label{eq:app-mmd-lower-prior-ensembles}
\end{equation}
Thus, for each label $j$, the two hypotheses use two regular $k$-gons on the Bloch equator, rotated relative to one another by $\pi/k$. The hidden choices $\ell_j$ are made once, not resampled whenever label $j$ appears. Both priors have finite support. Moreover, the components within each realized ensemble are pairwise distinct: equality of two comparison states would require $\theta_j-\theta_{j'}$ to be an integer multiple of $2\pi/k$, whereas $0<|\theta_j-\theta_{j'}|<\pi/(3k)$ for $j\ne j'$. The fixed ensemble also has distinct components. The construction and $\mathcal E_1$ may even be disclosed to the estimator; only the realized hidden choices remain unknown.

\emph{Why fewer than $k$ copies of one component cannot distinguish the hypotheses.} Let $\mathbb E_b$ denote averaging over the hidden choices under $H_b$, and define
\[
\Gamma_{b,j}^{(s)}:=\mathbb E_b[\sigma_{b,j}^{\otimes s}]=\frac1k\sum_{\ell=0}^{k-1}\rho_{\theta_j+(2\ell+b)\pi/k}^{\otimes s},\qquad \Gamma_{b,j}^{(0)}:=1.
\]
For a bit string $x\in\{0,1\}^s$, let $|x|$ be its number of ones. Expanding the tensor product gives
\[
|\psi_\theta\rangle^{\otimes s}=2^{-s/2}\sum_{x\in\{0,1\}^s}e^{i|x|\theta}|x\rangle.
\]
Hence, writing $m:=|x|-|y|$,
\[
\langle x|\Gamma_{b,j}^{(s)}|y\rangle=2^{-s}e^{im(\theta_j+b\pi/k)}\frac1k\sum_{\ell=0}^{k-1}e^{2\pi i m\ell/k}.
\]
The normalized finite sum is one when $m$ is a multiple of $k$ and zero otherwise. For $s<k$, we have $|m|\le s<k$, so only $m=0$ survives. Therefore,
\begin{equation}
\Gamma_{0,j}^{(s)}=\Gamma_{1,j}^{(s)}\qquad(0\le s<k).
\label{eq:app-mmd-lower-copy-matching}
\end{equation}
At $s=k$, the additional possibilities $m=\pm k$ correspond only to the matrix elements between the all-zero and all-one basis states. Consequently,
\[
\Gamma_{b,j}^{(k)}=G_k+(-1)^bA_j,
\]
where
\begin{widetext}
\[
G_k:=2^{-k}\sum_{\substack{x,y\in\{0,1\}^k\\|x|=|y|}}|x\rangle\langle y|,\qquad A_j:=2^{-k}\left(e^{ik\theta_j}|1^k\rangle\langle0^k|+e^{-ik\theta_j}|0^k\rangle\langle1^k|\right).
\]
\end{widetext}
Thus the first difference between the two hypotheses appears in a $k$-copy coherence, while all lower-copy averaged states agree.

\emph{Separation of the MMD-$k$ values.} Write
\[
\begin{gathered}
R:=M_k(\mathcal E_1),\\ \widehat M_b:=M_k(\mathcal E_{2,b}),\\ \overline A_N:=\frac1N\sum_{j=1}^NA_j,\\ \overline M_b:=\mathbb E_b\widehat M_b=G_k+(-1)^b\overline A_N.
\end{gathered}
\]
Here $\widehat M_b$ is the moment of the randomly drawn but subsequently fixed ensemble, not an estimate obtained from measurements. By the moment representation of MMD-$k$, the actual distance and a deterministic reference value are
\begin{equation}
\begin{gathered}
D_b:=D^{(k)}(\mathcal E_1,\mathcal E_{2,b})=\|R-\widehat M_b\|_{\mathrm{HS}}^2,\\ \overline D_b:=\|R-\overline M_b\|_{\mathrm{HS}}^2.
\end{gathered}
\label{eq:app-mmd-lower-reference-distances}
\end{equation}
We first calculate the reference-value gap and then control the fluctuations of $D_b$. Set $a_N:=N^{-1}\sum_j e^{ik\theta_j}$. Both $R-G_k$ and $\overline A_N$ have matrix element $2^{-k}a_N$ between $|1^k\rangle$ and $|0^k\rangle$, together with its complex conjugate in the reverse direction. Since $\overline A_N$ has no other nonzero entries,
\[
\operatorname{Tr}[(R-G_k)\overline A_N]=2\,|2^{-k}a_N|^2=2^{1-2k}|a_N|^2.
\]
Expanding the two squared norms therefore yields
\begin{equation}
\begin{aligned}
\overline D_1-\overline D_0&=4\operatorname{Tr}[(R-G_k)\overline A_N]\\
&=2^{3-2k}|a_N|^2\ge2^{1-2k}=:g_k>0.
\end{aligned}
\label{eq:app-mmd-lower-reference-gap}
\end{equation}
The last inequality follows from $0<k\theta_j<\pi/3$, which gives $\operatorname{Re}a_N=N^{-1}\sum_j\cos(k\theta_j)\ge1/2$. Thus $g_k$ depends only on $k$, not on $N$.

To compare the actual distances with these reference values, independence across labels gives
\[
\mathbb E_b\|\widehat M_b-\overline M_b\|_{\mathrm{HS}}^2=\frac1{N^2}\sum_{j=1}^N\left(1-\operatorname{Tr}[(\Gamma_{b,j}^{(k)})^2]\right)\le\frac1N.
\]
Indeed, cross terms between different labels vanish after averaging because the centered random matrices are independent, and each $\sigma_{b,j}^{\otimes k}$ is pure, with squared Hilbert--Schmidt norm one. All moment operators are density operators and hence have Hilbert--Schmidt norm at most one. Applying Cauchy--Schwarz to the difference of the two squared distances gives
\begin{widetext}
\[
|D_b-\overline D_b|=\left|\operatorname{Tr}[(\widehat M_b-\overline M_b)(\widehat M_b+\overline M_b-2R)]\right|\le4\|\widehat M_b-\overline M_b\|_{\mathrm{HS}}.
\]
\end{widetext}
Markov's inequality now implies
\begin{equation}
\begin{aligned}
&\Pr_b\!\\
&\left(|D_b-\overline D_b|>g_k/8\right)\\
&\le\Pr_b\!\\
&\left(\|\widehat M_b-\overline M_b\|_{\mathrm{HS}}>g_k/32\right)\\
&\le\frac{1024}{Ng_k^2}=:\zeta_N.
\end{aligned}
\label{eq:app-mmd-lower-distance-concentration}
\end{equation}
Since $k$ is fixed, $\zeta_N\to0$. Thus, with high probability, the two hypotheses produce distance values close to two reference values separated by at least $g_k$.

\emph{Indistinguishability under arbitrary measurements.} Consider any protocol using at most $M$ samples from each ensemble. This includes protocols that adaptively choose the sampling order or stop early: each can be simulated using two presampled streams of length $M$, with unused samples discarded. We therefore analyze these complete streams. Write their numerical labels as $L=((I_t,J_t))_{t=1}^M$, and let
\[
B_j(L):=|\{t:J_t=j\}|.
\]
Because all ensemble weights are $1/N$, the labels are independent and uniform, and their common law under both hypotheses is $q(L)=N^{-2M}$. Conditional on $L$, averaging the full quantum input over the hidden choices gives
\begin{equation}
\rho_L^{(b)}=U_L\left[\left(\bigotimes_{t=1}^M\rho_{\theta_{I_t}}\right)\otimes\left(\bigotimes_{j=1}^N\Gamma_{b,j}^{(B_j(L))}\right)\right]U_L^\dagger,
\label{eq:app-mmd-lower-conditional-input}
\end{equation}
where $U_L$ only restores the original register order. This factorization follows because every occurrence of label $j$ carries the same fixed state $\sigma_{b,j}$, while the hidden states associated with different labels are chosen independently. In particular, the factor for label $j$ is $\mathbb E_b[\sigma_{b,j}^{\otimes B_j(L)}]$, not $(\mathbb E_b\sigma_{b,j})^{\otimes B_j(L)}$.

Let $\mathcal C_k:=\{L:\max_jB_j(L)\ge k\}$ be the event that some comparison label appears at least $k$ times. The lower-copy identities above imply $\rho_L^{(0)}=\rho_L^{(1)}$ whenever $L\notin\mathcal C_k$. For any Borel set $S\subseteq\mathbb R$, the estimator's output event $\{\widehat D\in S\}$ is described by the label-dependent POVM effect $E_S^L$, so the Born rule and linearity give
\[
\Pr_b(\widehat D\in S)=\sum_Lq(L)\operatorname{Tr}(E_S^L\rho_L^{(b)}).
\]
The summands agree outside $\mathcal C_k$, and each conditional event probability lies between zero and one. Therefore,
\begin{equation}
\left|\Pr_1(\widehat D\in S)-\Pr_0(\widehat D\in S)\right|\le\Pr_L(\mathcal C_k).
\label{eq:app-mmd-lower-event-distinguishability}
\end{equation}
This argument imposes no restriction on the POVM, and therefore includes arbitrary ancillas, quantum memory, joint measurements, and adaptive processing. For any fixed $k$ sampling positions, the probability that their comparison labels coincide is $N(1/N)^k=N^{1-k}$. Taking a union bound over the $\binom Mk$ choices of positions yields
\begin{equation}
\Pr_L(\mathcal C_k)\le\binom MkN^{1-k}\le\frac{M^k}{k!N^{k-1}}.
\label{eq:app-mmd-lower-collision-bound}
\end{equation}

\emph{An accurate estimator would distinguish the hypotheses.} Set $\epsilon_0:=g_k/8$ and $\delta_0:=1/16$, and take $N$ sufficiently large that $\zeta_N\le1/16$. Suppose an estimator is $(\epsilon_0,\delta_0)$-accurate for every deterministic ensemble pair in the supports of the two priors. Turn it into a binary test by declaring $H_1$ when $\widehat D>\tau$ and $H_0$ otherwise, where
\[
\tau:=\frac{\overline D_0+\overline D_1}{2}.
\]
Whenever the estimator succeeds and $|D_b-\overline D_b|\le g_k/8$, we have $|\widehat D-\overline D_b|\le g_k/4$, which is smaller than the distance from $\overline D_b$ to the threshold. Hence the test is correct. Averaging the pointwise estimation guarantee over the original priors and taking a union bound, the two testing errors satisfy
\[
\begin{gathered}
\alpha:=\Pr_0(\widehat D>\tau)\le\delta_0+\zeta_N\le\frac18,\\ \beta:=\Pr_1(\widehat D\le\tau)\le\delta_0+\zeta_N\le\frac18.
\end{gathered}
\]
No conditioning of the hard priors is used here, so their label-wise independence is preserved. Applying the preceding event-probability bound with $S=(\tau,\infty)$ gives
\[
\frac34\le1-\alpha-\beta=\Pr_1(\widehat D>\tau)-\Pr_0(\widehat D>\tau)\le\frac{M^k}{k!N^{k-1}}.
\]
Consequently,
\begin{equation}
M\ge c_kN^{1-1/k},\qquad c_k:=\left(\frac{3k!}{4}\right)^{1/k}>0.
\label{eq:app-mmd-general-sample-lower-bound}
\end{equation}
The constants $\epsilon_0,\delta_0,c_k$ are independent of $N$, and the hard family consists of finite, valid uniform pure-state ensemble pairs. This proves the theorem.
\end{proof}

\subsection{Proof of Theorem~\ref{thm:mmdk-upper}: across-label upper bound for MMD-$k$}\label{app:mmdk-upper-bound}

We prove the upper bound for the cross term $\bar F^{(k)}(\scE_1,\scE_2)$; the two self terms are estimated in the same way using samples from $(\scE_1,\scE_1)$ and $(\scE_2,\scE_2)$. Write
$$
X_{ij}:=\Abs{\braket{\psi_i}{\phi_j}}^2,\qquad
\bar F^{(k)}(\scE_1,\scE_2)=\frac1{N^2}\sum_{i,j=1}^N X_{ij}^k.
$$
Let
$$
s=\left\lceil c\,\epsilon^{-2}\log\frac1\delta\right\rceil
$$
with $c$ sufficiently large. The proof uses two elementary occupancy estimates for uniform multinomial sampling, stated explicitly below. The sparse estimate tracks the number of row or column labels that appear at least $k$ times, and the dense estimate ensures that every row or column label appears at least $k$ times. The key difference from the label-local SWAP-test proof in Appendix~\ref{appsub: MMD-k upper bound} is that the occupancy bins are now $N$ row or column labels, rather than $N^2$ pair labels. The measurement step is also different: after forming row and column $k$-collision blocks, we perform only disjoint $k$-copy SWAP measurements.

We use the following standard concentration fact. Let $Y=(Y_{ij})_{i,j=1}^N$ be any matrix with entries in $[0,1]$. Suppose $i_1,\ldots,i_s$ and $j_1,\ldots,j_s$ are selected independently and uniformly without replacement from balanced populations containing the same number of copies of every label in $[N]$. Then
$$
\dsE\left[\frac1s\sum_{r=1}^s Y_{i_rj_r}\right]=\frac1{N^2}\sum_{i,j=1}^N Y_{ij},
$$
and, by the bounded-difference inequality for sampling without replacement,
\begin{equation}
\Pr\!\left[
\Abs{\frac1s\sum_{r=1}^s Y_{i_rj_r}-\frac1{N^2}\sum_{i,j=1}^N Y_{ij}}>\eta
\right]\le 2\exp(-c_0s\eta^2)
\label{eq:app-balanced-label-concentration}
\end{equation}
for a universal constant $c_0>0$.

We also use the following two elementary occupancy estimates. Let
$(T_1,\ldots,T_N)\sim\mathrm{Multinomial}(M;1/N,\ldots,1/N)$ be the occupancies obtained by throwing $M$ balls uniformly into $N$ bins.

First, fix $k\ge2$ and assume $M/N=O_k(1)$. Define
\[
    G_M:=\sum_{i=1}^N\mathbf 1\{T_i\ge k\}.
\]
Then by Eq. (\ref{eq: Em M<=cN2})
\begin{equation}\label{eq:row-column-sparse-occupancy}
    \dsE[G_M]
    =
    \Theta_k\!\left(\frac{M^k}{k!\,N^{k-1}}\right).
\end{equation}
Indeed, for a fixed bin,
\[
    \Pr(T_i\ge k)
    \le
    \dsE\binom{T_i}{k}
    =
    \binom{M}{k}\frac1{N^k}
    \le
    C_k\frac{M^k}{k!\,N^k}.
\]
Conversely,
\[
\begin{aligned}
\Pr(T_i\ge k)
    &\ge
    \Pr(T_i=k)\\
    &=
    \binom{M}{k}\frac1{N^k}\left(1-\frac1N\right)^{M-k}\\
    &\ge
    c_k\frac{M^k}{k!\,N^k},
\end{aligned}
\]
because $M/N=O_k(1)$ and $k$ is fixed. Summing over $i$ proves
\eqref{eq:row-column-sparse-occupancy}. Moreover, the multinomial occupancy variables are negatively associated, and increasing functions on disjoint coordinates preserve negative association. Therefore the indicators $\mathbf 1\{T_i\ge k\}$ satisfy the usual Chernoff lower-tail bound. Hence, for a constant $c_k'>0$ depending only on $k$,
\begin{equation}\label{eq:row-column-sparse-concentration}
    \Pr\left(G_M\le \frac12\dsE[G_M]\right)
    \le
    \exp\!\left(-c_k'\dsE[G_M]\right).
\end{equation}

Second, for any $M$, each marginal satisfies
\[
    T_i\sim \mathrm{Bin}\left(M,\frac1N\right),
    \qquad
    \dsE[T_i]=\frac{M}{N}.
\]
By the Chernoff lower-tail bound,
\begin{equation}\label{eq:row-column-dense-occupancy}
    \Pr\left(\min_{i\in[N]}T_i<\frac{M}{2N}\right)
    \le
    N\exp\!\left(-\frac{M}{8N}\right).
\end{equation}

\paragraph*{Sparse row/column-collision regime.}
Assume first that $k!s\lesssim N$. Take
$$
M_0=\left\lceil C_k(k!s)^{1/k}N^{1-1/k}\right\rceil
$$
samples from $(\scE_1,\scE_2)$, with labels $(I_t,J_t)$ and registers $A_tB_t$, $t=1,\ldots,M_0$. Let
$$
\begin{gathered}
\scR:=\{i\in[N]:\#\{t:I_t=i\}\ge k\},\\
\scC:=\{j\in[N]:\#\{t:J_t=j\}\ge k\}.
\end{gathered}
$$
For each $i\in\scR$, choose one disjoint row block of $k$ indices from $\{t:I_t=i\}$; for each $j\in\scC$, choose one disjoint column block of $k$ indices from $\{t:J_t=j\}$. Since we use at most one block per label, all row blocks are mutually disjoint, and all column blocks are mutually disjoint.

Let $G_R=|\scR|$, $G_C=|\scC|$, and $G=\min\{G_R,G_C\}$. Applying \eqref{eq:row-column-sparse-occupancy} to the $N$ row bins and independently to the $N$ column bins gives
$$
\begin{gathered}
\dsE[G_R]
=
\Theta_k\!\left(\frac{M_0^k}{k!\,N^{k-1}}\right)
=
\Theta_k(s),
\\
\dsE[G_C]
=
\Theta_k\!\left(\frac{M_0^k}{k!\,N^{k-1}}\right)
=
\Theta_k(s).
\end{gathered}
$$
Here the sparse assumption $k!s\lesssim N$ ensures $M_0/N=O_k(1)$. Choosing $C_k$ sufficiently large gives
$$
\dsE[G_R]\ge 4s,
\qquad
\dsE[G_C]\ge 4s.
$$
By \eqref{eq:row-column-sparse-concentration},
$$
\Pr(G_R<s)\le \exp(-c_k's),
\qquad
\Pr(G_C<s)\le \exp(-c_k's).
$$
Since $s=\lceil c\epsilon^{-2}\log(1/\delta)\rceil$ and $c$ is chosen sufficiently large,
$$
\Pr[G\ge s]\ge 1-\delta/12.
$$

On the event $G\ge s$, choose $s$ labels uniformly without replacement from $\scR$ and $s$ labels uniformly without replacement from $\scC$, and pair the chosen row and column labels by an independent uniformly random permutation. By exchangeability of the multinomial occupancy vector, the selected row labels are a uniform $s$-subset of $[N]$, the selected column labels are an independent uniform $s$-subset of $[N]$.

Let the $r$-th selected row block have common row label $i_r$ and indices
$$
\mathsf R_r=\{a_{r,1},\ldots,a_{r,k}\},
$$
and let the paired column block have common column label $j_r$ and indices
$$
\mathsf C_r=\{b_{r,1},\ldots,b_{r,k}\}.
$$
On these $2k$ registers, measure
\begin{equation}
W_r=\bigotimes_{\ell=1}^k \mathrm{SWAP}(A_{a_{r,\ell}},B_{b_{r,\ell}})
\label{eq:app-across-label-swap-observable}
\end{equation}
with projectors $(I\pm W_r)/2$. Let $Z_r\in\{-1,+1\}$ be the outcome. Since all $A$-registers in $\mathsf R_r$ carry $\ket{\psi_{i_r}}$ and all $B$-registers in $\mathsf C_r$ carry $\ket{\phi_{j_r}}$,
$$
\dsE[Z_r\mid I,J]=X_{i_rj_r}^k.
$$
The measurements for different $r$ act on disjoint quantum registers. Hence they form one valid joint measurement, and conditioned on the labels the outcomes $Z_1,\ldots,Z_s$ are independent and bounded.

Define
\begin{equation}
\widehat F_{12}^{(k)}=\frac1s\sum_{r=1}^s Z_r
\label{eq:app-across-label-frame-estimator}
\end{equation}
on the event $G\ge s$, and define it arbitrarily on $G<s$. Set $\mu_r:=X_{i_rj_r}^k$. On $G\ge s$,
\begin{equation}
\begin{aligned}
&\widehat F_{12}^{(k)}-\bar F^{(k)}(\scE_1,\scE_2)\\
&=
\frac1s\sum_{r=1}^s(Z_r-\mu_r)\\
&\quad+
\left(
\frac1s\sum_{r=1}^s\mu_r-\bar F^{(k)}(\scE_1,\scE_2)
\right).
\end{aligned}
\label{eq:app-across-label-error-decomposition}
\end{equation}
The first term is measurement noise. Conditional on the labels, Hoeffding's inequality gives
$$
\Pr\!\left[
\Abs{\frac1s\sum_{r=1}^s(Z_r-\mu_r)}>\epsilon/8\,\middle|\,I,J
\right]\le \delta/12
$$
for the above choice of $s$. The second term is label-selection noise. Applying the sampling-without-replacement concentration stated above to $Y_{ij}=X_{ij}^k$ gives
$$
\Pr\!\left[
\Abs{\frac1s\sum_{r=1}^s\mu_r-\bar F^{(k)}(\scE_1,\scE_2)}>\epsilon/8
\right]\le \delta/12.
$$
Combining these two bounds with $\Pr[G<s]\le\delta/12$, and adjusting constants, yields
\begin{equation}
\Pr\!\left[
\Abs{\widehat F_{12}^{(k)}-\bar F^{(k)}(\scE_1,\scE_2)}>\epsilon/4
\right]\le \delta/3.
\label{eq:app-across-label-frame-accuracy}
\end{equation}
Thus, in the sparse regime, the cross term is estimated using
$$
M_0=\scO\!\left(\left(k!\epsilon^{-2}\log\frac1\delta\right)^{1/k}N^{1-1/k}\right)
$$
samples, where the hidden constant may depend on $k$.

\paragraph*{Dense full-coverage regime.}
We now consider the complementary regime. Take
$$
M_0=C_k\max\left\{ks,\;N\left(\log\frac{N}{\delta}+k\right)\right\}
$$
samples from $(\scE_1,\scE_2)$. Let
$$
A_i=\#\{t:I_t=i\},\qquad B_j=\#\{t:J_t=j\}.
$$
For each row label, $A_i\sim \mathrm{Bin}(M_0,1/N)$, and for each column label, $B_j\sim \mathrm{Bin}(M_0,1/N)$. Applying \eqref{eq:row-column-dense-occupancy} separately to the row labels and the column labels gives
$$
\Pr\left[
\min_i A_i<\frac{M_0}{2N}
\;\text{or}\;
\min_j B_j<\frac{M_0}{2N}
\right]
\le
2N\exp\!\left(-\frac{M_0}{8N}\right).
$$
Since $M_0\ge C_kN\log(N/\delta)$, choosing $C_k$ sufficiently large gives, with probability at least $1-\delta/12$,
$$
\min_i A_i\ge \frac{M_0}{2N},\qquad
\min_j B_j\ge \frac{M_0}{2N}.
$$
On this event, every row label and every column label has at least
$$
q:=\left\lfloor \frac{M_0}{2kN}\right\rfloor
$$
disjoint $k$-blocks. Since $M_0\ge C_kks$ and $M_0\ge C_kNk$, choosing $C_k$ large enough ensures $Nq\ge s$.

For each row label, form $q$ disjoint row $k$-blocks; for each column label, form $q$ disjoint column $k$-blocks. Thus the row-block multiset contains exactly $q$ copies of every row label, and the column-block multiset contains exactly $q$ copies of every column label. Choose $s$ row blocks uniformly without replacement and independently choose $s$ column blocks uniformly without replacement, then pair them by a uniformly random permutation.

For each paired block, perform the same $k$-copy SWAP measurement as in the sparse construction and denote the outcome by $Z_r$. Again
$$
\dsE[Z_r\mid i_r,j_r]=X_{i_rj_r}^k,
$$
and all measured observables use disjoint registers. The selected labels are sampled without replacement from balanced populations, so the sampling-without-replacement concentration above applies to $\mu_r=X_{i_rj_r}^k$. Hoeffding's inequality again controls the conditional SWAP measurement noise. Therefore
$$
\Pr\!\left[
\Abs{\frac1s\sum_{r=1}^s Z_r-\bar F^{(k)}(\scE_1,\scE_2)}>\epsilon/4
\right]\le \delta/3.
$$
The dense-regime cross-term cost is therefore
$$
M_0=\scO\!\left(
\max\left\{
k\epsilon^{-2}\log\frac1\delta,\;
N\left(\log\frac{N}{\delta}+k\right)
\right\}
\right).
$$

Combining the two constructions, the cross term can be estimated with failure probability at most $\delta/3$ using
\begin{widetext}
\begin{equation}
M_0=
\begin{cases}
\scO\!\left(\left(k!\epsilon^{-2}\log(1/\delta)\right)^{1/k}N^{1-1/k}\right),
& k!\epsilon^{-2}\log(1/\delta)\lesssim N,\\[1mm]
\scO\!\left(\max\left\{k\epsilon^{-2}\log(1/\delta),\,N\left(\log(N/\delta)+k\right)\right\}\right),
& k!\epsilon^{-2}\log(1/\delta)\gtrsim N.
\end{cases}
\label{eq:app-across-label-sample-regimes}
\end{equation}
\end{widetext}
The structural change relative to the label-local SWAP-test proof in Appendix~\ref{appsub: MMD-k upper bound} is that the occupancy bins are row or column labels, of which there are $N$, rather than full pair labels, of which there are $N^2$. The measurement step also differs: we measure only disjoint multi-copy SWAP observables, so no overlapping quantum observable is treated as a separately measured outcome.

Repeating the same construction for $(\scE_1,\scE_1)$ and $(\scE_2,\scE_2)$ gives estimators $\widehat F_{11}^{(k)}$ and $\widehat F_{22}^{(k)}$ satisfying the same bound. Set
\begin{equation}
\widehat D^{(k)}=\widehat F_{11}^{(k)}+\widehat F_{22}^{(k)}-2\widehat F_{12}^{(k)}.
\label{eq:app-across-label-mmd-estimator}
\end{equation}
Allocating the error and failure budgets among the three terms and applying a union bound gives
\begin{equation}
\Pr\!\left[
\Abs{\widehat D^{(k)}-D^{(k)}(\scE_1,\scE_2)}>\epsilon
\right]\le\delta.
\label{eq:app-across-label-mmd-accuracy}
\end{equation}
This proves Theorem~\ref{thm:mmdk-upper}.

\subsection{Proof of Theorem~\ref{thm:large-order-across-swap} for sample complexity of MMD-$k$ with large $k$ }
\label{app:large-order-across-swap-proof}

Here we prove Theorem~\ref{thm:large-order-across-swap} from Sec.~\ref{secsub: sc of MMDk general}.
\begin{proof}
Throughout this proof, \(k=\Theta(N)\). We first prove the upper bound. It suffices to
estimate one term \(\bar F^{(k)}(\scE_u,\scE_v)\), \(u,v\in\{1,2\}\), since the three
terms in \(\scD^{(k)}\) are combined by a union bound. Write
\[
    X_{ij}:=\Abs{\braket{\psi_i}{\phi_j}}^2,
    \qquad
    \bar F^{(k)}(\scE_u,\scE_v)
    =
    \frac1{N^2}\sum_{i,j=1}^N X_{ij}^k .
\]
Set
\[
    s:=\left\lceil c\,\epsilon^{-2}\log\frac{12}{\delta}\right\rceil .
\]
Since \(\epsilon,\delta\) are fixed, \(s=\scO(1)\).

Draw
\[
    M_0=C\,Nk
\]
paired samples from \((\scE_u,\scE_v)\), with labels \((I_t,J_t)\) and quantum registers
\(A_tB_t\), \(t=1,\ldots,M_0\). Let
\[
    A_i:=\#\{t:I_t=i\},
    \qquad
    B_j:=\#\{t:J_t=j\}.
\]
Then \(A_i,B_j\sim\operatorname{Bin}(M_0,1/N)\) and
\[
    \dsE A_i=\dsE B_j=Ck .
\]
By the Chernoff lower-tail bound, for \(C>1\),
\[
    \Pr(A_i<k)
    \le
    \exp(-c_C k),
    \qquad
    \Pr(B_j<k)
    \le
    \exp(-c_C k),
\]
where \(c_C>0\). Since \(k=\Theta(N)\), a union bound gives
\[
    \Pr\!\left[
        \min_i A_i<k
        \;\text{or}\;
        \min_j B_j<k
    \right]
    \le
    2N\exp(-c_C k)
    \le
    \frac{\delta}{12}
\]
for all sufficiently large \(N\), after choosing \(C\) large enough. Denote this good event by \(G\).

On \(G\), every row label and every column label contains at least one disjoint \(k\)-block.
Choose \(s\) row labels \(i_1,\ldots,i_s\) uniformly without replacement from \([N]\), and
independently choose \(s\) column labels \(j_1,\ldots,j_s\) uniformly without replacement
from \([N]\). For each \(r\), choose one row \(k\)-block
\[
    \mathsf R_r=\{a_{r,1},\ldots,a_{r,k}\}
\]
with common row label \(i_r\), and one column \(k\)-block
\[
    \mathsf C_r=\{b_{r,1},\ldots,b_{r,k}\}
\]
with common column label \(j_r\). The chosen row blocks are disjoint, and the chosen column
blocks are disjoint.

On the corresponding \(2k\) registers, measure
\begin{equation}
    W_r
    :=
    \bigotimes_{\ell=1}^k
    \mathrm{SWAP}\!\left(A_{a_{r,\ell}},B_{b_{r,\ell}}\right),
\label{eq:app-large-k-swap-observable}
\end{equation}
with projectors \((I\pm W_r)/2\), and let \(Z_r\in\{-1,+1\}\) be the outcome. Then
\[
    \dsE[Z_r\mid i_r,j_r]=X_{i_rj_r}^k .
\]
The measured register sets are disjoint, so conditioned on the selected labels,
\(Z_1,\ldots,Z_s\) are independent and bounded.

Define
\begin{equation}
    \widehat F_{uv}^{(k)}
    :=
    \frac1s\sum_{r=1}^s Z_r
\label{eq:app-large-k-estimator}
\end{equation}
on \(G\), and define it arbitrarily on \(G^c\). On \(G\),
\begin{equation}
\begin{aligned}
&\widehat F_{uv}^{(k)}-\bar F^{(k)}(\scE_u,\scE_v)\\
    &=
    \frac1s\sum_{r=1}^s
    \left(Z_r-X_{i_rj_r}^k\right)\\
    &\quad+
    \left[
        \frac1s\sum_{r=1}^s X_{i_rj_r}^k
        -
        \frac1{N^2}\sum_{i,j=1}^N X_{ij}^k
    \right].
\end{aligned}
\label{eq:app-large-k-error-decomposition}
\end{equation}
The first term is measurement noise, and Hoeffding's inequality gives
\[
    \Pr\!\left[
        \Abs{
            \frac1s\sum_{r=1}^s
            \left(Z_r-X_{i_rj_r}^k\right)
        }
        >
        \frac{\epsilon}{8}
    \right]
    \le
    \frac{\delta}{12}.
\]
The second term is label-selection noise. Since the row and column labels are sampled
uniformly without replacement and \(X_{ij}^k\in[0,1]\), Serfling's inequality gives
\[
    \Pr\!\left[
        \Abs{
            \frac1s\sum_{r=1}^s X_{i_rj_r}^k
            -
            \frac1{N^2}\sum_{i,j=1}^N X_{ij}^k
        }
        >
        \frac{\epsilon}{8}
    \right]
    \le
    \frac{\delta}{12},
\]
after increasing the universal constant \(c\) in the definition of \(s\). Together with
\(\Pr(G^c)\le\delta/12\), this yields
\[
    \Pr\!\left[
        \Abs{
            \widehat F_{uv}^{(k)}
            -
            \bar F^{(k)}(\scE_u,\scE_v)
        }
        >
        \frac{\epsilon}{4}
    \right]
    \le
    \frac{\delta}{3}.
\]
Applying the same construction to \((u,v)=(1,1),(2,2),(1,2)\), and setting
\[
    \widehat{\scD}^{(k)}
    :=
    \widehat F_{11}^{(k)}
    +
    \widehat F_{22}^{(k)}
    -
    2\widehat F_{12}^{(k)},
\]
a union bound gives
\begin{equation}
    \Pr\!\left[
        \Abs{
            \widehat{\scD}^{(k)}
            -
            \scD^{(k)}(\scE_1,\scE_2)
        }
        >
        \epsilon
    \right]
    \le
    \delta.
\label{eq:app-large-k-estimation-guarantee}
\end{equation}
Thus \(M_{\rm across}=\scO(Nk)\).

We now prove the protocol-specific lower bound. The lower bound applies to the class of
across-label block-SWAP estimators whose state-dependent information about
\(\bar F^{(k)}(\scE_u,\scE_v)\) comes only from disjoint row and column \(k\)-collision
blocks as above.

Choose distinct phases \(\alpha_1,\ldots,\alpha_N\), and set
\[
    r:=\frac1{10k}.
\]
Define two \(N\)-state uniform ensembles in \(\dsC^2\):
\begin{equation}
\begin{gathered}
\ket{u_i}
    :=
    \sqrt{1-r^2}\ket0
    +
    r e^{i\alpha_i}\ket1,
    \\
    \ket{v_i}
    :=
    \sqrt{1-r^2}\ket1
    +
    r e^{i\alpha_i}\ket0,
\end{gathered}
\label{eq:app-large-k-hard-states}
\end{equation}
and let
\[
    \scU:=\left\{\left(\frac1N,\ket{u_i}\right)\right\}_{i=1}^N,
    \qquad
    \scV:=\left\{\left(\frac1N,\ket{v_i}\right)\right\}_{i=1}^N .
\]
The ensembles \(\scU\) and \(\scV\) are unitary images of one another, so
\[
    \bar F^{(k)}(\scU,\scU)=\bar F^{(k)}(\scV,\scV).
\]
Moreover, for all \(i,j\),
\[
    \Abs{\braket{u_i}{u_j}}^2\ge 1-4r^2,
    \qquad
    \Abs{\braket{u_i}{v_j}}^2\le 4r^2.
\]
Hence
\[
    \bar F^{(k)}(\scU,\scU)
    =
    \bar F^{(k)}(\scV,\scV)
    \ge
    (1-4r^2)^k
    =
    1-o(1),
\]
while
\[
    \bar F^{(k)}(\scU,\scV)
    \le
    (4r^2)^k
    =
    o(1).
\]
Therefore, for all sufficiently large \(N\),
\begin{equation}
    \scD^{(k)}(\scU,\scU)=0,
    \qquad
    \scD^{(k)}(\scU,\scV)>1.
\label{eq:app-large-k-hard-gap}
\end{equation}

Consider the two hypotheses
\[
    H_0:\quad
    (\scE_1,\scE_2)=(\scU,\scU),
    \qquad
    H_1:\quad
    (\scE_1,\scE_2)=(\scU,\scV).
\]
Suppose the total sample budget satisfies \(M\le c_0Nk\), where \(c_0>0\) is a sufficiently
small absolute constant. Let \(M_{12}\le M\) be the number of paired samples used for the
cross term \(\bar F^{(k)}(\scE_1,\scE_2)\), and let \(A_i,B_j\) be the corresponding row
and column occupancies. For each row label,
\[
    A_i
    \preceq
    \operatorname{Bin}\!\left(M,\frac1N\right),
    \qquad
    \dsE A_i\le c_0k,
\]
and similarly for \(B_j\). By the Chernoff upper-tail bound, choosing \(c_0\) sufficiently
small gives
\[
    \Pr(A_i\ge k)\le \exp(-c_1k),
    \qquad
    \Pr(B_j\ge k)\le \exp(-c_1k)
\]
for some \(c_1>0\). Hence
\begin{equation}
    \Pr\!\left[
        \max_i A_i\ge k
        \;\text{or}\;
        \max_j B_j\ge k
    \right]
    \le
    2N\exp(-c_1k)
    =
    o(1),
\label{eq:app-large-k-no-block-bound}
\end{equation}
because \(k=\Theta(N)\). Thus, with probability \(1-o(1)\), the cross term contains no
usable row/column \(k\)-collision block.

On this no-cross-block event, the cross-term part of any across-label block-SWAP estimator
has no state-dependent measurement outcome. The self-term transcripts have the same distribution
under \(H_0\) and \(H_1\), because \(\scV\) is a unitary image of \(\scU\) and the multi-SWAP
measurements depend only on within-ensemble Gram matrices. Therefore, if \(P_0\) and \(P_1\)
are the output laws under \(H_0\) and \(H_1\),
\[
    \TV(P_0,P_1)=o(1)
\]
whenever \(M\le c_0Nk\).

If an estimator had additive error at most \(\epsilon_0<1/4\) with failure probability at most
\(\delta_0<1/2\) under both hypotheses, then thresholding its output at \(1/2\) would give a
test with
\[
    \TV(P_0,P_1)\ge 1-2\delta_0,
\]
contradicting \(\TV(P_0,P_1)=o(1)\) for large \(N\). Therefore any such across-label
block-SWAP estimator must use
\[
    M_{\rm across}=\Omega(Nk).
\]
Combining this lower bound with the upper bound proves
\begin{equation}
    M_{\rm across}=\Theta(Nk)=\Theta(N^2)
\label{eq:app-large-k-complexity}
\end{equation}
for \(k=\Theta(N)\).
\end{proof}

\subsection{Proof of Theorem~\ref{thm: sample complexity MMD-k, classical} for classical analog}
\label{app:Proof of classical analog}
Here we prove Theorem~\ref{thm: sample complexity MMD-k, classical} from Sec.~\ref{secsub: sc of MMDk general}.

\begin{proof}
\textbf{Upper bound.}
We use the ordinary SWAP test. Fix \(a,b\in\{1,2\}\). To estimate \(\bar F^{(k)}(\scE_a,\scE_b)\), draw \(m\) independent pairs of states
\[
    \ket{x_t}\sim \scE_a,
    \qquad
    \ket{y_t}\sim \scE_b,
    \qquad
    t=1,\ldots,m .
\]
For each pair, perform a SWAP test and define
\[
    R_t=
    \begin{cases}
        +1, & \text{if the SWAP-test ancilla is measured as }0,\\
        -1, & \text{if the SWAP-test ancilla is measured as }1.
    \end{cases}
\]
Conditioned on \((x_t,y_t)\), the SWAP-test outcome satisfies
\begin{equation}
    \dsE[R_t\mid x_t,y_t]
    =
    \abs{\braket{x_t}{y_t}}^2
    =
    \mathbf 1\{x_t=y_t\}.
\end{equation}
Since \(X^k=X\) for \(X\in\{0,1\}\), we have
\begin{equation}
    \dsE[R_t]
    =
    \bar F^{(k)}(\scE_a,\scE_b).
\end{equation}
Thus
\begin{equation}
    \widehat{\bar F}^{(k)}_{ab}
    :=
    \frac1m\sum_{t=1}^m R_t
\end{equation}
is an unbiased estimator of \(\bar F^{(k)}(\scE_a,\scE_b)\). Since \(R_t\in[-1,1]\), Hoeffding's inequality gives
\begin{equation}
    \Pr\!\left[
        \left|
        \widehat{\bar F}^{(k)}_{ab}
        -
        \bar F^{(k)}(\scE_a,\scE_b)
        \right|
        >
        \eta
    \right]
    \le
    2\exp\!\left(-\frac{m\eta^2}{2}\right).
    \label{eq: classical swap hoeffding}
\end{equation}
Apply this independently to \((a,b)=(1,1),(2,2),(1,2)\), and define
\begin{equation}
    \widehat{\scD}^{(k)}
    :=
    \widehat{\bar F}^{(k)}_{11}
    +
    \widehat{\bar F}^{(k)}_{22}
    -
    2\widehat{\bar F}^{(k)}_{12}.
\end{equation}
If each of the three estimates is accurate to within \(\epsilon/4\), then
\begin{equation}
    \left|
    \widehat{\scD}^{(k)}
    -
    \scD^{(k)}
    \right|
    \le
    \epsilon .
\end{equation}
Taking \(\eta=\epsilon/4\) in Eq.~(\ref{eq: classical swap hoeffding}) and allocating failure probability \(\delta/3\) to each term, it suffices to choose
\begin{equation}
    m
    \ge
    \frac{32}{\epsilon^2}\log\frac{6}{\delta}.
\end{equation}
The total number of state samples is a constant multiple of \(m\), so
\begin{equation}
    M
    =
    \scO\!\left(
        \frac{1}{\epsilon^2}\log\frac1\delta
    \right).
\end{equation}
This bound is independent of the ensemble size and independent of \(k\).

\medskip

\textbf{Lower bound for arbitrary protocols.}
It suffices to prove the lower bound for a two-dimensional computational-basis subfamily. Let
\[
    q
    :=
    \left(\frac34,\frac14\right),
\]
and for a parameter \(\gamma\in(0,1/4)\), define
\begin{equation}
    p^+
    :=
    \left(\frac12+\gamma,\frac12-\gamma\right),
    \qquad
    p^-
    :=
    \left(\frac12-\gamma,\frac12+\gamma\right).
\label{eq:app-classical-hard-probabilities}
\end{equation}
Consider the two hypotheses
\[
\begin{gathered}
H_+:\quad
    (\scE_1,\scE_2)=(\scE_{p^+},\scE_q),
    \\
    H_-:\quad
    (\scE_1,\scE_2)=(\scE_{p^-},\scE_q).
\end{gathered}
\]
By Eq.~(\ref{eq: classical MMD-k collapse}),
\begin{align}
    D_+
    &:=
    \scD^{(k)}(\scE_{p^+},\scE_q)
    =
    2\left(\frac14-\gamma\right)^2,\\
    D_-
    &:=
    \scD^{(k)}(\scE_{p^-},\scE_q)
    =
    2\left(\frac14+\gamma\right)^2.
\end{align}
Hence
\begin{equation}
    |D_- - D_+|
    =
    2\gamma .
    \label{eq: classical hard gap}
\end{equation}
Set \(\gamma=4\epsilon\), assuming \(\epsilon\) is sufficiently small so that \(\gamma<1/4\). Then
\[
    |D_- - D_+|=8\epsilon .
\]

Now consider an arbitrary protocol in the general measurement model. The protocol may choose adaptively which ensemble to sample from, may store all quantum registers, and may perform an arbitrary final joint POVM. Let \(Y\) denote the final classical output of the protocol. Since all possible states under \(H_+\) and \(H_-\) are diagonal in the same computational basis, every such POVM is equivalent, on this subfamily, to first revealing the full computational-basis transcript \(Z\) and then applying a classical randomized post-processing kernel. Therefore, by data processing,
\begin{equation}
    \mathrm{KL}(P_+^Y\|P_-^Y)
    \le
    \mathrm{KL}(P_+^Z\|P_-^Z),
\end{equation}
where \(P_+^Y,P_-^Y\) are the laws of the protocol output under the two hypotheses, and \(P_+^Z,P_-^Z\) are the laws of the full classical transcript.

Samples from \(\scE_2=\scE_q\) have the same distribution under both hypotheses and contribute no KL divergence. Even if all \(M\) samples are taken from the informative ensemble \(\scE_1\), the transcript is no more informative than \(M\) independent Bernoulli samples with parameters \(1/2+\gamma\) under \(H_+\) and \(1/2-\gamma\) under \(H_-\). Hence the chain rule for KL divergence gives
\begin{widetext}
\begin{equation}
    \mathrm{KL}(P_+^Y\|P_-^Y)
    \le
    M\,
    \mathrm{KL}\!\left(
        \operatorname{Bern}\!\left(\frac12+\gamma\right)
        \middle\|
        \operatorname{Bern}\!\left(\frac12-\gamma\right)
    \right)
    \le
    C M\gamma^2 ,
    \label{eq: classical KL upper}
\end{equation}
\end{widetext}
for a universal constant \(C>0\), where the last inequality holds for \(\gamma\le1/4\).

Suppose an estimator \(\widehat D\) achieves additive error at most \(\epsilon\) and failure probability at most \(\delta\) under both hypotheses. Define the test
\begin{equation}
    \Psi
    :=
    \mathbf 1\!\left\{
        \widehat D>
        \frac{D_++D_-}{2}
    \right\}.
\end{equation}
Since \(|D_- - D_+|=8\epsilon\), this test has type-I and type-II errors at most \(\delta\):
\[
    \Pr_{H_+}(\Psi=1)\le\delta,
    \qquad
    \Pr_{H_-}(\Psi=0)\le\delta .
\]
By the Bretagnolle--Huber inequality,
\begin{equation}
    2\delta
    \ge
    \Pr_{H_+}(\Psi=1)+\Pr_{H_-}(\Psi=0)
    \ge
    \frac12
    \exp\!\left(
        -\mathrm{KL}(P_+^Y\|P_-^Y)
    \right).
\end{equation}
Combining this with Eq.~(\ref{eq: classical KL upper}) gives
\begin{equation}
    C M\gamma^2
    \ge
    \log\frac1{4\delta}.
\end{equation}
Using \(\gamma=4\epsilon\), we obtain
\begin{equation}
    M
    =
    \Omega\!\left(
        \frac{1}{\epsilon^2}\log\frac1\delta
    \right).
\end{equation}
This lower bound holds for arbitrary protocols because the hard instances are commuting computational-basis ensembles.

Combining the upper and lower bounds proves
\begin{equation}
    M
    =
    \Theta\!\left(
        \frac{1}{\epsilon^2}\log\frac1\delta
    \right).
\label{eq:app-classical-matching-complexity}
\end{equation}
\end{proof}

\subsection{Proof of Theorem~\ref{thm:lower-bound-full-discriminative}}\label{app:lower-bound-full-discriminative}

\begin{proof}
Choose \(N+1\) distinct pure states
\[
    \rho_1,\ldots,\rho_N,\tau
\]
in \(\mathbb C^d\). This is possible for every \(N\) when \(d\ge 2\), since the pure-state space is
infinite. Define
\[
    \mathcal E_N
    :=
    \left\{\left(\frac1N,\rho_i\right)\right\}_{i=1}^N
\]
and set
\begin{equation}
    \mathcal F_N^{(0)}
    :=
    \mathcal E_N,
    \qquad
    \mathcal F_N^{(1)}
    :=
    \left\{\left(\frac1N,\tau\right)\right\}
    \cup
    \left\{\left(\frac1N,\rho_i\right)\right\}_{i=2}^N .
\label{eq:app-full-discrimination-hard-ensembles}
\end{equation}
Since \(\tau\) is distinct from every \(\rho_i\), we have
\(\mathcal F_N^{(1)}\neq \mathcal E_N\). By full discriminative power,
\begin{equation}
    \Delta_N
    :=
    \scD_N(\mathcal E_N,\mathcal F_N^{(1)})
    >
    0,
\label{eq:app-full-discrimination-instance-gap}
\end{equation}
whereas
\[
    \scD_N(\mathcal E_N,\mathcal F_N^{(0)})=0.
\]

Consider the binary testing problem
\[
\begin{gathered}
H_0:\quad (\mathcal E_N,\mathcal F_N)=(\mathcal E_N,\mathcal F_N^{(0)}),
    \\
    H_1:\quad (\mathcal E_N,\mathcal F_N)=(\mathcal E_N,\mathcal F_N^{(1)}).
\end{gathered}
\]
The two hypotheses differ only in label \(1\) of the second ensemble. Let \(A\) be the event that the
protocol never receives a sample from label \(1\) of the second ensemble. Even if the protocol
adaptively chooses which ensemble to sample from, the number of samples drawn from the second
ensemble is at most \(M\), and each such draw has label \(1\) with probability \(1/N\). Hence
\[
    \Pr(A^c)\le \frac{M}{N}.
\]
Conditioned on \(A\), the complete labeled quantum input to the protocol is identical under
\(H_0\) and \(H_1\). This remains true for arbitrary quantum memory, arbitrary adaptive choices,
and arbitrary joint POVMs on all collected registers. Therefore, if \(P_0\) and \(P_1\) denote the
final output distributions under \(H_0\) and \(H_1\), then by coupling,
\begin{equation}
    \operatorname{TV}(P_0,P_1)\le \Pr(A^c)\le \frac{M}{N}.
\label{eq:app-full-discrimination-tv-bound}
\end{equation}

Now suppose that an estimator \(\widehat D\) achieves additive error at most \(\Delta_N/4\) and
failure probability at most \(\delta_0\) under both hypotheses. Define the test
\[
    \Psi
    :=
    \mathbf 1\!\left\{\widehat D>\frac{\Delta_N}{2}\right\}.
\]
Under \(H_0\), the true distance is \(0\), so
\[
    \Pr_{H_0}(\Psi=1)\le \delta_0.
\]
Under \(H_1\), the true distance is \(\Delta_N\), so
\[
    \Pr_{H_1}(\Psi=0)\le \delta_0.
\]
For any test, the total variation distance satisfies
\[
    \operatorname{TV}(P_0,P_1)
    \ge
    1-\Pr_{H_0}(\Psi=1)-\Pr_{H_1}(\Psi=0)
    \ge
    1-2\delta_0 .
\]
Combining the two total-variation bounds gives
\[
    \frac{M}{N}\ge 1-2\delta_0,
\]
and therefore
\begin{equation}
    M\ge (1-2\delta_0)N.
\label{eq:app-full-discrimination-sample-lower}
\end{equation}
This proves the theorem.

If \(\Delta_N\ge 4\epsilon_0\), then any estimator with additive error at most \(\epsilon_0\) is also
accurate to within \(\Delta_N/4\); In those cases, the lower bound also holds for a constant additive error.
\end{proof}

\subsection{Proof of Lemma~\ref{lemma:stability}}
\label{app:stability}
In this section, we detail the stability analyses of distance metrics. First, we make the definition of stability explicit.
\begin{definition}
\label{def:stability}
(Stability under state perturbations) For two pairs of uniform ensembles: (1) $\scE=\{(\frac{1}{N},\rho_i)\}_{i=1}^N$ and $\scF=\{(\frac{1}{N},\sigma_i)\}_{i=1}^N$; (2) $\widehat\scE=\{(\frac{1}{N},\widehat\rho_i)\}_{i=1}^N$ and $\widehat\scF=\{(\frac{1}{N},\widehat\sigma_i)\}_{i=1}^N$, suppose the states are pairwise close,
$
1-F(\rho_i,\widehat\rho_i)\le \epsilon,
1-F(\sigma_i,\widehat\sigma_i)\le \epsilon
$
for every \(i\in[N]\). We say a distance metric $\scD$ is stable, if 
$
\left|
    \scD(\mathcal E,\mathcal F)
    -
    \scD(\widehat{\mathcal E},\widehat{\mathcal F})
    \right|
    \le f(\epsilon)
$, 
where $f(\epsilon)$ is independent from $N$ and $\lim_{\epsilon\to0} f(\epsilon)= 0$.
\end{definition}

Now we prove Lemma~\ref{lemma:stability}.

\begin{proof}
We use
\[
    d_{\sin}(\rho,\sigma):=\sqrt{1-F(\rho,\sigma)} .
\]
For pure states, \(d_{\sin}\) is the sine distance and satisfies the triangle inequality. The perturbation assumption gives
\[
    d_{\sin}(\rho_i,\widehat\rho_i)\le \sqrt{\epsilon},
    \qquad
    d_{\sin}(\sigma_i,\widehat\sigma_i)\le \sqrt{\epsilon}.
\]

First consider the Wasserstein distance. Let
\[
    C_{ij}:=1-F(\rho_i,\sigma_j),
    \qquad
    \widehat C_{ij}:=1-F(\widehat\rho_i,\widehat\sigma_j).
\]
Since \(C_{ij}=d_{\sin}(\rho_i,\sigma_j)^2\), we have
\begin{align}
    |C_{ij}-\widehat C_{ij}|
    &=
    \left|
        d_{\sin}(\rho_i,\sigma_j)^2
        -
        d_{\sin}(\widehat\rho_i,\widehat\sigma_j)^2
    \right| \nonumber\\
    &\le
    2\left|
        d_{\sin}(\rho_i,\sigma_j)
        -
        d_{\sin}(\widehat\rho_i,\widehat\sigma_j)
    \right| \nonumber\\
    &\le
    2d_{\sin}(\rho_i,\widehat\rho_i)
    +
    2d_{\sin}(\sigma_j,\widehat\sigma_j)
    \le
    4\sqrt{\epsilon}.
\end{align}
Thus \(\|C-\widehat C\|_\infty\le 4\sqrt{\epsilon}\). For uniform ensembles,
\[
    W(\scE,\scF)
    =
    \min_{P}
    \langle P,C\rangle,
    \qquad
    W(\widehat\scE,\widehat\scF)
    =
    \min_{P}
    \langle P,\widehat C\rangle,
\]
where \(P\ge0\), \(P\mathbf 1=\frac1N\mathbf 1\), and
\(P^\top\mathbf 1=\frac1N\mathbf 1\). For every feasible \(P\),
\[
    \sum_{i,j}P_{ij}=1,
    \qquad
    |\langle P,C-\widehat C\rangle|
    \le
    \|C-\widehat C\|_\infty .
\]
Taking \(P\) to be an optimizer for either \(C\) or \(\widehat C\), we obtain
\begin{equation}
    \left|
        W(\scE,\scF)
        -
        W(\widehat\scE,\widehat\scF)
    \right|
    \le
    \|C-\widehat C\|_\infty
    \le
    4\sqrt{\epsilon}.
\label{eq:app-wasserstein-stability-bound}
\end{equation}
Therefore Wasserstein is stable with
\[
    f_W(\epsilon)=4\sqrt{\epsilon}.
\]

We now consider MMD-\(k\). For any two perturbed pairs of pure states satisfying the above component-wise condition, the same argument gives
\[
    \left|
        F(\rho_i,\sigma_j)
        -
        F(\widehat\rho_i,\widehat\sigma_j)
    \right|
    \le
    4\sqrt{\epsilon}.
\]
Since \(x\mapsto x^k\) is \(k\)-Lipschitz on \([0,1]\),
\[
    \left|
        F(\rho_i,\sigma_j)^k
        -
        F(\widehat\rho_i,\widehat\sigma_j)^k
    \right|
    \le
    4k\sqrt{\epsilon}.
\]
Averaging over \(i,j\) gives, for any two ensembles \(\scA,\scB\) and their perturbations
\(\widehat\scA,\widehat\scB\),
\[
    \left|
        \bar F^{(k)}(\scA,\scB)
        -
        \bar F^{(k)}(\widehat\scA,\widehat\scB)
    \right|
    \le
    4k\sqrt{\epsilon}.
\]
Applying this bound to the three terms in
\[
    \scD^{(k)}(\scE,\scF)
    =
    \bar F^{(k)}(\scE,\scE)
    +
    \bar F^{(k)}(\scF,\scF)
    -
    2\bar F^{(k)}(\scE,\scF),
\]
we obtain
\begin{equation}
    \left|
        \scD^{(k)}(\scE,\scF)
        -
        \scD^{(k)}(\widehat\scE,\widehat\scF)
    \right|
    \le
    16k\sqrt{\epsilon}.
\label{eq:app-mmd-stability-bound}
\end{equation}
Thus MMD-\(k\) is stable for every fixed \(k\), with
\[
    f_{\mathrm{MMD}\text{-}k}(\epsilon)=16k\sqrt{\epsilon}.
\]
For fixed \(k\), this function is independent of \(N\) and satisfies
\[
    \lim_{\epsilon\to0}16k\sqrt{\epsilon}=0.
\]
This proves the lemma.
\end{proof}

\subsection{Proof of Theorem~\ref{thm:general-sample-complexity-tomography}}\label{app:general-sample-complexity-tomography}

\begin{proof}
We first prove the upper bound. Let
\[
    \scE=\left\{\left(\frac1N,\rho_i\right)\right\}_{i=1}^N,
    \qquad
    \scF=\left\{\left(\frac1N,\sigma_i\right)\right\}_{i=1}^N .
\]
By tomography-stability at accuracy \(\epsilon_0\), choose a constant
\(\eta_0>0\), independent of \(N\), such that if every component state of
\(\scE,\scF\) is replaced by a state with infidelity at most \(\eta_0\), then
the value of \(\scD_N\) changes by at most \(\epsilon_0\).

For a single unknown \(d\)-dimensional pure state \(\rho\), let
\(t_{\rm tom}(d,\eta_{\rm tomo},\nu)\) denote the number of copies needed by a
state-tomography subroutine, using an arbitrary collective POVM on those copies,
to output an estimate \(\widehat\rho\) satisfying
\[
    1-F(\rho,\widehat\rho)\le \eta_{\rm tomo}
\]
with failure probability at most \(\nu\). Here \(\eta_{\rm tomo}\) is the
tomography target accuracy, while \(\nu\) is the failure probability assigned to
one single-state tomography call. In the present tomography-evaluation protocol,
we set the tomography accuracy to be the stability radius,
\[
    \eta_{\rm tomo}=\eta_0.
\]
Since there are \(2N\) component states to reconstruct, we assign failure
probability
\[
    \nu=\frac{\delta_0}{4N}
\]
to each single-state tomography call. Then a union bound over all \(2N\)
tomography calls gives total tomography failure probability at most
\[
    2N\nu
    =
    2N\cdot \frac{\delta_0}{4N}
    =
    \frac{\delta_0}{2}.
\]
Set
\[
    t=t_{\rm tom}(d,\eta_0,\nu),
    \qquad
    L=\log\frac{4N}{\delta_0}.
\]
Draw
\begin{equation}
    m=C\,N(t+L)
\label{eq:app-tomography-sample-budget}
\end{equation}
samples from each ensemble, where \(C>0\) is a sufficiently large universal
constant. For any fixed label in either ensemble, let \(T_i\) be the number of
times that label appears. Then
\[
    T_i\sim \operatorname{Bin}\!\left(m,\frac1N\right),
    \qquad
    \mathbb E[T_i]=\frac{m}{N}=C(t+L).
\]
Choosing \(C\ge 8\), we have \(t\le \mathbb E[T_i]/2\). Hence, by the Chernoff
lower-tail bound,
\[
\begin{aligned}
\Pr(T_i<t)
    &\le
    \Pr\!\left(T_i<\frac{\mathbb E[T_i]}2\right)\\
    &\le
    \exp\!\left(-\frac{\mathbb E[T_i]}8\right)
    \le
    \exp\!\left(-\frac{CL}{8}\right).
\end{aligned}
\]
Because \(L=\log(4N/\delta_0)\), taking \(C\ge 8\) gives
\[
    \exp\!\left(-\frac{CL}{8}\right)
    \le
    \exp(-L)
    =
    \frac{\delta_0}{4N}.
\]
Thus
\[
    \Pr(T_i<t)\le \frac{\delta_0}{4N}.
\]
A union bound over the \(2N\) component labels in \(\scE\) and \(\scF\) implies
that, with probability at least \(1-\delta_0/2\), every component state in both
ensembles is observed at least \(t\) times.

Condition on this coverage event. For each label, use \(t\) copies and perform
collective tomography on those copies. By definition of \(t_{\rm tom}\), each
component-state reconstruction has infidelity at most \(\eta_0\) with failure
probability at most \(\nu=\delta_0/(4N)\). A second union bound over the \(2N\)
component states gives that, with probability at least \(1-\delta_0/2\),
\[
    1-F(\rho_i,\widehat\rho_i)\le \eta_0,
    \qquad
    1-F(\sigma_i,\widehat\sigma_i)\le \eta_0
\]
for every \(i\in[N]\). Let
\[
    \widehat{\scE}
    :=
    \left\{\left(\frac1N,\widehat\rho_i\right)\right\}_{i=1}^N,
    \qquad
    \widehat{\scF}
    :=
    \left\{\left(\frac1N,\widehat\sigma_i\right)\right\}_{i=1}^N .
\]
Combining the coverage event and the tomography-success event, with probability
at least \(1-\delta_0\), all component states are reconstructed to the required
entrywise accuracy. Therefore, by tomography-stability,
\begin{equation}
    \left|
        \scD_N(\scE,\scF)
        -
        \scD_N(\widehat{\scE},\widehat{\scF})
    \right|
    \le
    \epsilon_0 .
\label{eq:app-tomography-distance-accuracy}
\end{equation}
Thus the tomography-evaluation estimator
\[
    \widehat D
    :=
    \scD_N(\widehat{\scE},\widehat{\scF})
\]
has additive error at most \(\epsilon_0\) and failure probability at most
\(\delta_0\).

The total number of samples is
\begin{equation}
    M=2m
    =
    O\!\left(
        N\left[
            t_{\rm tom}\!\left(d,\eta_0,\frac{\delta_0}{4N}\right)
            +
            \log\frac{N}{\delta_0}
        \right]
    \right).
\label{eq:app-tomography-general-complexity}
\end{equation}

We now explain why this is \(O(N\log N)\) for fixed
\(d,\epsilon_0,\delta_0\). The sample-optimal collective tomography result of
Haah et al.~\cite{sample_optimal_tomography_7956181} gives a collective POVM
which reconstructs a rank-\(r\), \(d\)-dimensional state to infidelity accuracy
\(\eta_{\rm tomo}\) using
\[
    O\!\left(
        \frac{dr}{\eta_{\rm tomo}}
        \log\frac{d}{\eta_{\rm tomo}}
    \right)
\]
copies with constant success probability. Applying standard confidence
amplification, we may take
\[
    t_{\rm tom}(d,\eta_{\rm tomo},\nu)
    =
    O\!\left(
        \frac{dr}{\eta_{\rm tomo}}
        \log\frac{d}{\eta_{\rm tomo}}
        \log\frac1\nu
    \right).
\]
Since the component states are pure, \(r=1\), and therefore
\begin{equation}
    t_{\rm tom}(d,\eta_{\rm tomo},\nu)
    =
    O\!\left(
        \frac{d}{\eta_{\rm tomo}}
        \log\frac{d}{\eta_{\rm tomo}}
        \log\frac1\nu
    \right).
\label{eq:app-pure-state-tomography-cost}
\end{equation}
In the present theorem, \(d\), \(\epsilon_0\), and \(\delta_0\) are fixed
constants, and \(\eta_{\rm tomo}=\eta_0\) is independent of \(N\). Thus the only
\(N\)-dependent factor in the tomography cost is
\[
    \log\frac1\nu
    =
    \log\frac{4N}{\delta_0}
    =
    \log N+O(1).
\]
Consequently,
\[
    t_{\rm tom}\!\left(d,\eta_0,\frac{\delta_0}{4N}\right)
    =
    O\!\left(
        \frac{d}{\eta_0}
        \log\frac{d}{\eta_0}
        \log\frac{4N}{\delta_0}
    \right)
    =
    O(\log N).
\]
The remaining logarithmic term
\[
    L=\log\frac{4N}{\delta_0}
\]
comes from the label-coverage step: we need enough random samples so that every
one of the \(N\) labels appears at least \(t\) times with high probability.
Therefore
\begin{equation}
    M
    =
    O\!\left(
        N\left[
            \log N
            +
            \log\frac{N}{\delta_0}
        \right]
    \right)
    =
    O(N\log N)
\label{eq:app-tomography-upper-complexity}
\end{equation}
for fixed \(d,\epsilon_0,\delta_0\).

We now prove the tomography-evaluation lower bound. It suffices to consider the
problem of reconstructing a single \(N\)-state uniform pure-state ensemble;
reconstructing both ensembles is at least as hard.

Fix any target infidelity accuracy \(\eta_*<1/2\). Choose \(N\) pairs of pure
states
\[
    \{\rho_i^{(0)},\rho_i^{(1)}\}_{i=1}^N
\]
such that all \(2N\) states are distinct and
\[
    F(\rho_i^{(0)},\rho_i^{(1)})=0
    \qquad\text{for every }i.
\]
This is possible in any dimension \(d\ge2\). Let
\(B_1,\ldots,B_N\) be independent uniform bits and consider the random ensemble
\begin{equation}
    \scE_B
    :=
    \left\{
        \left(\frac1N,\rho_i^{(B_i)}\right)
    \right\}_{i=1}^N .
\label{eq:app-tomography-hard-ensemble}
\end{equation}
Let \(m\) be the number of samples drawn from this ensemble, and let \(U\) be
the number of labels that are not observed. If
\[
    m\le \frac14 N\log N,
\]
then the coupon-collector lower bound implies
\[
    \Pr(U>0)\to 1
    \qquad
    \text{as }N\to\infty .
\]
For completeness, one can see this directly as follows. Let \(I_i\) be the
indicator that label \(i\) is unseen. Then
\[
    \mathbb E[U]
    =
    N\left(1-\frac1N\right)^m
    \ge
    N^{1/2}
\]
for all sufficiently large \(N\). Moreover, for \(i\ne j\),
\[
    \mathbb E[I_iI_j]
    =
    \left(1-\frac2N\right)^m
    \le
    \left(1-\frac1N\right)^{2m}
    =
    \mathbb E[I_i]\mathbb E[I_j],
\]
so
\[
    \operatorname{Var}(U)\le \mathbb E[U].
\]
Therefore
\[
    \Pr(U=0)
    \le
    \frac{\operatorname{Var}(U)}{\mathbb E[U]^2}
    \le
    \frac1{\mathbb E[U]}
    \to 0.
\]

Condition on any transcript for which at least one label \(i\) is unseen. The
transcript contains no quantum system whose state depends on \(B_i\), so,
conditional on the transcript, the bit \(B_i\) remains uniform. For any output
state \(\widehat\rho_i\), it is impossible to be \(\eta_*\)-accurate for both
possible values \(\rho_i^{(0)}\) and \(\rho_i^{(1)}\). Indeed, since the two
states are orthogonal,
\begin{equation}
    F(\rho_i^{(0)},\widehat\rho_i)
    +
    F(\rho_i^{(1)},\widehat\rho_i)
    \le 1,
\label{eq:app-tomography-unseen-label-obstruction}
\end{equation}
so both fidelities cannot exceed \(1-\eta_*>1/2\). Thus, conditioned on an
unseen label, the probability of reconstructing all component states to
infidelity at most \(\eta_*\) is at most \(1/2\).

Consequently, if \(m\le \frac14 N\log N\), the average success probability over
the above random choice of \(\scE_B\) is at most
\[
    \Pr(U=0)+\frac12\Pr(U>0)
    =
    \frac12+o(1).
\]
Hence there exists at least one ensemble in this family on which the tomography
procedure fails to achieve success probability, say, \(2/3\). Therefore any
procedure that uniformly reconstructs every component state to fixed infidelity
less than \(1/2\) with constant success probability must use
\[
    m=\Omega(N\log N)
\]
samples from the ensemble. Applying this to the tomography stage of tomography
evaluation gives
\[
    M=\Omega(N\log N).
\]
Together with the upper bound, this proves
\begin{equation}
    M_{\rm tomo}=\Theta(N\log N)
\label{eq:app-tomography-matching-complexity}
\end{equation}
for fixed \(d,\epsilon_0,\delta_0\).
\end{proof}

\subsection{Proof of Theorem~\ref{thm:wasserstein-general-lower}: general lower bound for Wasserstein distance}
\label{app:wasserstein-general-lower}
\begin{proof}
We use the estimation-to-testing reduction from Appendix~\ref{app: proof lower general}. The construction differs in its geometric ingredient: two finite distributions of pure states have identical low-order moments, but their supports are separated by a constant fidelity cost. We use these distributions to generate two hypotheses for the second ensemble, while keeping the first ensemble fixed. Under one hypothesis, a low-cost transport plan exists with high probability; under the other, every transport plan has a constant cost. As in Appendix~\ref{app: proof lower general}, the sampled inputs remain indistinguishable until sufficiently many copies of the same comparison component are obtained.

Fix $\gamma\in(0,1)$, $\epsilon_0\in(0,1/2)$, and $\delta_0\in(0,1/2)$, and set
\begin{equation}
t:=\lceil\gamma^{-1}\rceil-1\ge1,\qquad \alpha:=\frac{1-2\epsilon_0}{4}>0.
\label{eq:app-wasserstein-design-parameters}
\end{equation}
We construct the ensembles in dimension $d_0:=\lceil Ct\rceil$, where $C\ge2$ will depend only on $\epsilon_0$. The integer $t$ and the dimension $d_0$ remain fixed as $N\to\infty$.

\emph{Step 1: A finite distribution matching the Haar moments.} Let
\[
H_s:=\int \rho_\psi^{\otimes s}\,d\psi,\qquad \rho_\psi:=|\psi\rangle\langle\psi|,\qquad H_0:=1,
\]
where $d\psi$ is the normalized Haar measure on pure states in $\mathbb C^{d_0}$. Each $H_s$ is invariant under $V^{\otimes s}$ for every unitary $V$, by invariance of the integration measure. Since $|\psi\rangle^{\otimes t}$ belongs to $\operatorname{Sym}^t(\mathbb C^{d_0})$, whose dimension is
\[
D_t:=\binom{d_0+t-1}{t},
\]
the operators $\rho_\psi^{\otimes t}$ lie in a real vector space of Hermitian matrices of dimension $D_t^2$. Their set is compact, and $H_t$ lies in its convex hull. Carath\'eodory's theorem, which represents any point in a convex hull in $\mathbb R^m$ using at most $m+1$ points, therefore gives distinct pure states $\tau_1,\ldots,\tau_K$ and positive weights $w_1,\ldots,w_K$ such that
\begin{equation}
K\le D_t^2+1,\qquad \sum_{a=1}^Kw_a=1,\qquad \sum_{a=1}^Kw_a\tau_a^{\otimes t}=H_t.
\label{eq:app-wasserstein-finite-design}
\end{equation}
Zero weights can be discarded and repeated states merged. Taking partial traces gives
\begin{equation}
\sum_{a=1}^Kw_a\tau_a^{\otimes s}=H_s\qquad(0\le s\le t).
\label{eq:app-wasserstein-design-moments}
\end{equation}
This finite weighted state $t$-design supplies the moments required for indistinguishability, with $K$ independent of $N$; its weights will specify the hard prior over hidden components of uniform ensembles.

\emph{Step 2: Separating the support from a rotated copy.} We next choose a unitary $U$ such that every state $\tau_a$ has small fidelity with every rotated state $U\tau_{a'}U^\dagger$. For Haar-random $U$ and fixed $a,a'$, the overlap $X:=\operatorname{Tr}(\tau_aU\tau_{a'}U^\dagger)$ has the same distribution as $Y_1/(Y_1+\cdots+Y_{d_0})$, where the $Y_\ell$ are independent exponential random variables of mean one. This follows by expressing a Haar vector as a normalized vector of independent complex Gaussian entries. Conditioning on $Y_2,\ldots,Y_{d_0}$ gives
\[
\Pr_U(X>\alpha)=\mathbb E\exp\!\left[-\frac{\alpha}{1-\alpha}\sum_{\ell=2}^{d_0}Y_\ell\right]=(1-\alpha)^{d_0-1}.
\]
A union bound over the $K^2$ pairs therefore yields
\[
\Pr_U\!\left(\max_{a,a'}\operatorname{Tr}(\tau_aU\tau_{a'}U^\dagger)>\alpha\right)\le K^2(1-\alpha)^{d_0-1}.
\]
The dimension can be chosen proportional to $t$. Indeed, $d_0=\lceil Ct\rceil$ implies
\[
D_t\le\left[\frac{e(d_0+t-1)}{t}\right]^t\le[e(C+1)]^t,\qquad K\le2D_t^2,
\]
and hence
\[
K^2(1-\alpha)^{d_0-1}\le4e^\alpha\exp\!\left(t\left[4\log(e(C+1))-\alpha C\right]\right).
\]
Choose $C\ge2$ so that $\alpha C>4\log(e(C+1))+\log4+\alpha$. Such a constant exists and depends only on $\epsilon_0$, and the last bound is then strictly less than one for every $t\ge1$. We may therefore fix a unitary $U$ satisfying
\begin{equation}
\operatorname{Tr}(\tau_aU\tau_{a'}U^\dagger)\le\alpha\qquad\text{for all }a,a'.
\label{eq:app-wasserstein-separated-supports}
\end{equation}
This step gives two distributions with identical moments through order $t$ and supports separated by a constant fidelity cost, in a dimension proportional to $t$ and independent of $N$. The separation will force a constant transport cost under one hypothesis.

\emph{Step 3: Valid uniform ensembles with distinct components.} Small label-dependent rotations prevent repeated support points from becoming repeated ensemble components. Set $\eta:=\alpha/8$. For each $N$, choose unitaries $V_1,\ldots,V_N$ with $\|V_j-I\|_\infty<\eta$, and define
\[
\tau_{b,j,a}:=V_jU^b\tau_a(U^\dagger)^bV_j^\dagger,\qquad b\in\{0,1\}.
\]
Choose the $V_j$ so that, for each fixed $b$, the sets $\{\tau_{b,j,a}:a\in[K]\}$ are disjoint for different $j$. Such a choice exists inductively: at each step only finitely many equalities with previously chosen pure states are forbidden, and each equality has probability zero under a Haar-random unitary conditioned to lie in this neighborhood of the identity. Fix these unitaries before drawing any hidden variables.

For any density operator $X$, $\|VXV^\dagger-X\|_1\le2\|V-I\|_\infty$. Applying this estimate to each factor in a fidelity shows, for all relevant indices, that
\begin{equation}
\begin{gathered}
\operatorname{Tr}(\tau_{0,i,a}\tau_{1,j,a'})\le\alpha+4\eta\le2\alpha,\\ 1-\operatorname{Tr}(\tau_{0,i,a}\tau_{0,j,a})\le4\eta=\alpha/2.
\end{gathered}
\label{eq:app-wasserstein-perturbed-costs}
\end{equation}
Thus the perturbations preserve the separation between the two supports and make states associated with the same original support point close to one another.

Choose deterministic indices $a_1,\ldots,a_N\in[K]$ whose frequencies $P_a:=N^{-1}|\{i:a_i=a\}|$ satisfy
\[
\|P-w\|_1\le\frac KN.
\]
This is obtained by rounding the counts $Nw_a$ down and distributing the remaining counts, so that each frequency differs from $w_a$ by at most $1/N$. Fix
\begin{equation}
\rho_i:=\tau_{0,i,a_i},\qquad \mathcal E_1:=\left\{\left(\frac1N,\rho_i\right)\right\}_{i=1}^N.
\label{eq:app-wasserstein-reference-ensemble}
\end{equation}
Under hypothesis $H_b$, draw $\ell_1,\ldots,\ell_N$ independently with $\Pr(\ell_j=a)=w_a$ and set
\begin{equation}
\sigma_{b,j}:=\tau_{b,j,\ell_j},\qquad \mathcal E_{2,b}:=\left\{\left(\frac1N,\sigma_{b,j}\right)\right\}_{j=1}^N.
\label{eq:app-wasserstein-comparison-ensembles}
\end{equation}
The hidden indices are drawn once and held fixed throughout sampling. The weights $w_a$ describe only the hard prior; every realized ensemble has physical sampling weights exactly $1/N$. Together with $\mathcal E_1$, their finite supports define $\mathfrak F_N$; the reference ensemble may even be disclosed to the estimator.

The construction therefore gives valid uniform ensembles with $N$ distinct pure components, while preserving the transport-cost bounds needed to separate the hypotheses.

\emph{Step 4: A constant Wasserstein gap.} Write $W_b:=W(\mathcal E_1,\mathcal E_{2,b})$. Under $H_1$, every entry of the transport cost matrix is at least
\[
g:=1-2\alpha,
\]
so every feasible transport plan, whose total mass is one, has cost at least $g$. Hence $W_1\ge g$ for every realization.

Under $H_0$, let $Q_a:=N^{-1}|\{j:\ell_j=a\}|$. For each $a$, transport mass $\min(P_a,Q_a)$ between components associated with the same support state $\tau_a$; each such pair has cost at most $\alpha/2$. The remaining mass is
\[
1-\sum_a\min(P_a,Q_a)=\frac12\|P-Q\|_1
\]
and can be transported arbitrarily at cost at most one per unit mass. This constructs a feasible plan and gives
\[
W_0\le\frac{\alpha}{2}+\frac12\|P-Q\|_1.
\]
Since $NQ$ is multinomial with probabilities $w$,
\[
\begin{gathered}
\mathbb E_0\|Q-w\|_2^2=\frac{1-\sum_aw_a^2}{N}\le\frac1N,\\ \mathbb E_0\|Q-w\|_1\le\sqrt{\frac KN}.
\end{gathered}
\]
The second inequality follows from Cauchy--Schwarz. Combining it with the deterministic rounding error and applying Markov's inequality yields
\begin{equation}
\begin{aligned}
\Pr_0(W_0>\alpha)&\le\Pr_0(\|P-Q\|_1>\alpha)\\
&\le\frac1\alpha\left(\frac KN+\sqrt{\frac KN}\right)=:\zeta_N.
\end{aligned}
\label{eq:app-wasserstein-null-gap-probability}
\end{equation}
Here $K$ is independent of $N$, so $\zeta_N\to0$. Thus $W_0\le\alpha$ with high probability, whereas $W_1\ge g$ deterministically, and the gap satisfies
\begin{equation}
g-\alpha=1-3\alpha=2\epsilon_0+\alpha>2\epsilon_0.
\label{eq:app-wasserstein-estimation-gap}
\end{equation}

This gap allows accurate estimation to be converted into a binary test, with only a vanishing additional error under $H_0$. The priors remain unconditioned, preserving their label-wise independence.

\emph{Step 5: Low-copy indistinguishability and the sample lower bound.} For each comparison label, define
\[
\Gamma_{b,j}^{(s)}:=\mathbb E_b[\sigma_{b,j}^{\otimes s}]=\sum_{a=1}^Kw_a\tau_{b,j,a}^{\otimes s},\qquad \Gamma_{b,j}^{(0)}:=1.
\]
The exact design identities and unitary invariance of $H_s$ imply
\begin{equation}
\Gamma_{b,j}^{(s)}=(V_jU^b)^{\otimes s}H_s\bigl((V_jU^b)^\dagger\bigr)^{\otimes s}=H_s\qquad(0\le s\le t).
\label{eq:app-wasserstein-low-copy-matching}
\end{equation}
The small rotations therefore do not introduce any low-order mismatch.

By the presampling argument in Appendix~\ref{app: proof lower general}, it suffices to analyze $M$ samples from each ensemble, even for adaptive protocols using fewer samples. Let $L=((I_m,J_m))_{m=1}^M$ be the complete numerical label transcript and $B_j(L):=|\{m:J_m=j\}|$. Its law is $q(L)=N^{-2M}$ under both hypotheses. Independence of the hidden indices gives the conditional prior-averaged input
\begin{equation}
\rho_L^{(b)}=U_L\left[\left(\bigotimes_{m=1}^M\rho_{I_m}\right)\otimes\left(\bigotimes_{j=1}^N\Gamma_{b,j}^{(B_j(L))}\right)\right]U_L^\dagger,
\label{eq:app-wasserstein-conditional-input}
\end{equation}
where $U_L$ only restores the sampled register order. Hence $\rho_L^{(0)}=\rho_L^{(1)}$ whenever every comparison label appears at most $t$ times. The label-dependent POVM argument of Appendix~\ref{app: proof lower general} therefore gives, for every Borel set $S\subseteq\mathbb R$,
\begin{equation}
\begin{aligned}
&\left|\Pr_1(\widehat W\in S)-\Pr_0(\widehat W\in S)\right|\\
&\le\Pr_L\!\left(\max_jB_j(L)\ge t+1\right)\\
&\le\binom M{t+1}N^{-t}\le\frac{M^{t+1}}{(t+1)!N^t}.
\end{aligned}
\label{eq:app-wasserstein-testing-advantage}
\end{equation}
Indeed, outside the repeated-label event the Born probabilities agree, while on that event their difference is at most one; the collision bound follows by a union bound over the $(t+1)$-subsets of sampling positions.

Suppose that $\widehat W$ has additive error at most $\epsilon_0$ and failure probability at most $\delta_0$ uniformly over $\mathfrak F_N$. Declare $H_1$ when $\widehat W>\tau$ and $H_0$ otherwise, where $\tau:=(\alpha+g)/2$. Since $g-\alpha>2\epsilon_0$, an accurate estimate gives the correct decision whenever $W_0\le\alpha$ under $H_0$, and for every realization under $H_1$. Averaging the pointwise estimation guarantee over the original, unconditioned priors gives
\[
\Pr_0(\widehat W>\tau)\le\delta_0+\zeta_N,\qquad \Pr_1(\widehat W\le\tau)\le\delta_0.
\]
Consequently,
\[
1-2\delta_0-\zeta_N\le\Pr_1(\widehat W>\tau)-\Pr_0(\widehat W>\tau)\le\frac{M^{t+1}}{(t+1)!N^t}.
\]
For sufficiently large $N$, $\zeta_N\le(1-2\delta_0)/2$, and hence
\begin{equation}
M\ge\left[\frac{(t+1)!(1-2\delta_0)}2\right]^{1/(t+1)}N^{t/(t+1)}\ge cN^{1-\gamma},
\label{eq:app-wasserstein-sample-lower}
\end{equation}
where $c>0$ is independent of $N$ and the last inequality uses $t+1=\lceil\gamma^{-1}\rceil$. Finally, $d_0=\lceil Ct\rceil=O_{\epsilon_0}(\gamma^{-1})$ is independent of $N$. Embedding the construction into any $\mathbb C^d$ with $d\ge d_0$ preserves all costs and the indistinguishability argument, completing the proof.

Thus a constant testing advantage requires a constant probability of a $(t+1)$-fold label collision, which forces the claimed sample lower bound even for collective measurements and adaptive protocols.
\end{proof}

\section{Details on Section~\ref{sec:numerical and applications}: numerical simulations and applications}
\label{app:numerical details}

\subsection{The algorithm of numerical simulation}\label{app: algorithms}

\subsubsection{Coherent protocol simulation}
\label{app:algorithm-simulation}

We describe the numerical procedure used to estimate the sample complexity upper
bound for learning quantum ensembles.  The goal of the simulation is not to
derive a new estimator, but to determine, for each pair $(N,k)$, a sample budget
$M$ for which the empirical failure probability is at most $\delta$ at additive
accuracy $\epsilon$.

Let
\begin{equation}
    s = \frac{\log(1/\delta)}{\epsilon^2}.
\end{equation}
The predicted sample-complexity scale is
\begin{equation}
    M =
    \begin{cases}
    O\!\left((k!s)^{1/k}N^{1-1/k}\right),
    & k!s \lesssim N, \\[4pt]
    O\!\left(\max\{ks,\,N(\log(N/\delta)+k)\}\right),
    & k!s \gtrsim N .
    \end{cases}
\end{equation}
The simulation tests this scaling by searching for the smallest sample budget
that achieves the prescribed empirical success probability.

For each value of $(N,k)$, we generate independent random ensembles
$\mathcal E_1,\mathcal E_2$ of size $N$.  The reference value
$D_k(\mathcal E_1,\mathcal E_2)$ is computed exactly from the empirical Gram
matrices when feasible, and by high-precision Monte Carlo sampling over label
pairs for large $N$.  For a candidate sample budget $M$, we run $T$ independent
measurement trials and estimate
\begin{equation}
    \widehat p(M)
    =
    \frac{1}{T}
    \sum_{t=1}^T
    \mathbf 1
    \left\{
        \left|\widehat D_k^{(t)} - D_k\right|
        \leq \epsilon
    \right\}.
\end{equation}
The budget $M$ is accepted if
\begin{equation}
    \widehat p(M) \geq 1-\delta .
\end{equation}
The accepted value is therefore an empirical upper bound on the required sample
complexity for that ensemble instance.

\begin{algorithm*}[t]
\caption{Simulation of the sample-complexity upper bound}
\label{alg:simulation-upper-bound}
\begin{algorithmic}[1]
\Require Ensemble size $N$, order $k$, accuracy $\epsilon$, failure probability $\delta$
\Require Number of ensemble repetitions $R$, number of estimator trials $T$
\State Set $s \gets \log(1/\delta)/\epsilon^2$
\For{$r=1,\ldots,R$}
    \State Generate independent ensembles $\mathcal E_1^{(r)},\mathcal E_2^{(r)}$ of size $N$
    \State Compute or estimate the reference value
    \[
        D_k^{(r)} = D_k(\mathcal E_1^{(r)},\mathcal E_2^{(r)}).
    \]
    \State Initialize an upper search point $M_{\mathrm{hi}}$ using the predicted scale:
    \[
        M_{\mathrm{hi}}
        \asymp
        \begin{cases}
        (k!s)^{1/k}N^{1-1/k},
        & k!s \leq N, \\[3pt]
        \max\{ks,\,N(\log(N/\delta)+k)\},
        & k!s > N .
        \end{cases}
    \]
    \While{$M_{\mathrm{hi}}$ is not accepted}
        \State $M_{\mathrm{hi}} \gets 2M_{\mathrm{hi}}$
    \EndWhile
    \State Set $M_{\mathrm{lo}}\gets 1$
    \While{$M_{\mathrm{hi}}-M_{\mathrm{lo}}$ is larger than the desired resolution}
        \State $M_{\mathrm{mid}}\gets \lfloor(M_{\mathrm{lo}}+M_{\mathrm{hi}})/2\rfloor$
        \State Run $T$ independent estimator trials with budget $M_{\mathrm{mid}}$
        \State Compute the empirical success probability $\widehat p(M_{\mathrm{mid}})$
        \If{$\widehat p(M_{\mathrm{mid}})\geq 1-\delta$}
            \State $M_{\mathrm{hi}}\gets M_{\mathrm{mid}}$
        \Else
            \State $M_{\mathrm{lo}}\gets M_{\mathrm{mid}}$
        \EndIf
    \EndWhile
    \State Record $M_{N,k}^{(r)}\gets M_{\mathrm{hi}}$
\EndFor
\State Report the mean, standard deviation, and quantiles of
$\{M_{N,k}^{(r)}\}_{r=1}^R$
\end{algorithmic}
\end{algorithm*}

The estimator trials in Algorithm~\ref{alg:simulation-upper-bound} use the same
measurement model as the protocol studied in the main text.  In the numerical
implementation, we do not simulate a full quantum circuit.  Instead, conditioned
on the sampled labels, the measurement outcomes are sampled directly from the
corresponding Born probabilities.  This gives an exact classical simulation of
the measurement statistics while avoiding circuit-level overhead.

The output of the simulation is a collection of empirical upper bounds
$M_{N,k}^{(r)}$.  Plotting $\log M$ against $\log N$ reveals two regimes.  In the
sparse-collision regime, the fitted slopes approach $1-1/k$, increasing with
$k$.  In the dense regime, repeated labels are abundant and the observed
dependence on $N$ becomes much weaker, producing slopes close to zero.

\subsubsection{Local SWAP-test based simulation}
As shown in Algorithm \ref{alg:estimate-M}, we estimated the minimum number of shots needed $M(N,\epsilon,\delta)$ by bisection method. The initial maximum value is the upper bound given by theorems in Section \ref{sec: sample complexity}, denoted by $M_{\text{hoeff}}$, to make the results robust, before bisection search, we test if $M_{\text{hoeff}}$ can achieve the additive error $\epsilon$ and the failure probability $\delta$, and adjust the maximum value if needed. The number of repetitions $K=30$, the number of trials $T=20$.

\begin{algorithm*}
\caption{Estimating the minimum sample $M(N,\epsilon,\delta)$ to estimate a distance metric between ensembles}
\label{alg:estimate-M}
\begin{algorithmic}

\State {\bfseries Input:} Ensemble samplers $\scE_1,\scE_2$; ensemble size $N$; additive error $\epsilon>0$; failure probability $\delta\in(0,1)$; repetitions $K$; number of trials $T$; metric/estimator $D(\cdot,\cdot)$ (e.g.\ Wasserstein or MMD-$k$ and the corresponding estimator); initial guess $M_{\mathrm{hoeff}}(N,\epsilon,\delta)$ given by upper bounds; maximum bracketing expansions $J_{\max}$.
\State {\bfseries Output:} Trial-wise estimates $\{M^{(t)}\}_{t=1}^T$ and an aggregate $\widehat{M}(N,\epsilon,\delta)$.

\For{$t=1$ {\bfseries to} $T$}
    \State Sample $\scS_1^{(t)}=\{\ket{\psi_1}^{(t)},\ldots,\ket{\psi_N}^{(t)}\}\overset{\mathrm{i.i.d.}}{\sim}\scE_1$
           and $\scS_2^{(t)}=\{\ket{\phi_1}^{(t)},\ldots,\ket{\phi_N}^{(t)}\}\overset{\mathrm{i.i.d.}}{\sim}\scE_2$ independently.
    \State Compute $D_{\mathrm{true}}^{(t)} \gets D(\scS_1^{(t)},\scS_2^{(t)})$.

    \State $lo \gets 0$, \;\; $hi \gets M_{\mathrm{hoeff}}(N,\epsilon,\delta)$.
    \State \hspace{0.2cm}{\bfseries (Robust bracketing to ensure $hi$ passes)}
    \State $j \gets 0$
    \Repeat
        \State $s \gets 0$
        \For{$r=1$ {\bfseries to} $K$}
            \State Using $hi$ shots, compute $\widehat{D}^{(t,r)}(hi)$ from $\scS_1^{(t)},\scS_2^{(t)}$.
            \If{$\widehat{D}^{(t,r)}(hi)$ is defined and $\big|\widehat{D}^{(t,r)}(hi)-D_{\mathrm{true}}^{(t)}\big|\le \epsilon$}
                \State $s \gets s+1$
            \EndIf
        \EndFor
        \State $p_{hi} \gets s/K$
        \If{$p_{hi} < 1-\delta$}
            \State $hi \gets 2hi$; \;\; $j \gets j+1$
        \EndIf
    \Until{$(p_{hi}\ge 1-\delta)$ {\bfseries or} $(j\ge J_{\max})$}

    \State \hspace{0.2cm}{\bfseries (Bisection with the same stopping rule as the implementation)}
    \While{$hi - lo > \max\{100,\lfloor hi/50 \rfloor\}$}
        \State $mid \gets \left\lfloor (lo + hi)/2 \right\rfloor$
        \If{$mid = 0$}
            \State $lo \gets mid$
            \State \textbf{continue}
        \EndIf

        \State $s \gets 0$
        \For{$r=1$ {\bfseries to} $K$}
            \State Using $mid$ shots, compute $\widehat{D}^{(t,r)}(mid)$ from $\scS_1^{(t)},\scS_2^{(t)}$.
            \If{$\widehat{D}^{(t,r)}(mid)$ is defined and $\big|\widehat{D}^{(t,r)}(mid)-D_{\mathrm{true}}^{(t)}\big|\le \epsilon$}
                \State $s \gets s+1$
            \EndIf
        \EndFor
        \State $p_{mid} \gets s/K$

        \If{$p_{mid} \ge 1-\delta$}
            \State $hi \gets mid$
        \Else
            \State $lo \gets mid$
        \EndIf
    \EndWhile

    \State Set $M^{(t)} \gets hi$.
\EndFor

\State Aggregate $\widehat{M}(N,\epsilon,\delta)$ from $\{M^{(t)}\}_{t=1}^T$ (mean and std).
\end{algorithmic}
\end{algorithm*}

\subsection{Details on QuDDPM and training}
\label{app:quddpm}

We summarize the schematic (Fig.~\ref{fig: QUDDPM_structure}) and the layerwise training protocol (Fig.~\ref{fig: QUDDPM_training}) of the quantum denoising diffusion probabilistic model (QuDDPM). Let $\mathcal{E}_0$ be an unknown distribution over $n$-qubit pure states on $\mathcal{V}\simeq(\mathbb{C}^2)^{\otimes n}$. The training dataset is an ensemble
\begin{equation*}
S_0=\bigl\{\ket{\psi_i^{(0)}}\bigr\}_{i=1}^{N}\sim \mathcal{E}_0,
\end{equation*}
and the objective is to learn a generative procedure that outputs samples $\ket{\tilde\psi^{(0)}}$ whose induced ensemble $\tilde{\mathcal{E}}_0$ matches $\mathcal{E}_0$. QuDDPM introduces a sequence of intermediate ensembles $\{S_k\}_{k=0}^{T}$ (forward diffusion) and $\{\tilde S_k\}_{k=0}^{T}$ (backward denoising), where $T$ is the number of diffusion steps shown in Fig.~\ref{fig: QUDDPM_structure}.

\begin{figure*}[tbp]
    \centering
    \includegraphics[width=0.7\linewidth]{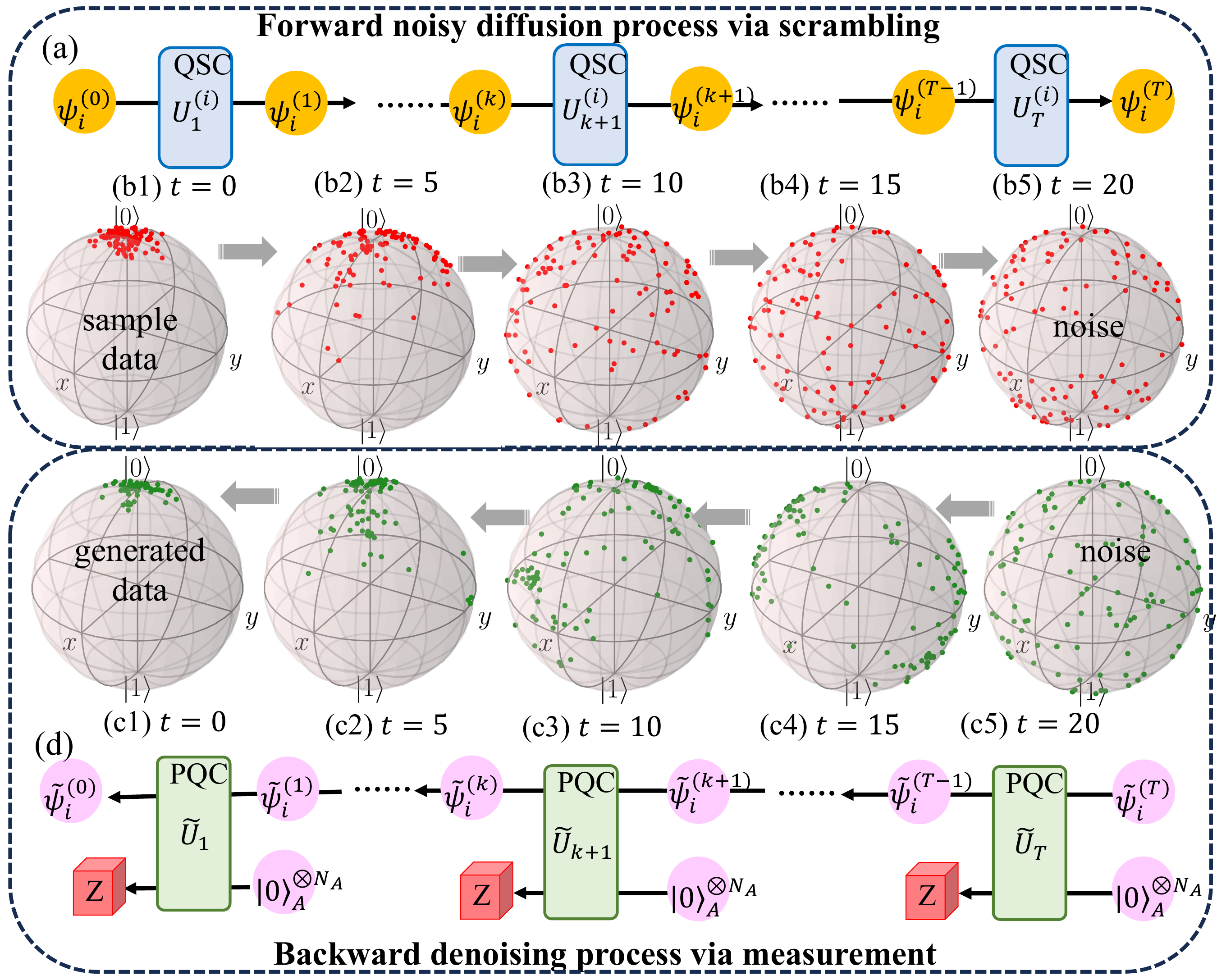}
    \caption{Schematic of QuDDPM. Reprinted from Ref.~\cite{QuDDPM_PhysRevLett.132.100602}. Copyright (2024) American Physical Society. Used with permission. The forward noisy process is implemented by a quantum scrambling circuit (QSC) in (a), while in the backward denoising process is achieved via measurement enabled by ancilla and parametrized quantum circuit (PQC) in (d). Subplots (b1)-(b5) and (c1)-(c5) present the Bloch sphere dynamics in generation of states clustering around $\ket{0}.$}
    \label{fig: QUDDPM_structure}
\end{figure*}

The forward (noisy) diffusion is implemented by applying random scrambling unitaries. For each training sample $\ket{\psi_i^{(0)}}$, draw a depth-$T$ scrambling circuit consisting of unitaries $\{U_\ell^{(i)}\}_{\ell=1}^{T}$, and define the diffused states by
\begin{equation}
\ket{\psi_i^{(k)}} := \Bigl(\prod_{\ell=1}^{k} U^{(i)}_{\ell}\Bigr)\ket{\psi_i^{(0)}},\qquad k=0,1,\dots,T.
\label{eq:app-quddpm-forward-update}
\end{equation}
The corresponding intermediate ensemble is $S_k:=\{\ket{\psi_i^{(k)}}\}_{i=1}^{N}$. As $k$ increases, the ensemble transitions from structured data ($k=0$) toward an (approximately) unstructured noise ensemble ($k=T$), as illustrated in Fig.~\ref{fig: QUDDPM_structure}(a--b).

The backward (denoising) process starts from a noise ensemble $\tilde S_T=\{\ket{\tilde\psi_i^{(T)}}\}_{i=1}^{\tilde N}$ and applies a sequence of parametrized quantum circuits (PQCs) $\{\tilde U_k(\theta_k)\}_{k=1}^{T}$, each acting on the $n$ system qubits together with $n_A$ ancilla qubits initialized in $\ket{0}^{\otimes n_A}$, followed by a projective measurement of the ancillas in the computational basis, as in Fig.~\ref{fig: QUDDPM_structure}(c--d). One denoising step can be written schematically as
\begin{equation}
\begin{gathered}
\ket{\tilde\psi_i^{(k-1)}}\ \leftarrow\ \textsf{Meas}_{A}\!\left[\tilde U_k(\theta_k)\bigl(\ket{\tilde\psi_i^{(k)}}\otimes\ket{0}^{\otimes n_A}\bigr)\right],
\\ k=T,T-1,\dots,1,
\end{gathered}
\label{eq:app-quddpm-denoising-update}
\end{equation}
which defines the intermediate denoising ensembles $\tilde S_k:=\{\ket{\tilde\psi_i^{(k)}}\}_{i=1}^{\tilde N}$ down to $\tilde S_0$. After training, generation proceeds by drawing $\ket{\tilde\psi^{(T)}}\sim \tilde S_T$ and applying the learned denoising steps sequentially from $k=T$ to $1$ to obtain a sample $\ket{\tilde\psi^{(0)}}$.

\begin{figure*}[tbp]
    \centering
    \includegraphics[width=0.6\linewidth]{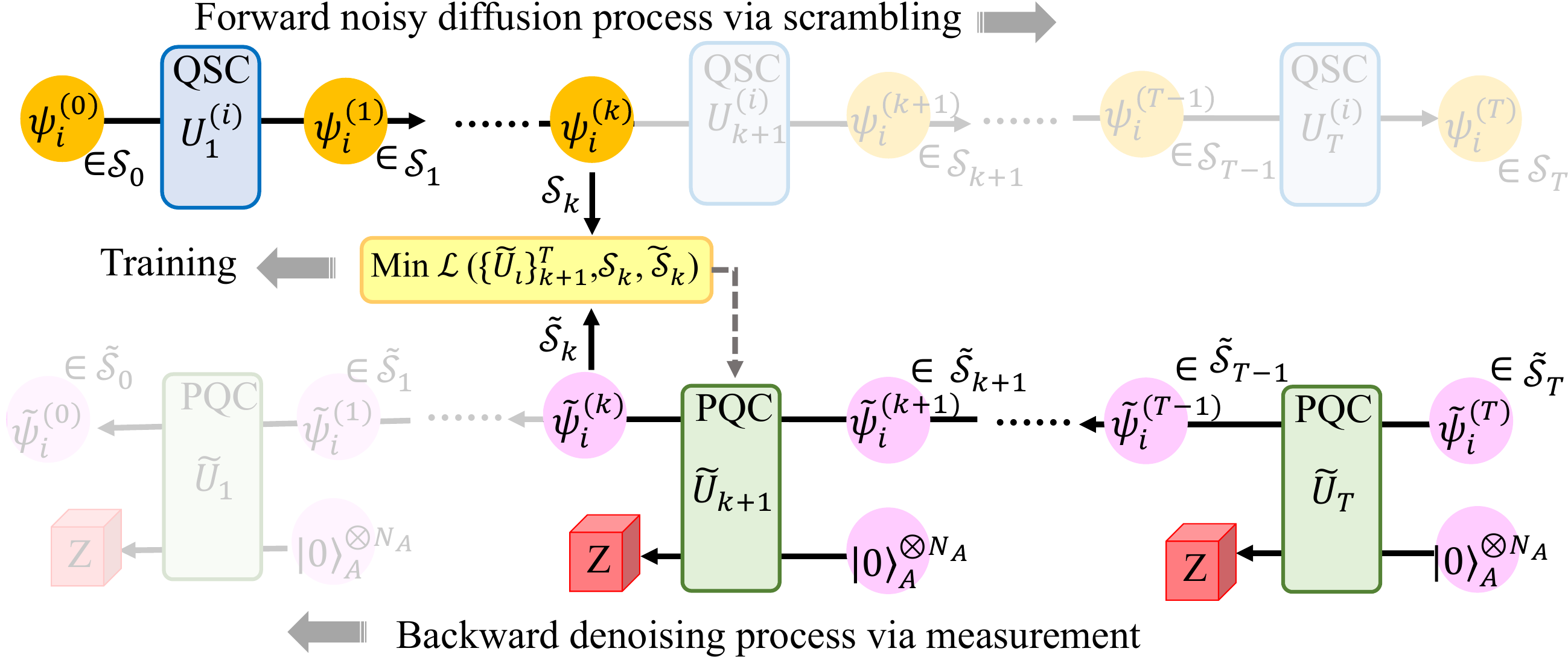}
    \caption{Training process of QuDDPM. Reprinted from Ref.~\cite{QuDDPM_PhysRevLett.132.100602}. Copyright (2024) American Physical Society. Used with permission. The training of QuDDPM at each step $t=k$. Pairwise  distance between states in generated ensemble $\tilde{\psi}_i^{(k)} \in \tilde{\mathcal{S}}_k$ and true  diffusion ensemble $\psi_i^{(k)} \in \mathcal{S}_k$ is measured and utilized in the evaluation of the loss function $\scL$.}
    \label{fig: QUDDPM_training}
\end{figure*}

Training is performed in $T$ cycles and proceeds from the earliest denoising block $\tilde U_T$ toward the last block $\tilde U_1$, as summarized in Fig.~\ref{fig: QUDDPM_training}. At the cycle corresponding to index $k$, the forward ensemble $S_k$ is prepared by applying the scrambling unitaries up to step $k$ on the dataset. In parallel, the current model ensemble $\tilde S_k$ is prepared by starting from $\tilde S_T$ and applying the already-trained denoising blocks $\tilde U_T,\tilde U_{T-1},\dots,\tilde U_{k+1}$ down to step $k$. The parameters of the next block $\tilde U_{k+1}$ are then updated while keeping all previously trained blocks fixed, with the objective of making $\tilde S_k$ approach $S_k$ at that diffusion depth. Repeating this layerwise procedure for $k=T-1,T-2,\dots,0$ yields a trained reverse process that maps the noise ensemble $\tilde S_T$ back to the data ensemble $S_0$.

In our demonstration of application of MMD-$2$, we choose $\scE_0$ to be the circular $1$-qubit state ensemble, and $\scS_0$ consists $N=1000$ states, number of ancilla qubits in each measurement $n_a=2$, number of steps $T=60$, number of layers in each step $L=6$. After training, we evaluate $\scL(\tilde{\scS_t}, \scE_0)$ across the step $t$, for ground diffusion data, training data, and testing data. As shown in Fig.~\ref{fig: loss_D2_W_after_training}, the values of both MMD-$2$ and Wasserstein of training and testing data match with the ground value well. It is noteworthy that $W(\tilde{\scS_0}, \scE_0) \approx 0$ for training and testing data, meaning QuDDPM almost completely learns to generate the circular state ensemble only using MMD-$2$ as loss function. The configuration and performance are summarized in Table~\ref{tab:quddpm-training-performance}.

\begin{figure*}[tbp]
    \centering
    \includegraphics[width=0.8\linewidth]{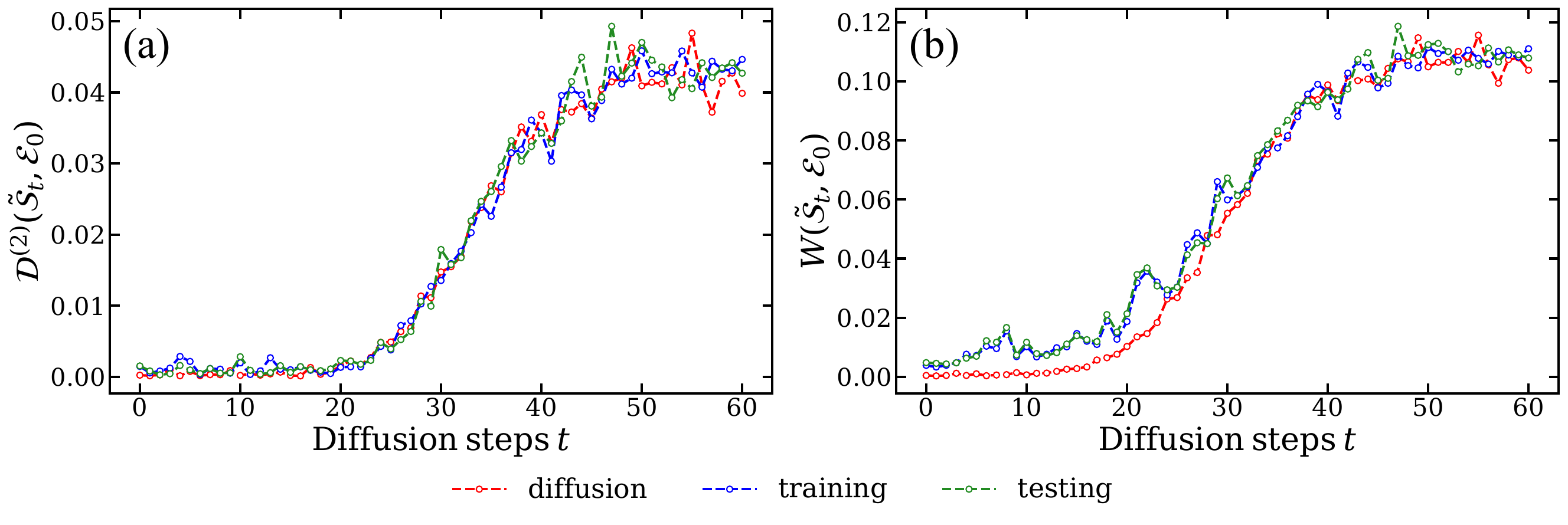}
    \caption{$\scL(\tilde{\scS_t}, \scE_0)$ across the step $t$, for ground diffusion data (red), training data (blue), and testing data (green), when $\scL$ is $\scD^{(2)}$ (a, MMD-$2$) and $W$ (b, Wasserstein). \label{fig: loss_D2_W_after_training}}
\end{figure*}
\begin{table*}[tbp]
    \centering
    \caption{Training configuration and final performance of the QuDDPM model trained with MMD-$2$.}
    \label{tab:quddpm-training-performance}
    \begingroup
    \small
    \setlength{\tabcolsep}{6pt}
    \renewcommand{\arraystretch}{1.15}
    \resizebox{\linewidth}{!}{
    \begin{tabular}{ccccccl}
        \toprule
        \(n\) & \(n_a\) & \(L\) & \(T\) & \(N\) & Cost & Performance \\
        \midrule
        \(1\) & \(2\) & \(6\) & \(60\) & \(1000\) & MMD-\(2\) &
        \(\langle Y_{\mathrm{data}}\rangle^2=0,\quad
        \langle Y_{\mathrm{train}}\rangle^2=0.00807\pm0.0337,\quad
        \langle Y_{\mathrm{test}}\rangle^2=0.01104\pm0.0496\) \\
        \bottomrule
    \end{tabular}
    }
    \par\vspace{0.5em}
    \endgroup
\end{table*}

In each training step, the number of iterations is $5000$, we use the Adam optimizer \cite{Adam_kingma2017adammethodstochasticoptimization} with learning rate $lr=5\times10^{-4}$, and in the final step of training ($t=0$, where the loss $\scL$ is near zero), we set $\scL'=\sqrt{\scL^2+10^{-8}}$ to enhance training stability. The training history for every step is shown in Fig.~\ref{fig:train_history}.

\begin{figure*}[tbp]
    \centering
    \includegraphics[width=0.8\linewidth]{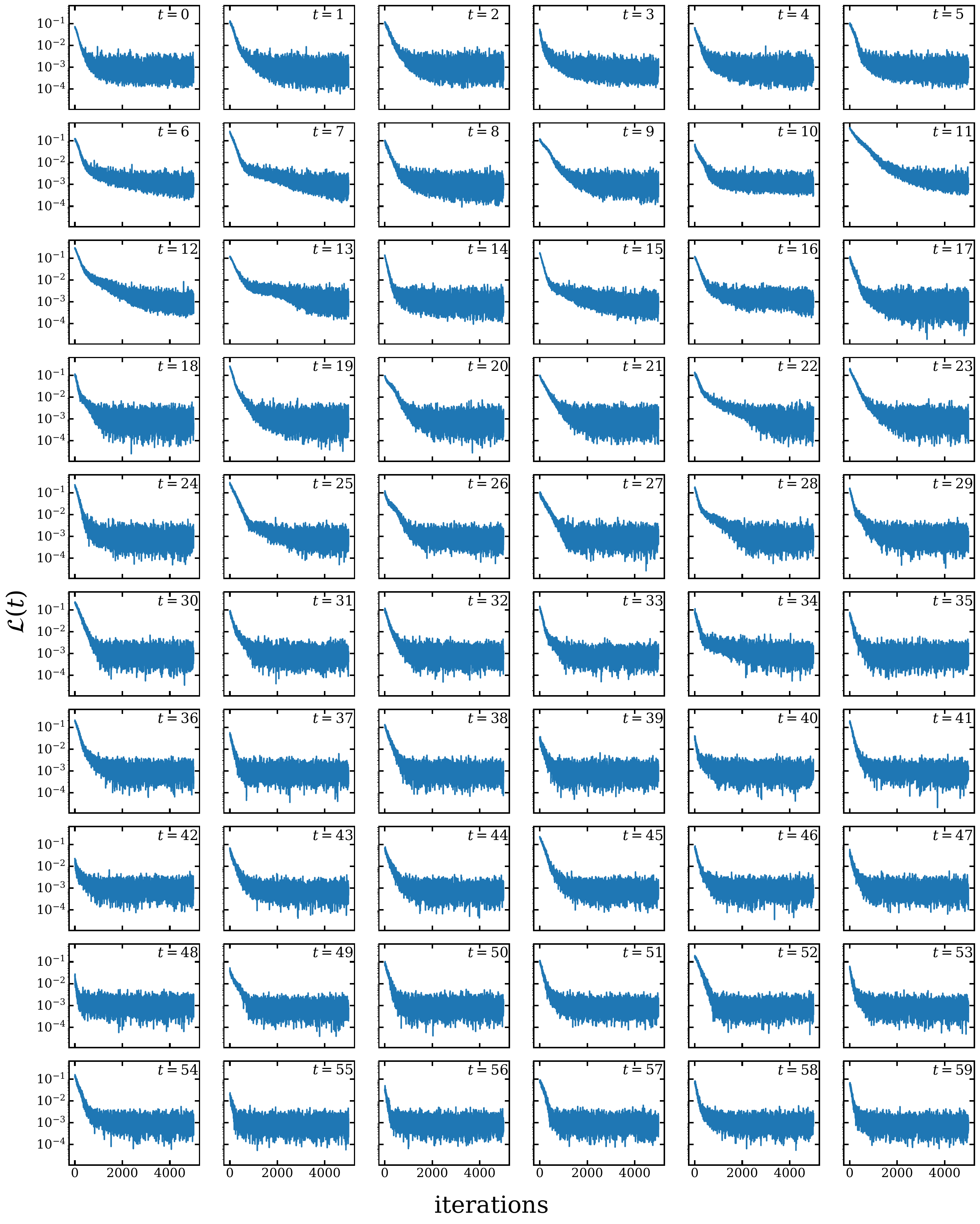}
    \caption{Training history for every step. \label{fig:train_history}}
\end{figure*}

\begin{figure*}[tbp]
    \centering
    \includegraphics[width=\linewidth]{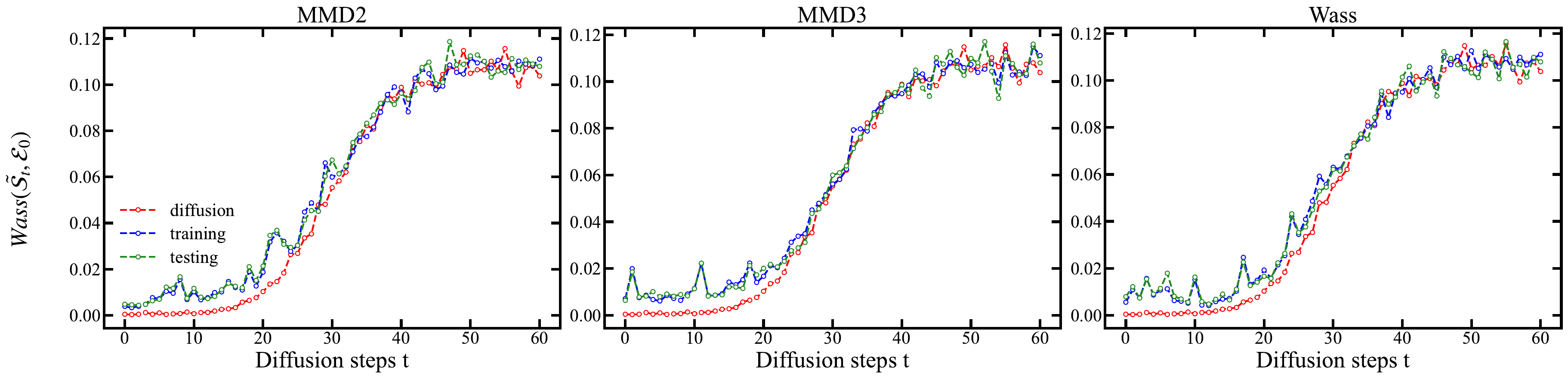}
    \caption{Performance of models trained by different loss functions evaluated by $W(\tilde{\scS_t}, \scE_0)$ across the step $t$ for ground diffusion data (red), training data (blue), and testing data (green).}
    \label{fig: performance-different-loss}
\end{figure*}

To preliminarily see whether loss functions with higher discriminative power will lead to better performance, we also train QuDDPM using MMD-3 and Wasserstein with the same hyperparameters. As shown in Fig.~\ref{fig: performance-different-loss}, there is no obvious difference among the performance of these three models. The complete relation between the discriminative power of loss function and the performance of model is still open.

Beyond the QuDDPM experiments, one can also give a simple analytical example separating
MMD-\(1\) from higher-order metrics. Let \(\scE_{\mathrm{Bell}}\) be the uniform ensemble over the
four Bell states
\[
    \ket{\Phi^\pm}
    =
    \frac{\ket{00}\pm\ket{11}}{\sqrt2},
    \qquad
    \ket{\Psi^\pm}
    =
    \frac{\ket{01}\pm\ket{10}}{\sqrt2},
\]
and let \(\scE_{\mathrm{comp}}\) be the uniform ensemble over the computational-basis states
\[
    \ket{00},\quad \ket{01},\quad \ket{10},\quad \ket{11}.
\]
Both ensembles have the same average state,
\[
    \frac14\sum_{\psi\in\scE_{\mathrm{Bell}}}\ketbra{\psi}{\psi}
    =
    \frac14\sum_{\phi\in\scE_{\mathrm{comp}}}\ketbra{\phi}{\phi}
    =
    \frac{I_4}{4},
\]
and therefore MMD-\(1\), which only compares the first moment, cannot distinguish them:
\begin{equation}
    \scD^{(1)}(\scE_{\mathrm{Bell}},\scE_{\mathrm{comp}})=0.
\label{eq:app-bell-computational-mmd-one}
\end{equation}
However, the two ensembles differ at the second-moment level. Since each ensemble is an
orthonormal basis, we have
\[
    \bar F^{(2)}(\scE_{\mathrm{Bell}},\scE_{\mathrm{Bell}})
    =
    \bar F^{(2)}(\scE_{\mathrm{comp}},\scE_{\mathrm{comp}})
    =
    \frac14.
\]
Moreover, every Bell state has squared overlap \(1/2\) with exactly two computational-basis
states and squared overlap \(0\) with the other two, so
\[
    \bar F^{(2)}(\scE_{\mathrm{Bell}},\scE_{\mathrm{comp}})
    =
    \frac18.
\]
Thus
\begin{equation}
    \scD^{(2)}(\scE_{\mathrm{Bell}},\scE_{\mathrm{comp}})
    =
    \frac14+\frac14-2\cdot\frac18
    =
    \frac14>0.
\label{eq:app-bell-computational-mmd-two}
\end{equation}
The Wasserstein distance is also strictly positive. With cost
\(C(\psi,\phi)=1-\abs{\braket{\psi}{\phi}}^2\), every Bell--computational pair has cost at least
\(1/2\), and this value is achievable by transporting each Bell state only to the two computational
basis states on which it has support. Hence
\begin{equation}
    W(\scE_{\mathrm{Bell}},\scE_{\mathrm{comp}})
    =
    \frac12.
\label{eq:app-bell-computational-wasserstein}
\end{equation}
This example shows explicitly that MMD-\(1\) can fail even when the ensembles are manifestly
different, while MMD-\(2\) and Wasserstein already detect the discrepancy.


\section{Additional sample-complexity analyses}
\label{app:additional analyses}

\subsection{Sample complexity based on local SWAP-test without quantum memory}
\label{app:local swap}

\subsubsection{Background of SWAP test}
\label{app:swap_details}

In this section, we adopt the standard approach of SWAP test, as we explain below.

\begin{figure}[tbp]
\centering
\resizebox{\linewidth}{!}{ %
\begin{quantikz}[row sep=0.5cm, column sep=0.55cm]
\lstick{$\ket{0}$}    & \gate{H} & \ctrl{1}  & \gate{H} & \gate{Z}  \\
\lstick{$\ket{\psi_i}$} & \qw      & \swap{1}  & \qw      & \qw      \\
\lstick{$\ket{\phi_j}$} & \qw      & \swap{-1} & \qw      & \qw   
\end{quantikz}%
}
\caption{The circuit of SWAP test. $H$ is the Hadamard gate, $Z$ represents Pauli-Z measurements in the computation basis, the connected dot and crosses represent a controlled-NOT gate, with the dot indicating the control qubit. }
\label{fig:swap-test}
\end{figure}

Below, we interpret the quantum circuit in Fig.~\ref{fig:swap-test}. As we introduce in Section~\ref{Sec: ensembles and measures}, an $n$-qubit quantum state is a $2^n$-dimensional vector. Quantum circuits are therefore described by unitary matrix acting on the vectors. For convenience, we express a quantum circuit via the circuit diagram. Common single-qubit gates include the Hadamard gate $\begin{quantikz}
&\gate{H}&
\end{quantikz} = \frac{1}{\sqrt{2}}
\begin{pmatrix}
1 & 1 \\
1 & -1
\end{pmatrix}$, the Pauli-X gate (the NOT gate) $\begin{quantikz}
&\gate{X}&
\end{quantikz} =
\begin{pmatrix}
0 & 1 \\
1 & 0
\end{pmatrix}$ and the Pauli-Z gate $\begin{quantikz}
&\gate{Z}&
\end{quantikz} =
\begin{pmatrix}
1 & 0 \\
0 & -1
\end{pmatrix}$. Besides single qubit gates, two-qubit gates are necessary for quantum computation. A common choice of the two-qubit gate is the controlled-NOT gate (CNOT) 
$
\begin{quantikz}
& \ctrl{1}  &    \\
& \targ{-1}  &              
\end{quantikz}
=\ketbra{0}{0}\otimes I+\ketbra{1}{1}\otimes X$. Another common two-qubit gate is the SWAP gate $\begin{quantikz}
&\swap{1}  & \\
&\swap{-1} &                
\end{quantikz}$ that exchanges the quantum state between the two systems, i.e. ${\rm SWAP} (\ket{\phi}\otimes\ket{\psi})=(\ket{\psi}\otimes\ket{\phi})$. The SWAP gate can be decomposed into three CNOT gates, represented by the circuit,  
\begin{equation}
\begin{quantikz}
&\swap{1}  & \\
&\swap{-1} &                
\end{quantikz}=\begin{quantikz}
 & \ctrl{1} & \targ{}  & \ctrl{1} & \qw \\
 & \targ{}  & \ctrl{-1}& \targ{}  & \qw
\end{quantikz}.
\end{equation}
The above is an example of how a quantum circuit diagram conveniently express the unitary transform on two qubits. Note that SWAP gate can be directly generalized to act on two systems, each with $n$ qubits, which can be realized by a combination of pairwise SWAP. In Fig.~\ref{fig:swap-test}, we also adopt a controlled-SWAP gate (Fredkin gate) for multiple qubits,
\begin{equation}
\begin{quantikz}[row sep=0.5cm, column sep=0.55cm]
& \ctrl{1}  &    \\
 \qw      & \swap{1}  & \qw           \\
 \qw      & \swap{-1} & \qw         
\end{quantikz}
=\ketbra{0}{0}\otimes I + \ketbra{1}{1}\otimes {\rm SWAP},
\end{equation}
which SWAP the states of the two quantum systems only when the control qubit is in $\ket{1}$ state. For more details on the SWAP test quantum circuits, readers can refer to, e.g. a recent review on SWAP test~\cite{nishimura2025survey}.

\subsubsection{SWAP test for sample distance estimation}

Given two quantum ensembles $\scE_1$ and $\scE_2$ with $N$ states each, suppose each time we can randomly obtain one state $\ket{\psi_i}\sim\scE_1$ and the other state $\ket{\phi_j}\sim\scE_2$, while knowing the index $i$ and $j$ (we use $\ell$ to denote $(i,j)$ to lighten the notation). To access the information of their fidelity $X_{\ell}=\abs{\braket{\psi_i}{\phi_j}}^2$, we can use SWAP test (see Fig.~\ref{fig:swap-test}). The probability of obtaining $0$ in the Pauli-Z measurement on the ancilla qubit of the output state is
\begin{equation}
    \Pr[\text{ancilla qubit in} ~\ket{0}] = \frac{1}{2}(1+\abs{\braket{\psi_i}{\phi_j}}^2) =\frac{1}{2}(1+X_{\ell}).
\end{equation}

We denote the SWAP-test outcome by a Bernoulli random variable $Y_\ell\in\{0,1\}$,
where $Y_\ell=1$ if the ancilla is measured in $\ket{0}$ and $Y_\ell=0$ otherwise.
Define $R_\ell:=2Y_\ell-1\in\{-1,+1\}$, so that $R_\ell=+1$ when obtaining $0$
and $R_\ell=-1$ when obtaining $1$ on the ancilla qubit. Then
$\mathbb E[R_\ell\mid \ell]=|\langle\psi_i|\phi_j\rangle|^2=X_\ell$. We write the data $(R_{\ell},\ell)$ as one sample, and the two states collapse after SWAP test.

Suppose we repeat the sampling process $M$ times and then we have $M$ samples $\{(R_{\ell_t},\ell_t)_{t=1}^M\}$, which can used to estimate some distance metrics of the two ensembles. However, due to the destructive and probabilistic nature of quantum measurements, we need a certain number of samples $s$ on the same label $\ell$ to estimate the value of $X_{\ell}$ accurately. However, for each sample, we can not deterministically control which label to be sampled, so the total number of samples to reach $s$ on the same label is roughly $N^2s$, which is expensive since $N$ is usually very large. For example, to estimate Wasserstein distance, we need to estimate $X_{\ell}$ for every $\ell \in \{1,...,N^2\}$ accurately, so the number of samples needed is very large. However, for MMD-$1$, since we have 


\begin{equation}
    \bar{F}^{(1)}(\scE_a, \scE_b) = \dsE_{\ell}[X_{\ell}] = \dsE_{\ell}[\dsE[R_{\ell_t} \mid \ell_t=\ell]] =\dsE_{\ell_t}[R_{\ell_t}], \label{eq: estimation of MMD-1}
\end{equation}
where $\dsE[R_{\ell_t} \mid \ell_t=\ell]$ is the expectation of $R_{\ell_t}$ only with $\ell_t=\ell$, gives $X_\ell$, and $\dsE_{\ell_t}[R_{\ell_t}]$ is the expectation over all $R_{\ell_t}$ obtained. Eq. (\ref{eq: estimation of MMD-1}) means that we can ignore the information of labels when estimating MMD-$1$, avoiding the overhead of sampling. As discussed in Section \ref{Sec: ensembles and measures}, Wasserstein has the full discriminative power and MMD-$1$ only has $\scP_M(\scD^{(1)})=1$. Hence, there may exist a tradeoff between discriminative power and sample complexity for distance metrics, which will be discussed in later sections.

We can define a limited version of sample complexity merely based on SWAP test.
\begin{definition}[\textbf{SWAP-test Sample complexity}]
    Given a distance metric $\scD$ and two quantum ensembles $\scE_1$ and $\scE_2$ each consisting of $N$ states, suppose that we can randomly assess two states from the two ensembles (either from the same or the different ensembles) and conduct a SWAP test on the two states, obtaining one sample $(R_{\ell},\ell)$, where $\ell$ is the label representing $(i,j)$, the indices of two states. Then the sample complexity $M_{\rm SWAP}(N,\epsilon,\delta)$ is the number of samples needed to estimate $\scD(\scE_1, \scE_2)$ for some estimator using $\{(R_{\ell_t},\ell_t)_{t=1}^M\}$, up to an additive error $\epsilon$ and failure probability $\delta$.\label{def: sample complexity}
\end{definition}

A key difference between Definition \ref{def: sample complexity} and the general one Definition \ref{def: sample complexity general} is that in \ref{def: sample complexity} we do not have quantum memory.

\subsubsection{Estimating process and sample complexity of Wasserstein distance}\label{secsub: sc of Wass}
Here, we examine the SWAP-test sample complexity of Wasserstein distance. To estimate the Wasserstein distance $W(\scE_1,\scE_2)$ of two ensembles $\scE_1$ and $\scE_2$, we conduct experiment to obtain $M$ samples $\{(R_{\ell_t},\ell_t)_{t=1}^M\}$. Suppose $M$ is large enough so that every label has samples. Then we can estimate the fidelity of all pairs of states using $\widehat{X_{\ell}}=\frac{1}{T_{\ell}}\sum_{\ell_t=\ell}R_{\ell_t}$ and then obtain the estimated cost matrix $\widehat{C}$, which works as input parameter in the optimization problem Eq. (\ref{eq: calculation of Wasserstein}) to obtain the estimated value of Wasserstein distance $\widehat{W}$.
\begin{theorem}[\textbf{SWAP-test Sample complexity of Wasserstein distance}]
The sample complexity to estimate the Wasserstein distance between two $N$-state uniform pure quantum ensembles, using the above SWAP-test based estimation process is
    \begin{equation}
        M_{\rm SWAP}=\scO\Big(\frac{N^2}{\epsilon^2}\log\frac{N^2}{\delta}\Big), \label{eq: sample compelxity wasserstein}
    \end{equation}
    where $\epsilon$ is the additive error between estimated value and true value and $\delta$ is the failure probability.\label{thm: sample complexity Wasserstein}
\end{theorem}
The proof is in Appendix~\ref{appsub: proof of Wasserstein upper bound}. An implicit lower bound of $M$ is the number of samples needed to make all labels $\ell$ have at least one sample to give an estimated value of $X_\ell$, given by $M_{\rm SWAP}=\Omega(N^2\log N)$ (see Appendix~\ref{appsub: minM to reach T}). Then we have $M_{\rm SWAP}=\Theta(N^2 \log N)$.

\subsection{Estimating process and sample complexity of MMD-$k$}
To give an estimated value $\widehat{\scD^{(k)}(\scE_1,\scE_2)}$ of $\scD^{(k)}(\scE_1,\scE_2)$, the key is to obtain $\widehat{\bar{F}^{(k)}(\scE_a, \scE_b)}$ for $(a,b)=(1,1),(1,2),(2,2)$. Suppose for each $(a,b)$ we have $M/3$ samples. To estimate $\bar{F}^{(k)}=\dsE[X^k_{\ell}]$, we first utilize the kernel from Unbiased-statistics (U-stat) \cite{U-stat_10.1214/aoms/1177730196} to estimate $X_{\ell}^k$:
\begin{equation}
    Z_{\ell} = \frac{1}{\binom{T_{\ell}}{k}}\sum_{1\le r_1 <r_2<...<r_k \le T_{\ell}} \prod_{s=1}^k R_{\ell_{(r_s)}}, \label{eq: U-stat for MMD-k}
\end{equation}
where $T_{\ell}$ denotes the number of samples on the same label $\ell$, and $\ell_{(r_s)}$ is rewritten from $\ell_t$ to index the data of the same label. Notice that Eq.(\ref{eq: U-stat for MMD-k}) can be calculated only for $T_{\ell} \ge k$. We have $\dsE[Z_{\ell}]=X_{\ell}^k$, the estimator for $\bar{F}^{(k)}$ we use is:
\begin{equation}
    \widehat{\bar{F}^{(k)}}=\frac{1}{m}\sum_{\ell:T_{\ell}\ge k}Z_{\ell},\label{eq: estimator of MMD-k}
\end{equation}
where $m=\sum_{\ell=1}^{N^2} \mathbf{1}\{T_{\ell} \ge k\}$, representing the number of labels with no less than $k$ samples. Below, we analyze the sample complexity and optimality of the above estimator.

We begin with the sample complexity to estimate MMD-$k$ (with proof in Appendix~\ref{appsub: MMD-k upper bound}). For simplicity, we focus on the common case of large ensembles $N^2 \ge \frac{k!}{\epsilon^2}\log\frac{1}{\delta}$, while the general case can be found in Theorem~\ref{theorem:general_version}.

\begin{theorem}[\textbf{SWAP-test Sample complexity of MMD-$k$ with fixed $k$}]
    For two $N$-state uniform pure quantum ensembles, with $k$ fixed, the SWAP-test sample complexity to estimate MMD-$k$ between the two ensembles using the U-stat estimator Eq. (\ref{eq: estimator of MMD-k}) is:
    \begin{equation}
        M_{\rm SWAP} = \scO\Big((\frac{k!}{\epsilon^2}\log \frac{1}{\delta})^{\frac{1}{k}}N^{2-\frac{2}{k}}\Big). \label{eq: sample complexity of MMD-k, upper bound for most case}
    \end{equation}\label{thm: sample complexity of MMD-k, upper bound, most general case}
\end{theorem}

Theorem~\ref{thm: sample complexity of MMD-k, upper bound, most general case} gives the upper bound of sample complexity to estimate MMD-$k$. Importantly, this scaling is information-theoretically optimal, as shown by the matching lower bound below (with proof in Appendix~\ref{appsub: MMD-k proof of lower bound}).

\begin{theorem} [\textbf{SWAP-test Sample complexity of MMD-$k$, lower bound}]
    There exist two $N$-state uniform pure quantum ensembles $\scE_1$ and $\scE_2$, such that the sample complexity to estimate MMD-$k$ between the two ensembles using \textbf{any estimator} is lower bounded by:
    \begin{equation}
        M_{\rm SWAP} = \Omega\Big((\frac{k!}{\epsilon^2}\log \frac{1}{\delta})^{\frac{1}{k}}N^{2-\frac{2}{k}}\Big), \label{eq: M of MMD-k, lower bound}
    \end{equation}
    where $\epsilon$ is the additive error between estimated value and true value and $\delta$ is the failure probability.\label{thm: lower bound of sample complexity MMD-k}
\end{theorem}
Combining Theorem~\ref{thm: sample complexity of MMD-k, upper bound, most general case} and Theorem~\ref{thm: lower bound of sample complexity MMD-k}, we conclude that to guarantee the estimation of MMD-$k$, the optimal scaling for sample complexity $M$ with $N$ is that $M \sim N^{2-\frac{2}{k}}$, and it can be achieved by the U-stat estimator Eq.~(\ref{eq: estimator of MMD-k}) we introduce.

\begin{figure*}[tbp]
    \centering
    \includegraphics[width=0.85\linewidth]{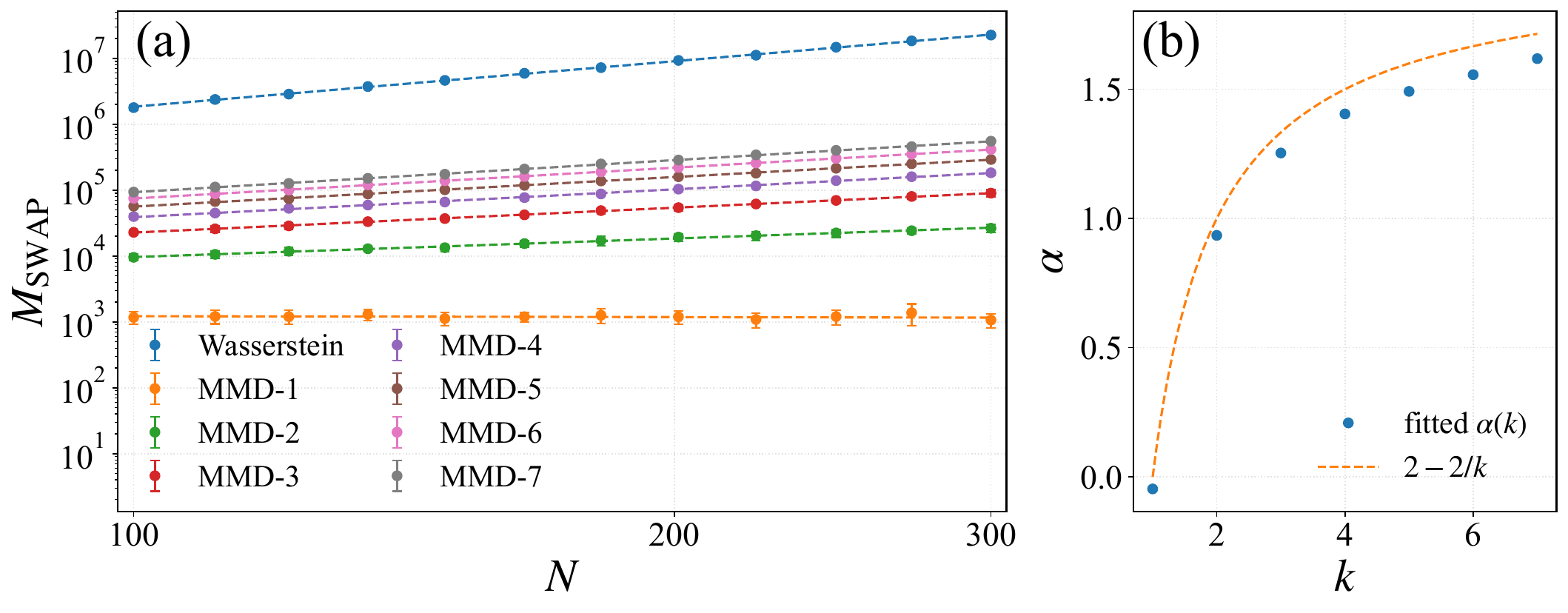}
    \caption{Numerical simulation of sample complexity. (a) Number of samples needed to estimate MMD-$k$ and Wasserstein as $N$ increases. Both axes are in logarithmic scale. (b) Scaling coefficient of MMD-$k$ with different values of $k$. The corresponding slope for Wasserstein is $2.29>2$ due to finite-size effects from the extra $\log(N)$ in Eq.~\eqref{eq: sample compelxity wasserstein}.    }
    \label{fig: numerical sample complexity}
\end{figure*}

As shown in Lemma~\ref{thm: k to reach full disriminative power}, sometimes we consider the case when $k$ also scales with $N$, such as $k\sim N$ (with proof in Appendix~\ref{appsub: MMD-k k large}).

\begin{theorem}[\textbf{SWAP-test Sample Complexity of MMD-$k$ with $k\sim N$}]
    When $k$ has the scaling as $k\sim N$.
    Consider two $N$-state uniform pure quantum ensembles, the SWAP-test sample complexity to estimate MMD-$k$ between them using the U-stat estimator Eq. (\ref{eq: estimator of MMD-k}) has the scaling
    \begin{equation}
        M_{\rm SWAP} = \Theta\Big(N^2k\Big).
    \end{equation}
        \label{thm: sample complexity of MMD-k, k large}
\end{theorem}
Theorem~\ref{thm: sample complexity of MMD-k, k large} shows that when MMD-$k$ achieves the full discriminative power (given by Lemma~\ref{thm: k to reach full disriminative power}), the SWAP-test sample complexity scales as $M_{\rm SWAP}\sim N^3$, which is faster than the scaling $M\sim N^2 \log N$ that Wasserstein needs. It is reasonable since MMD-$k$ and Wasserstein process the same data $\{(R_{\ell_t},\ell_t)_{t=1}^M\}$ differently and reveal different aspects of the information in data. We summarize the relationship between sample complexity and discriminative power for MMD-$k$ and Wasserstein in Fig.~\ref{fig: sc and dp}.

\subsubsection{Numerical simulation for the sample complexity}

We perform similar simulation as in Section~\ref{sec: simulation complexity}, except that here we perform SWAP test based estimation without quantum memory. The algorithm to calculate $M_{\rm SWAP}(N,\epsilon,\delta)$ for a distance metric $\scD$ and two ensembles $S_1$ and $S_2$ is Algorithm \ref{alg:simulation-upper-bound} in Appendix \ref{app: algorithms}.

With a fixed additive error $\epsilon=0.1$ and failure probability $\delta=1/3$, we evaluate the sample complexity $M$ to estimate MMD-$k$ and the Wasserstein distance as $N$ increases from $100$ to $300$, with $k \in \{1,\ldots,7\}$. To examine the scaling with $N$, we plot $\log M$ versus $\log N$ in Fig.~\ref{fig: numerical sample complexity}(a) and fit the data points to extract the slopes. The resulting slopes for estimating MMD-$k$ at different $k$ are shown in Fig.~\ref{fig: numerical sample complexity}(b), and compared with the predicted upper bound $2-2/k$ from Eq.~(\ref{eq: sample complexity of MMD-k, upper bound for most case}).

\subsection{Generalization to weakly noisy mixed states}
\label{app: mixed states} 

In this section, we consider a practical scenario where quantum state ensembles are subject to noise. For simplicity, we focus on the weak noise limit where the noisy mixed states are close to pure. We begin by introducing the $\epsilon$-ball model.

\begin{definition}[\textbf{The $\epsilon$-ball around a quantum state}]
    Given a pure quantum state $\ket{\widetilde{\psi}}$, the $\epsilon$-ball with radius $\epsilon_b$ around $\ket{\widetilde{\psi}}$ is the set of states
    \begin{equation}
        \scB(\ket{\widetilde{\psi}},\epsilon_b)
        =
        \{\ket{\psi}:1-\abs{\braket{\psi}{\widetilde{\psi}}}^2\le\epsilon_b\},
    \end{equation}
    and $\ket{\widetilde{\psi}}$ is the central state of $\scB(\ket{\widetilde{\psi}},\epsilon_b)$.\label{def: epsilon ball}
\end{definition}

Throughout this section, when we say that an ensemble is uniformly divided into $n$ $\epsilon$-balls, we mean that the ball label of a uniformly sampled state is uniform on $[n]$. We consider two protocols. The \emph{across-label SWAP protocol} stores a batch of quantum samples, forms row and column $k$-collision blocks according to the ball labels, and performs multi-register SWAP measurements across different labels. The \emph{label-local SWAP protocol} performs an ordinary SWAP test immediately on each sampled pair and then applies the U-statistic estimator Eq.~(\ref{eq: estimator of MMD-k}) to samples with the same pair-of-ball label.

Based on the above definition, we have the following two sample-complexity bounds.
\begin{theorem}[\textbf{Sample complexity of MMD-$k$, given prior information}]
    \label{thm: sample complexity MMD-k, prior information}
    Fix $k\ge2$. Consider two $N$-state uniform pure quantum ensembles $\scE_1$ and $\scE_2$. Suppose states in each ensemble can be divided uniformly into $n$ $\epsilon$-balls with radius $\epsilon_b$, and we are given the prior information about which $\epsilon$-ball every sampled state falls into. Let
    \begin{equation}
        \epsilon_{\mathrm{eff}}
        :=
        \epsilon-32k\sqrt{\epsilon_b}.
    \end{equation}
    Assume $\epsilon_{\mathrm{eff}}>0$. Then, in the sparse-collision regimes, the sample complexities to estimate MMD-$k$ up to additive error $\epsilon$ and failure probability $\delta$ are:
    \begin{align}
        M_{\mathrm{across}}
        &=
        \scO\Bigg(
        \bigg(
            \frac{k!}{\epsilon_{\mathrm{eff}}^2}
            \log\frac{1}{\delta}
        \bigg)^{\frac1k}
        n^{1-\frac1k}
        \Bigg),
        \qquad
        \frac{k!}{\epsilon_{\mathrm{eff}}^2}\log\frac1\delta \lesssim n,
        \label{eq: M of MMD-k, epsilon ball, across upper bound}\\
        M_{\mathrm{local}}
        &=
        \scO\Bigg(
        \bigg(
            \frac{k!}{\epsilon_{\mathrm{eff}}^2}
            \log\frac{1}{\delta}
        \bigg)^{\frac1k}
        n^{2-\frac2k}
        \Bigg),
        \qquad
        \frac{k!}{\epsilon_{\mathrm{eff}}^2}\log\frac1\delta \lesssim n^2 .
        \label{eq: M of MMD-k, epsilon ball, local upper bound}
    \end{align}
    Here the hidden constants may depend on $k$. The first bound is achieved by the across-label SWAP protocol, and the second bound is achieved by the label-local SWAP protocol.
\end{theorem}

The major insight is that the across-label protocol uses the $n$ row-ball labels and the $n$ column-ball labels separately, so its collision cost scales as $n^{1-1/k}$. By contrast, the label-local protocol groups samples by the full pair-of-ball label, of which there are $n^2$, so its collision cost scales as $(n^2)^{1-1/k}=n^{2-2/k}$. The proofs of the two parts are given in Appendix~\ref{appsub: MMD-k prior information across-label} and Appendix~\ref{appsub: MMD-k prior information label-local}.

Theorem~\ref{thm: sample complexity MMD-k, prior information} also allows one to generalize our results for pure-state ensembles perturbed by small noise. Specifically, suppose there is a perturbation noise represented by the depolarizing channel
\begin{equation}
    \mathcal N_{\lambda_b}(\rho) = (1-\lambda_b)\rho+\lambda_b\frac{I}{d},
    \label{eq: depolarizing channel}
\end{equation}
where $I$ is the identity matrix and $d$ is the dimension of the quantum system. For an input pure state $\ket{\phi}$, let $\ket{\xi}$ be Haar-random in the orthogonal complement of $\ket{\phi}$ and define
\begin{equation}
    \ket{\phi_\xi}
    =
    \sqrt{1-\epsilon_b}\ket{\phi}
    +
    \sqrt{\epsilon_b}\ket{\xi},
    \qquad
    \epsilon_b=(1-1/d)\lambda_b .
\end{equation}
Then
\begin{equation}
\begin{aligned}
\dsE_\xi\!\left[\ketbra{\phi_\xi}{\phi_\xi}\right]
    &=
    (1-\epsilon_b)\ketbra{\phi}{\phi}
    +
    \epsilon_b\frac{I-\ketbra{\phi}{\phi}}{d-1}\\
    &=
    (1-\lambda_b)\ketbra{\phi}{\phi}
    +
    \lambda_b\frac{I}{d}.
\end{aligned}
\end{equation}
Therefore, the output of the depolarizing channel admits a pure-state decomposition supported on an $\epsilon$-ball of radius $\epsilon_b=(1-1/d)\lambda_b$ around the input pure state.

\begin{theorem}[\textbf{Sample complexity of MMD-$k$, close-to-pure mixed states}]
    \label{thm: sample complexity MMD-k, mixed state}
    Fix $k\ge2$. Consider two $N$-state uniform pure quantum ensembles. Suppose both ensembles go through the depolarizing channel Eq.~(\ref{eq: depolarizing channel}). Let
    \begin{equation}
        \epsilon_{\mathrm{mix}}
        :=
        \epsilon-32k\sqrt{(1-1/d)\lambda_b}.
    \end{equation}
    Assume $\epsilon_{\mathrm{mix}}>0$. Then, in the sparse-collision regimes, the sample complexities to estimate MMD-$k$ between the output ensembles up to additive error $\epsilon$ and failure probability $\delta$ are:
    \begin{align}
M_{\mathrm{across}}
        &=
        \scO\Bigg(
        \bigg(
            \frac{k!}{\epsilon_{\mathrm{mix}}^2}
            \log\frac{1}{\delta}
        \bigg)^{\frac1k}
        N^{1-\frac1k}
        \Bigg),
        \notag\\
        &\qquad
        \frac{k!}{\epsilon_{\mathrm{mix}}^2}\log\frac1\delta \lesssim N,
        \label{eq: M of MMD-k, mixed state, across upper bound}\\
        M_{\mathrm{local}}
        &=
        \scO\Bigg(
        \bigg(
            \frac{k!}{\epsilon_{\mathrm{mix}}^2}
            \log\frac{1}{\delta}
        \bigg)^{\frac1k}
        N^{2-\frac2k}
        \Bigg),
        \notag\\
        &\qquad
        \frac{k!}{\epsilon_{\mathrm{mix}}^2}\log\frac1\delta \lesssim N^2 .
        \label{eq: M of MMD-k, mixed state, local upper bound}
\end{align}
    Here the hidden constants may depend on $k$. The first bound is achieved by the across-label SWAP protocol, and the second bound is achieved by the label-local SWAP protocol.
\end{theorem}

Theorem~\ref{thm: sample complexity MMD-k, mixed state} is obtained from Theorem~\ref{thm: sample complexity MMD-k, prior information} by taking one $\epsilon$-ball around each original pure state, so that $n=N$ and $\epsilon_b=(1-1/d)\lambda_b$.

\section{Proofs for Appendix~\ref{app:additional analyses}
}

\subsection{Proof of Theorem \ref{thm: sample complexity Wasserstein}: label-local-swap-test Wasserstein upper bound}\label{appsub: proof of Wasserstein upper bound}
We first introduce a lemma to be used.
\begin{lemma}\label{lm: maxCij}
A sufficient condition for $\abs{\widehat{W}-W(C)}\leq \epsilon$ is $\max_{i,j} \abs{\widehat{C}_{ij}-C_{ij}} \leq \epsilon$. 
\end{lemma}
Here, $\widehat{W} = W(\widehat{C})$ represents the Wasserstein distance calculated by program (\ref{eq: calculation of Wasserstein}) using the estimated value $\widehat{C}$ of the population cost matrix $C$ and $W(C) = W(\scE_a,\scE_b)$ denotes the population Wasserstein distance.
\begin{proof}
    Under this condition, we may decompose $\widehat{C} = C+\Delta C$ with $\norm{\Delta C}_\infty \leq \epsilon$. So $\abs{\langle P,\Delta C\rangle} \leq \norm{P}\norm{\Delta C}_\infty \leq  \epsilon$ for any $P$. Then
    \begin{equation*}
    \begin{aligned}
    W(\widehat{C})
    &= \langle \widehat{P}, \widehat{C}\rangle 
    = \langle \widehat{P}, C\rangle + \langle \widehat{P}, \Delta C\rangle \\
    &\ge \min_{P}\langle P, C\rangle - \epsilon
    = W(C)-\epsilon,\\
    W(C)
    &= \langle P^{*}, C\rangle = \langle P^{*}, \widehat{C}\rangle - \langle P^{*}, \Delta C\rangle \\
    &\ge \min_{P}\langle P, \widehat{C}\rangle - \epsilon
    = W(\widehat{C})-\epsilon ,
    \end{aligned}
    \end{equation*}
where $P^*$ is the population optimizer for Wasserstein distance $W(C)$ using population cost matrix $C$. Namely,  \begin{equation}\abs{W(\widehat{C})-W(C)}\leq \epsilon.
\label{eq:app-wasserstein-cost-error}
\end{equation}
\end{proof}

\begin{proof}[Proof of Theorem \ref{thm: sample complexity Wasserstein}]
Recall that we use $\ell$ to represent label $(i,j)$. When estimating the fidelity $X_{\ell}$, suppose that we have access to $T_{\ell}$-many repetitions of the corresponding pair of states, and conduct $T_{\ell}$ SWAP tests to obtain $T_{\ell}$ samples. Hoeffding's inequality gives that for any additive error $\epsilon>0$,
\begin{equation}\label{eq:app-local-swap-hoeffding-bound}
    \Pr(\abs{\widehat{X_{\ell}}-X_{\ell}}\geq \epsilon) \le 2 \exp{(-T_{\ell}\epsilon^2/2)}.
\end{equation}
Since $\widehat{C_{\ell}} = 1-\widehat{X_{\ell}}$, we have
\begin{equation}
    \Pr(\abs{\widehat{C_{\ell}}-C_{\ell}}\ge \epsilon) \leq 2 \exp{(-T_{\ell}\epsilon^2/2)}.
    \label{eq: HoeffdingC_{ell}}
\end{equation}
Now we consider the condition shown in Lemma \ref{lm: maxCij},
\begin{align}
\label{eq: hoeffding for all labels}
&
    \Pr \left( \max_{\ell} \Abs{\widehat{C_{\ell}}-C_{\ell}}\ge\epsilon \right) \nonumber\\
    &= \Pr\left( \bigcup_{\ell=1}^{N^2} \{\abs{\widehat{C_{\ell}}-C_{\ell}} \ge \epsilon \} \right)\nonumber\\
    &\le \sum_{\ell=1}^{N^2} \Pr\left( \Abs{\widehat{C_{\ell}}-C_{\ell}} \ge \epsilon \right) \nonumber \\ 
    &\leq 2N^2 \exp{ \left( -\frac{T_{\ell}\epsilon^2}{2} \right)} \nonumber\\
    &\leq 2N^2\exp{ \left( -\frac{T_{\min}\epsilon^2}{2} \right)},
\end{align}
where $T_{\min} = \min_{\ell} T_{\ell}$. 
From Eq. (\ref{eq: hoeffding for all labels}), to make $\Pr \left( \max_{\ell} \Abs{\widehat{C_{\ell}}-C_{\ell}}>\epsilon \right) \leq \delta$ we need \begin{equation}T_{\min} \ge \frac{2}{\epsilon^2}\log{\frac{2N^2}{\delta}}.
\label{eq:app-wasserstein-minimum-pair-count}
\end{equation} Eq.~(\ref{eq: M needed to achieve Tmin}) gives that the number of samples needed is \begin{equation}M=\scO\Big( N^2(\frac{2}{\epsilon^2}\log{\frac{2N^2}{\delta}}+\log\frac{N^2}{\delta}) \Big).
\label{eq:app-wasserstein-local-total-samples}
\end{equation} 


\end{proof}

\subsection{Proof of Theorem \ref{thm: sample complexity of MMD-k, upper bound, most general case}: label-local-swap-test MMD-$k$ upper bound }\label{appsub: MMD-k upper bound}

To help us better examine the sample complexity $M$ in different parameter regions, we can think the estimation process as follows. Given the estimation accuracy to achieve (the additive error $\epsilon$ and the failure probability $\delta$), we gradually increase $M$ and stop increasing when estimation using $M$ samples can reach such accuracy. As shown in the estimator Eq.~(\ref{eq: estimator of MMD-k}), we only utilize the samples of the labels with $T_\ell\ge k$, so it is important to track how $m=\sum_{\ell=1}^{N^2} \mathbf{1}\{T_{\ell} \ge k\}$ varies when $M$ increases.

When $M\le cN^2$ ($c$ is a constant), given by Eq.~(\ref{eq: Em M<=cN2}), we have $\dsE[m] = \Theta\Big(\frac{M^k}{k!N^{2k-2}}\Big)$, with $m\le N^2$. The threshold of $M$ to reach $m=N^2$ is determined by $T_{\min}\ge k$, which is given by Eq.~(\ref{eq: M needed to achieve Tmin}):
\begin{equation}
    M=N^2(\log \frac{N^2}{\delta}+k).\label{eq: M to reach Tmin=k}
\end{equation}
So there are two regions, we need to consider them separately. We first give a lemma which will be useful for the region $m=N^2$.

\begin{lemma}[Balanced-count event for uniform multinomial sampling]\label{lem:balanced_count}
Let $n:=N^2$. Draw $M$ labels i.i.d. uniformly from $[n]$, and let
$T_\ell$ be the number of times label $\ell\in[n]$ appears, so that
$(T_1,\dots,T_n)\sim \mathrm{Multinomial}(M;1/n,\dots,1/n)$ and
$\dsE[T_\ell]=\lambda:=M/n$ for each $\ell$.
Fix $\eta\in(0,1)$ and define the balanced-count event
\begin{equation}
E_{\mathrm{bal}}^{-}(\eta):=\Big\{\min_{\ell\in[n]} T_\ell \ge (1-\eta)\lambda\Big\}.
\label{eq:app-local-balanced-count-event}
\end{equation}
Then
\begin{equation}\label{eq:balanced_count_prob}
\Pr\big(E_{\mathrm{bal}}^{-}(\eta)\big)\;\ge\;1-n\exp\!\Big(-\frac{\eta^2}{2}\lambda\Big).
\end{equation}
In particular, if
\begin{equation}\label{eq:balanced_count_M_suff}
M \;\ge\; \frac{2n}{\eta^2}\log\frac{n}{\delta},
\end{equation}
then $\Pr(E_{\mathrm{bal}}^{-}(\eta))\ge 1-\delta$.
\end{lemma}

\begin{proof}
Fix any $\ell\in[n]$. Since $(T_1,\dots,T_n)$ is multinomial with uniform cell probability,
the marginal distribution is
\[
T_\ell \sim \mathrm{Bin}\Big(M,\frac{1}{n}\Big),\qquad \dsE[T_\ell]=\lambda:=\frac{M}{n}.
\]
By the Chernoff bound for binomial random variables
as Proposition \ref{prop:Chernoff-S}, for any $\eta\in(0,1)$,
\begin{equation}\label{eq:chernoff_one_bin}
\Pr\big(T_\ell \le (1-\eta)\lambda\big)\;\le\;\exp\!\Big(-\frac{\eta^2}{2}\lambda\Big).
\end{equation}
Union bound gives that 
\[
\begin{aligned}
&\Pr\Big(\min_{\ell\in[n]}T_\ell < (1-\eta)\lambda\Big)\\
&=
\Pr\Big(\exists \ell\in[n]: T_\ell < (1-\eta)\lambda\Big)\\
&\le
\sum_{\ell=1}^n \Pr\big(T_\ell < (1-\eta)\lambda\big)\\
&\le
n\exp\!\Big(-\frac{\eta^2}{2}\lambda\Big),
\end{aligned}
\]
which proves \eqref{eq:balanced_count_prob}. The sufficient condition \eqref{eq:balanced_count_M_suff}
follows by setting the right-hand side to be at most $\delta$ and recalling $\lambda=M/n$.
\end{proof}

Here we present the general version of Theorem \ref{thm: sample complexity of MMD-k, upper bound, most general case}.
\begin{theorem}[Sample complexity of MMD-$k$ with fixed $k$, general version]
    For two $N$-state uniform pure quantum ensembles $\scE_1$ and $\scE_2$, with $k$ fixed, the sample complexity to estimate MMD-$k$ between the two ensembles using the U-stat estimator Eq. (\ref{eq: estimator of MMD-k}) is:
    \begin{widetext}
\begin{equation}\label{eq: sample complexity of MMD-k, upper bound for general case}
    M
    = \mathcal{O}\Bigg(
    \min\Bigg\{
    \left(\frac{k!}{\epsilon^2}\log\frac{1}{\delta}\right)^{\!\frac{1}{k}}
    \,N^{2-\frac{2}{k}},
    \;
    \max\Bigg\{
    \frac{k}{\epsilon^2}\log\frac{1}{\delta},
    \;
    N^2\!\left(\log\frac{N^2}{\delta} + k\right)
    \Bigg\}
    \Bigg\}
    \Bigg).
    \end{equation}
\end{widetext}
    where $\epsilon$ is the additive error between estimated value and true value and $\delta$ is the failure probability.

    In the minimization form, the first term is for the case that $m<N^2$, with the condition $N^2 \ge \frac{k!}{\epsilon^2}\log\frac{1}{\delta}$, and the second term is for the case that $m=N^2$.
\label{theorem:general_version}
\end{theorem}
\begin{proof}
    \textbf{Case 1.} Assume that $m< N^2$. Namely, not all the label $\ell$ are sampled more than $k$-times.
    Now the estimator of $\bar{F}^{(k)}$ is $\widehat{\bar{F}^{(k)}(\scE_a,\scE_b)} = \frac{1}{m}\sum _{l=1}^m Z_{\ell}$ with sample ensembles $\scE_a, \scE_b$ and kernel $Z_{\ell}$ taking the form Eq. (\ref{eq: U-stat for MMD-k}). Triangle inequality gives that
    \begin{equation}
\begin{aligned}
&\left|\widehat{\bar{F}(\scE_a,\scE_b)} - \bar{F} \right| \\
&\leq
    \left| \frac{1}{m}\sum_{\ell} (Z_{\ell}-X_{\ell}^k) \right|\\
&\quad+ \Abs{\frac{1}{m}\sum_{\ell}X_{\ell}^k - \frac{1}{N^2}\sum_{\ell} X_{\ell}^k}.
\end{aligned}
\label{eq:app-local-mmd-error-decomposition}
\end{equation} 
    The first term captures the variance contribution, whereas the remaining term accounts for the subset selection error. For the first term, by virtue of the Hoeffding inequality for U-statistics \cite{Lee_U_statistics_19}, we have for all labels with no less than $k$ repetitions, $\eta>0$,
    \begin{equation}\label{eq: Hoeffding for each label}
        \Pr  \big(\abs{Z_{\ell}-X_{\ell}^k}\ge\eta \big) \leq 2 \exp \big(-\lfloor T_{\ell}/k\rfloor \eta^2/2 \big).
    \end{equation}
    Set $V_{\ell}=Z_{\ell}-X_{\ell}^k$ for simplicity, and we have by Theorem 2.6.3 in \cite{Vershynin_high_dimensional_probability} and the fact that $\min T_\ell  = k$,
    \begin{equation}\label{eq: general Hoeffding of mA}
        \Pr \big(\frac{1}{m}\abs{\sum_{l=1}^{m} V_{\ell}} \ge\eta\big) \leq 2\exp \big(- m \eta^2/2 \big). 
    \end{equation}
    For the second term, by Serfling inequality \cite{Serfling_10.1214/aos/1176342611}, one obtains that
    \begin{equation}
\begin{aligned}
\label{eq: general Hoeffding of mB}
&
        \Pr  \Big( \abs{\frac{1}{m}\sum_{\ell}X_{\ell}^k - \frac{1}{N^2}\sum_{\ell} X_{\ell}^k} \ge\eta\Big) \\
&\le 2 \exp(-\frac{m\eta^2}{2f(m)}) \leq 2\exp(-m\eta^2/2),
\end{aligned}
\end{equation}
    where $f(m) = 1-\frac{m-1}{N^2} \in [0,1]$. As a result of inequality (\ref{eq: general Hoeffding of mA}) and inequality (\ref{eq: general Hoeffding of mB}),
    \begin{align}\label{eq: general Hoeffding of overall}
        \Pr \left( |\frac{1}{m}\sum Z_\ell - \frac{1}{N^2} \sum_\ell X^k_\ell| \geq \eta \right) \leq 4\exp(-m\eta^2/2).
    \end{align}
    Define $E_{\delta, M} := E^{(1)}_{\delta, M} \cap E^{(2)}_{\delta, M} $ with event 
    $$E^{(1)}_{\delta, M} : = \left\{ |m - \dsE [m]| \leq  \sqrt{3 \dsE [m] \log(\frac{2}{\delta})}  \right\},$$
    and event 
    $$
\begin{aligned}
E^{(2)}_{\delta, M} &: = \left\{ |\widehat{\bar{F}(\scE_a,\scE_b)} - \bar{F} | \leq \sqrt{\frac{8\log(4/\delta)}{m}} \right\} \\
&= \left\{ |\frac{1}{m}\sum Z_\ell - \frac{1}{N^2} \sum_\ell X^k_\ell| \leq \sqrt{\frac{8\log(4/\delta)}{m}} \right\}.
\end{aligned}
$$ 
    Here, we suppress the dependence of $m=m(M)$ on $M$ for notational simplicity. By Proposition \ref{prop:Chernoff-S}, we have $\Pr(E^{(1)}_{\delta, M}) \geq 1 - \delta$ and we have $\Pr(E^{(2)}_{\delta, M}) \geq 1 - \delta$ via inequality (\ref{eq: general Hoeffding of overall}) above. Thus, 
    $\Pr(E_{\delta, M}) \geq 1 - 2\delta$. Equation (\ref{eqn_ES_bound_order}) (plug in $n=N^2$) indicates that 
    \begin{align*}
          \dsE [m] = \Theta \left( \frac{M^k}{k!\,N^{2k-2}} \right).
    \end{align*}
    Take
    \begin{equation}
        M = \mathcal{O} \Bigl( \frac{128k!}{\epsilon^2} \ln \frac{12}{\delta} \Bigr)^{\frac{1}{k}} N^{2-\frac{2}{k}}, \label{eq: M for k-MMD when lambda<<1}
    \end{equation}
    and we have $E_{\delta, M}\subset  E^{(0)}_{\delta,\epsilon} : = \{ |\widehat{\bar{F}(\scE_a,\scE_b)} - \bar{F} | \leq \epsilon \}$. As a consequence, 
    \begin{align*}
       \Pr(E^{(0)}_{\delta,\epsilon}) \geq \Pr(E_{\delta, M})\geq 1 - 2\delta.
    \end{align*}
    Plugging into the condition $M<cN^2$, we have the condition to pick the first term
    \begin{equation}
        N^2 \ge c\frac{k!}{\epsilon^2}\log\frac{1}{\delta}.
    \end{equation}

    \textbf{Case 2.}
    Assume $m=N^2$, i.e., $T_{\min}\ge k$ (equivalently, every label $\ell\in[n]$ has at least $k$ samples).
    Then the subset selection error vanishes and
    \[
\begin{aligned}
&\widehat{\bar{F}^{(k)}(\scE_a,\scE_b)}-\bar{F}^{(k)}\\
    &=
    \frac{1}{N^2}\sum_{\ell=1}^{N^2}\big(Z_\ell-X_\ell^k\big)
    =
    \frac{1}{n}\sum_{\ell=1}^{n}V_\ell,
    \qquad n:=N^2,
\end{aligned}
\]
    where $V_\ell:=Z_\ell-X_\ell^k$ and $\dsE[V_\ell]=0$.
    Each $V_\ell$ is sub-Gaussian \cite{Vershynin_high_dimensional_probability} with
    \begin{equation}\label{eq:psi2_Vell}
    \norm{V_\ell}_{\psi_2}^2 \;\le\; C\,\frac{k}{T_\ell},
    \end{equation}
    for an absolute constant $C>0$. 
    Define $\lambda:=M/n$ and consider the balanced-count event with $\eta=1/2$,
    \[
    E_{\mathrm{bal}} := E_{\mathrm{bal}}^{-}(1/2)
    =\Big\{\min_{\ell\in[n]}T_\ell \ge \frac{\lambda}{2}\Big\}.
    \]
    By Lemma~\ref{lem:balanced_count},
    \begin{equation}\label{eq:Ebal_prob_case2}
    \Pr(E_{\mathrm{bal}})
    \ge
    1-n\exp\!\Big(-\frac{\lambda}{8}\Big).
    \end{equation}
    On $E_{\mathrm{bal}}$, we have $T_\ell\ge \lambda/2$ for all $\ell$, and hence by \eqref{eq:psi2_Vell}
    \begin{equation}\label{eq:psi2_uniform_on_Ebal}
    \max_{\ell\in[n]}\norm{V_\ell}_{\psi_2}^2
    \le
    C\,\frac{k}{\lambda/2}
    =
    \frac{2Ck}{\lambda}.
    \end{equation}
    
    Theorem 2.6.3 in \cite{Vershynin_high_dimensional_probability} as well as inequality \eqref{eq:psi2_uniform_on_Ebal}
    yield that for all $\epsilon>0$,
    \begin{widetext}
\begin{equation}\label{eq:case2_conditional_conc}
    \Pr\Big( \bigl\{\abs{\frac{1}{n}\sum_{\ell=1}^{n}V_\ell}\ge \epsilon \bigl\} \cap E_{\mathrm{bal}}  \Big)
    \le
    2\exp\!\Big(-c\,\frac{n\epsilon^2}{\max_\ell \norm{V_\ell}_{\psi_2}^2}\Big)
    \le
    2\exp\!\Big(-c'\,\frac{n\lambda}{k}\epsilon^2\Big)
    =
    2\exp\!\Big(-c'\,\frac{M}{k}\epsilon^2\Big),
    \end{equation}
\end{widetext}
    for absolute constants $c,c'>0$, where we used $n\lambda=M$.
    \clearpage
    As a result, we obtain that 
    \begin{align}
&\Pr\Big(\abs{\widehat{\bar{F}^{(k)}(\scE_a,\scE_b)}-\bar{F}^{(k)}}\ge \epsilon\Big)
    \nonumber\\
    &=
    \Pr\Big(\abs{\frac{1}{n}\sum_{\ell=1}^{n}V_\ell}\ge \epsilon\Big)\nonumber\\
    &\le
    \Pr(E_{\mathrm{bal}}^c)
    +
    \Pr\Big(\bigl\{\abs{\frac{1}{n}\sum_{\ell=1}^{n}V_\ell}\ge \epsilon \bigl\} \cap E_{\mathrm{bal}}\Big)\nonumber\\
    &\le
    n\exp\!\Big(-\frac{\lambda}{8}\Big)
    +
    2\exp\!\Big(-c'\,\frac{M}{k}\epsilon^2\Big),\label{eq:case2_total_fail}
\end{align}
    where we used \eqref{eq:Ebal_prob_case2} and \eqref{eq:case2_conditional_conc}. Therefore, it suffices to choose 
    \begin{equation}
    M
    =
    \mathcal O\!\left(
    \max\left\{
    \frac{k}{\epsilon^2}\log\frac{1}{\delta},
    \;
    N^2\Big(\log\frac{N^2}{\delta}+k\Big)
    \right\}
    \right),
\label{eq:app-local-mmd-dense-samples}
\end{equation}
so that the right hand side of  \eqref{eq:case2_total_fail} is upper bounded by
$\delta$, which completes the proof of Case 2.

Combine the two cases and we obtain the final result.

\vspace{5mm}

\end{proof}



\subsection{Proof of Theorem~\ref{thm: lower bound of sample complexity MMD-k}: label-local-swap-test MMD-$k$ lower bound}
\label{appsub: MMD-k proof of lower bound}

\begin{proof}
This proof is for the label-local SWAP-test transcript model of
Definition~\ref{def: sample complexity}. In this model, the learner observes
only the classical transcript $(R_j,L_j)_{j=1}^M$. The hard instance used here
is different from the one used for the fully general lower bound in
Appendix~\ref{app: proof lower general}: here the hidden variables are attached
to pair labels $\ell=(i,j)\in[N]\times[N]$, whereas in the fully general proof
they are attached to column labels.

In this proof, we use $C$ for universal constant whose value may differ line by line. To prove the lower bound of the sample complexity to estimate MMD-$k$, we need to construct two pairs of ensembles that can yield their MMD-$k$ hard to discriminate. For that, 
Section \ref{Appendix_worst_case_construction} enables us to construct two quantum ensemble pairs whose fidelity matrix is $F_0  : X_{ij} \overset{\text{i.i.d.}}{\sim} \mu_0$, (respectively, $F_1  : X_{ij} \overset{\text{i.i.d.}}{\sim} \mu_1$) satisfying $|\theta_1 - \theta_0| = \Delta_k = 2\varepsilon$ and $|\Delta m_j| \leq C_k\varepsilon(\alpha/N)^{j-k}$ for $j\ge k$, where $\theta_b = \dsE_{\mu_b} X^k$, $b\in\{0,1\}$, and the first $k-1$ moments of $\mu_0,\mu_1$ agree. $P_0$ (respectively $P_1$) denotes the distribution of $(R_j, L_j)_{j=1}^M$ under $H_0$ (resp. $H_1$). Assume an estimator $\hat{\theta}$ satisfies that $\Pr_{H_i} ( |\hat{\theta} - \theta_i| \leq \epsilon ) \geq 1 - \delta$ for $i=0,1$. Then we could therefore define a test $\Psi : = \mathbf 1 \{ \hat{\theta} > (\theta_0 + \theta_1)/2 \}$ and the corresponding type-I error $\Pr_{H_0} (\Psi = 1)\leq \delta$ as well as the type-II error $\Pr_{H_1} (\Psi = 0)\leq \delta$. Bretagnolle--Huber inequality \cite{canonne2022short} gives that
\begin{align}\label{Bretagnolle_Huber_inequality}
    2\delta\geq \Pr_{H_0} (\Psi = 1)+ \Pr_{H_1} (\Psi = 0) \geq \frac{1}{2} \exp( - \text{KL} (P_1\|P_0)).
\end{align}
Note that the label transcript has the same distribution under \(H_0\) and \(H_1\). Conditioning on the labels gives
\begin{align}
     \text{KL} (P_1\|P_0)
     &= \dsE_L\left[
     \sum_{\ell=1}^{N^2}
     \text{KL} ( Q_{1,T_\ell} \| Q_{0,T_\ell} )
     \right] \notag\\
     &\leq
     \dsE_L\left[
     \sum_{\ell=1}^{N^2}
     \chi^2 ( Q_{1,T_\ell} \| Q_{0,T_\ell} )
     \right].
\label{eq:app-local-mmd-transcript-divergence}
\end{align}

For $\vr = (r_1,\dots,r_s) \in \{-1,+1\}^s$, we have that
\begin{equation}Q_{i,s}(\vr) := \dsE_{X\sim \mu_i} \left[
\prod_{t=1}^s \frac{1 + r_t X}{2}  \right] = 2^{-s} \sum_{j=0} ^s m^{(i)}_j e_j(\vr), 
\label{eq:app-local-mmd-outcome-mixture}
\end{equation}
where 
$e_j(\vr) : = \sum_{ \substack{S\subseteq[s]\\|S| = j} } \prod_{t\in S} r_t$. One could check that $\sum_{ \vr\in\{-1,+1\}^s }e_j(\vr) e_{j'}(\vr) = 0$ for $j\neq j'$ and $\sum_{\vr} e_j^2(\vr) = 2^s \binom{s}{j}$. Note that $X\in [0, \alpha/N]$ ($\alpha$ is fixed here, see \ref{Appendix_worst_case_construction} for details),
\begin{align*}
    Q_{0,s}(\vr ) = 2^{-s} \dsE_{X\sim \mu_0} \left[ \prod_{t=1}^s (1 + r_t X) \right] \geq 2^{-s} (1-\alpha/N)^s.
\end{align*}
Denote $\Delta m_j := m_j^{(1)} - m_j^{(0)}$ and we have by our construction that $\Delta m_j=0$ for $j=0,\dots,k-1$ and for $j\geq k$,
\begin{align*}
    |\Delta m_j| \leq C\epsilon \left(\frac{\alpha}{N}\right)^{j-k} \leq C\, \epsilon (1/N)^{j-k}.
\end{align*}
As a result,
\begin{align*}
\sum_{ \vr } (Q_{1,s}( \vr) - Q_{0,s} (\vr) )^2 &= \sum_{ \vr } 2^{-2s} \sum_{j=k}^s (\Delta m_j)^2 e_j^2(\vr)\\
    &= 2^{-s} \sum_{j=k}^s (\Delta m_j)^2 \binom{s}{j}.
\end{align*}
So we get that  
\begin{widetext}
\begin{align}
    \chi^2 ( Q_{1,s} \| Q_{0,s} )  &= \sum_{\vr}\frac{(Q_{1,s}(\vr) - Q_{0,s}(\vr) )^2}{Q_{0,s}(\vr)} \leq (1 - 1/N)^{-s} \sum_{j=k}^s (\Delta m_j)^2 \binom{s}{j}\notag\\
    &\leq C \epsilon^2 (1 - 1/N)^{-s} \sum_{j=k}^s  \binom{s}{j} \left( \frac{1}{N} \right)^{2(j-k)} \leq C \epsilon^2 (1 - 1/N)^{-s} \binom{s}{k} \sum_{t=0}^{s-k} \frac{1}{t!}  \left( \frac{s}{N^2} \right)^{t}\notag\\
    &\leq C \epsilon^2 (1 - 1/N)^{-s} \binom{s}{k} \exp \left(\frac{ s}{N^2} \right) \leq C \epsilon^2  \binom{s}{k}.
\label{eq:app-local-mmd-chi-square-bound}
\end{align}
\end{widetext}
Thus, we obtain that
\begin{align}\label{final_KL_lower_bound_proof}
    \text{KL} (P_1 \|P_0)
    &\leq
    C \epsilon^2\,
    \dsE_L\left[
    \sum_{\ell=1}^{N^2} \binom{T_\ell}{k}
    \right] \notag\\
    &=
    C \epsilon^2
    N^2\binom{M}{k}\left(\frac1{N^2}\right)^k
    =
    C \epsilon^2 \frac{\binom{M}{k}}{N^{2k-2}}.
\end{align}
So we have by (\ref{Bretagnolle_Huber_inequality}) and (\ref{final_KL_lower_bound_proof}) that 
\begin{align}
    M \geq C N^{2-2/k} \left( \frac{k! \log(1/\delta)}{\epsilon^2} \right)^{1/k}.
\label{eq:app-local-mmd-sample-lower}
\end{align}
\end{proof}

\subsection{Proof of Theorem~\ref{thm: sample complexity of MMD-k, k large}}\label{appsub: MMD-k k large}

\begin{proof}
Let $n:=N^2$ be the number of labels. Draw $M$ labels $L_1,\dots,L_M$ i.i.d.\ uniformly from $[n]$ and
let $T_\ell:=\sum_{i=1}^M \mathbf 1\{L_i=\ell\}$ be the occupancy of label $\ell$.
Then $(T_1,\dots,T_n)\sim \mathrm{Multinomial}(M;1/n,\dots,1/n)$ and each marginal satisfies
\[
T_\ell \sim \mathrm{Bin}\Big(M,\frac{1}{n}\Big),\qquad \mu:=\dsE[T_\ell]=\frac{M}{n}.
\]
Recall $m=\sum_{\ell=1}^n \mathbf 1\{T_\ell\ge k\}$ and the estimator uses only labels with $T_\ell\ge k$.

\paragraph{Lower bound.}
Fix $\eta\in(0,1)$ and set $M=(1-\eta)nk$, so $\mu=(1-\eta)k$.
For a fixed $\ell$, write $k=(1+\alpha)\mu$ with $\alpha=\eta/(1-\eta)$. By the Chernoff bound as Proposition \ref{prop:Chernoff-S},
\[
\Pr(T_\ell\ge k)=\Pr\big(T_\ell\ge (1+\alpha)\mu\big)
\le \exp\!\left(-\frac{\mu\alpha^2}{2+\alpha}\right).
\]
A direct calculation gives
\[
\frac{\mu\alpha^2}{2+\alpha}
=
\frac{(1-\eta)k\cdot (\eta/(1-\eta))^2}{2+\eta/(1-\eta)}
=
\frac{\eta^2}{2-\eta}\,k,
\]
hence $\Pr(T_\ell\ge k)\le \exp\!\big(-\frac{\eta^2}{2-\eta}k\big)$.
By the union bound,
\begin{equation}
\begin{aligned}
\Pr(m\ge 1)&=\Pr(\exists\,\ell:\ T_\ell\ge k)\\
&\le \sum_{\ell=1}^n \Pr(T_\ell\ge k)\\
&\le n\exp\!\left(-\frac{\eta^2}{2-\eta}\,k\right).
\end{aligned}
\label{eq:app-local-collision-obstruction}
\end{equation}
Thus if $k\ge \frac{2-\eta}{\eta^2}\log\frac{n}{\delta}$, then $\Pr(m=0)\ge 1-\delta$.
In particular, for $k\sim N$ we have $k\gg \log n$, so with high probability $m=0$ when $M=(1-\eta)nk$.
Therefore any $(\epsilon,\delta)$-guarantee requires $M=\Omega(nk)=\Omega(N^2k)$.

\paragraph{Upper bound.}
Fix $\eta\in(0,1)$ and set $M=(1+\eta)nk$, so $\mu=(1+\eta)k$.
For a fixed $\ell$, write $k=(1-\beta)\mu$ with $\beta=\eta/(1+\eta)$. By the Chernoff lower-tail bound as Proposition \ref{prop:Chernoff-S},
\[
\begin{aligned}
\Pr(T_\ell<k)&\le \Pr\big(T_\ell\le (1-\beta)\mu\big)\\
&\le \exp\!\left(-\frac{\mu\beta^2}{2}\right)\\
&= \exp\!\left(-\frac{\eta^2}{2(1+\eta)}\,k\right).
\end{aligned}
\]
Union bound over $\ell\in[n]$ yields
\begin{equation}
\Pr(T_{\min}<k)\le n\exp\!\left(-\frac{\eta^2}{2(1+\eta)}\,k\right).
\label{eq:app-local-full-coverage-tail}
\end{equation}
Hence if $k\ge \frac{2(1+\eta)}{\eta^2}\log\frac{n}{\delta}$, then $\Pr(T_{\min}\ge k)\ge 1-\delta$, and thus $m=n$.

For $m=n$, an upper bound of sample complexity is given by $\frac{k}{\epsilon^2}\log \frac{1}{\delta}$ (see Theorem~\ref{theorem:general_version}), for large $N$, we can assume $nk> \frac{k}{\epsilon^2}\log \frac{1}{\delta} $, thus only keep $n k$ as the upper bound.

Combining the lower and upper bounds, for $k\sim N$ we obtain
$M=\Theta(nk)=\Theta(N^2k)$ (equivalently, $M=\Theta(N^3)$ when $k=\Theta(N)$).
\end{proof}

We also care about $k\sim\log N$, here we introduce the sample complexity of MMD-$k$ when $k\sim\log N$:
\begin{theorem}[Sample complexity when $k=\Theta(\log N)$]
Let $n:=N^2$ be the number of labels. Assume $k=c\log N$ for a constant $c>0$.
Fix $\epsilon\in(0,1)$ and $\delta\in(0,1/2)$. Then there exist constants $C,c'>0$
such that:
\begin{enumerate}
\item \textbf{(Upper bound)} If
\begin{equation}
M \;\ge\; C\Big(nk+\frac{k}{\epsilon^2}\log\frac{1}{\delta}\Big),
\label{eq:app-local-log-k-upper}
\end{equation}
then the U-statistic estimator Eq.~(\ref{eq: estimator of MMD-k}) achieves additive error at most $\epsilon$ with
failure probability at most $\delta$. In particular, for fixed $\epsilon,\delta$,
$M=\scO(nk)=\scO(N^2\log N)$.
\item \textbf{(Lower bound)} There exist two $N$-state uniform pure ensembles such that
any estimator achieving additive error $\epsilon$ with failure probability at most $\delta$
must satisfy
\begin{equation}
M\;\ge\; c'\, nk\cdot \epsilon^{-2/k}\cdot \big(\log(1/\delta)\big)^{1/k}.
\label{eq:app-local-log-k-lower}
\end{equation}
In particular, for $k=\Theta(\log N)$ and fixed $\epsilon,\delta$,
$M=\Omega(nk)=\Omega(N^2\log N)$.
\end{enumerate}
Consequently, for $k=\Theta(\log N)$ and fixed $\epsilon,\delta$, the optimal
sample complexity is $M=\Theta(N^2\log N)$.
\end{theorem}

\begin{proof}
Let $L_1,\dots,L_M$ be i.i.d.\ uniform over $[n]$ and define occupancies
$T_\ell:=\sum_{i=1}^M \mathbf 1\{L_i=\ell\}$. Then $T_\ell\sim \mathrm{Bin}(M,1/n)$ with
$\mu=\mathbb E[T_\ell]=M/n$.

\textbf{Upper bound.}
Fix a constant $A>1$ and set $M:=A\,nk$, so $\mu=Ak$.
Let $\beta:=1-\frac{k}{\mu}=1-\frac1A=\frac{A-1}{A}\in(0,1)$. By the Chernoff bound,
\[
\begin{aligned}
\Pr(T_\ell<k)&=\Pr\big(T_\ell\le (1-\beta)\mu\big)
\\
&\le \exp\!\Big(-\frac{\mu\beta^2}{2}\Big)\\
&=\exp\!\Big(-\frac{(A-1)^2}{2A}\,k\Big).
\end{aligned}
\]
By a union bound,
\[
\Pr(T_{\min}<k)\le n\exp\!\Big(-\frac{(A-1)^2}{2A}\,k\Big).
\]
For $k=c\log N=\tfrac{c}{2}\log n$, choosing $A$ (a constant) so that
$\frac{(A-1)^2}{2A}\cdot \frac{c}{2}>1$ implies $\Pr(T_{\min}\ge k)\ge 1-\delta/2$
for all sufficiently large $n$. On this event, $m=n$ and thus
\[
\widehat{\bar F}^{(k)}=\frac{1}{n}\sum_{\ell=1}^n Z_\ell,
\qquad 0\le Z_\ell\le 1,\qquad \mathbb E[Z_\ell\mid T_\ell]=X_\ell^k.
\]
Conditional on the label partition (equivalently on $(T_1,\dots,T_n)$), the collections
of samples used to form $Z_\ell$ are disjoint across $\ell$, hence $Z_1,\dots,Z_n$
are conditionally independent. Therefore Hoeffding's inequality yields
\begin{equation}
\Pr\!\Big(\Big|\widehat{\bar F}^{(k)}-\frac1n\sum_{\ell=1}^n X_\ell^k\Big|\ge \epsilon/2\ \Big|\ (T_\ell)\Big)
\le 2\exp\!\Big(-\frac{n\epsilon^2}{2}\Big).\label{eq: hoeffding for m=n}
\end{equation}
Combining with $\Pr(T_{\min}\ge k)\ge 1-\delta/2$ gives
$\Pr(|\widehat{\bar F}^{(k)}-\bar F^{(k)}|\le \epsilon/2)\ge 1-\delta$
(we assume $N^2 > \frac{2}{\epsilon^2}\log \frac{4}{\delta}$, then Eq.~(\ref{eq: hoeffding for m=n}) is satisfied.), and we have shown $\Pr(T_{\min}\ge k)\ge 1-\delta/2$, so $M=Ank$ is an upper bound. 

\textbf{Lower bound.}
By the minimax lower bound (Theorem~\ref{thm: lower bound of sample complexity MMD-k}),
\[
M \;\ge\; c_0\left(\frac{k!\,\log(1/\delta)}{\epsilon^2}\right)^{\!1/k} N^{2-2/k}.
\]
For $k=\Theta(\log N)$, Stirling implies $(k!)^{1/k}=\Theta(k)$ and
$N^{2-2/k}=\Theta(N^2)$, while $\epsilon^{-2/k}$ and $(\log(1/\delta))^{1/k}$
are $\Theta(1)$ for fixed $\epsilon,\delta$. Hence $M=\Omega(N^2k)=\Omega(N^2\log N)$.
\end{proof}

\subsection{Proof of Theorem~\ref{thm: sample complexity MMD-k, prior information}}
\label{appsub: MMD-k prior information across-label}

\subsubsection{Across-label part}
We first isolate the common deterministic error caused by coarsening the ensemble into $\epsilon$-balls.

\begin{lemma}[Coarsening error in the $\epsilon$-ball model]
    \label{lem: coarsening error epsilon balls}
    Consider two uniform pure-state ensembles $\scE_1$ and $\scE_2$. Suppose each ensemble is uniformly divided into $n$ $\epsilon$-balls of radius $\epsilon_b$. For $u,v\in\{1,2\}$, let $B_a^{(u)}$ be the $a$-th ball of $\scE_u$ and let $\ket{\widetilde{\psi}_a^{(u)}}$ be its central state. Define
    \begin{equation}
        \mu_{ab}^{uv}
        :=
        \dsE\!\left[
            \abs{\braket{\psi}{\phi}}^2
            \,\middle|\,
            \psi\in B_a^{(u)},\ \phi\in B_b^{(v)}
        \right],
    \end{equation}
    where $\ket{\psi}$ and $\ket{\phi}$ are sampled from the conditional distributions induced by the original ensembles. Define the coarsened moment functional
    \begin{equation}
        \bar F^{(k)}_{\mathrm{ball}}(\scE_u,\scE_v)
        :=
        \frac1{n^2}\sum_{a,b=1}^n
        \left(\mu_{ab}^{uv}\right)^k
    \end{equation}
    and
    \begin{equation}
\begin{aligned}
&\scD^{(k)}_{\mathrm{ball}}(\scE_1,\scE_2)\\
        &:=
        \bar F^{(k)}_{\mathrm{ball}}(\scE_1,\scE_1)
        +
        \bar F^{(k)}_{\mathrm{ball}}(\scE_2,\scE_2)
        -
        2\bar F^{(k)}_{\mathrm{ball}}(\scE_1,\scE_2).
\end{aligned}
\end{equation}
    Then
    \begin{equation}
        \Abs{
            \scD^{(k)}(\scE_1,\scE_2)
            -
            \scD^{(k)}_{\mathrm{ball}}(\scE_1,\scE_2)
        }
        \le
        32k\sqrt{\epsilon_b}.
    \end{equation}
\end{lemma}

\begin{proof}
    Let
    \begin{equation}
        d(\psi,\phi):=\sqrt{1-\abs{\braket{\psi}{\phi}}^2}.
    \end{equation}
    This is the sine distance between pure states and satisfies the triangle inequality. Fix $u,v\in\{1,2\}$ and a ball-pair label $\ell=(a,b)$. For $\ket{\psi}\in B_a^{(u)}$ and $\ket{\phi}\in B_b^{(v)}$, define
    \begin{equation}
        X:=\abs{\braket{\psi}{\phi}}^2,
        \qquad
        \widetilde{x}_{ab}^{uv}
        :=
        \abs{
            \braket{\widetilde{\psi}_a^{(u)}}{\widetilde{\psi}_b^{(v)}}
        }^2 .
    \end{equation}
    By Definition~\ref{def: epsilon ball},
    \begin{equation}
        d(\psi,\widetilde{\psi}_a^{(u)})\le\sqrt{\epsilon_b},
        \qquad
        d(\phi,\widetilde{\psi}_b^{(v)})\le\sqrt{\epsilon_b}.
    \end{equation}
    Therefore,
    \begin{align}
        \Abs{X-\widetilde{x}_{ab}^{uv}}
        &=
        \Abs{
            d(\psi,\phi)^2
            -
            d(\widetilde{\psi}_a^{(u)},\widetilde{\psi}_b^{(v)})^2
        } \nonumber\\
        &\le
        2\Abs{
            d(\psi,\phi)
            -
            d(\widetilde{\psi}_a^{(u)},\widetilde{\psi}_b^{(v)})
        } \nonumber\\
        &\le
        2d(\psi,\widetilde{\psi}_a^{(u)})
        +
        2d(\phi,\widetilde{\psi}_b^{(v)})
        \le
        4\sqrt{\epsilon_b}.
    \end{align}
    Since $x\mapsto x^k$ is $k$-Lipschitz on $[0,1]$,
    \begin{equation}
        \Abs{X^k-(\widetilde{x}_{ab}^{uv})^k}
        \le
        4k\sqrt{\epsilon_b}.
        \label{eq: epsilon ball pointwise kth moment}
    \end{equation}
    Taking conditional expectation over $\psi\in B_a^{(u)}$ and $\phi\in B_b^{(v)}$ gives
    \begin{equation}
        \Abs{
            \dsE[X^k\mid \ell]
            -
            (\widetilde{x}_{ab}^{uv})^k
        }
        \le
        4k\sqrt{\epsilon_b}.
        \label{eq: conditional kth moment to center}
    \end{equation}
    Also,
    \begin{equation}
        \Abs{
            \mu_{ab}^{uv}
            -
            \widetilde{x}_{ab}^{uv}
        }
        =
        \Abs{
            \dsE[X\mid \ell]
            -
            \widetilde{x}_{ab}^{uv}
        }
        \le
        4\sqrt{\epsilon_b},
    \end{equation}
    and hence
    \begin{equation}
        \Abs{
            (\mu_{ab}^{uv})^k
            -
            (\widetilde{x}_{ab}^{uv})^k
        }
        \le
        4k\sqrt{\epsilon_b}.
        \label{eq: conditional mean to center kth}
    \end{equation}
    Combining Eqs.~(\ref{eq: conditional kth moment to center}) and~(\ref{eq: conditional mean to center kth}) yields
    \begin{equation}
        \Abs{
            \dsE[X^k\mid \ell]
            -
            (\mu_{ab}^{uv})^k
        }
        \le
        8k\sqrt{\epsilon_b}.
    \end{equation}
    Averaging over the uniform ball-pair label $\ell=(a,b)$ gives
    \begin{equation}
        \Abs{
            \bar F^{(k)}(\scE_u,\scE_v)
            -
            \bar F^{(k)}_{\mathrm{ball}}(\scE_u,\scE_v)
        }
        \le
        8k\sqrt{\epsilon_b}.
        \label{eq: Fk coarsening error}
    \end{equation}
    Applying Eq.~(\ref{eq: Fk coarsening error}) to the three terms in
    \[
        \scD^{(k)}
        =
        \bar F^{(k)}(\scE_1,\scE_1)
        +
        \bar F^{(k)}(\scE_2,\scE_2)
        -
        2\bar F^{(k)}(\scE_1,\scE_2)
    \]
    gives
    \begin{equation}
\begin{aligned}
&\Abs{
            \scD^{(k)}(\scE_1,\scE_2)
            -
            \scD^{(k)}_{\mathrm{ball}}(\scE_1,\scE_2)
        }\\
        &\le
        8k\sqrt{\epsilon_b}
        +
        8k\sqrt{\epsilon_b}
        +
        2\cdot 8k\sqrt{\epsilon_b}\\
        &=
        32k\sqrt{\epsilon_b}.
\end{aligned}
\end{equation}
\end{proof}

\begin{proof}[Proof of the across-label part of Theorem~\ref{thm: sample complexity MMD-k, prior information}]
    Set
    \begin{equation}
        \epsilon'
        :=
        \epsilon-32k\sqrt{\epsilon_b}>0,
        \qquad
        s
        :=
        \left\lceil
            c(\epsilon')^{-2}\log\frac1\delta
        \right\rceil,
    \end{equation}
    where $c>0$ is a sufficiently large universal constant. By Lemma~\ref{lem: coarsening error epsilon balls}, it suffices to estimate
    $\scD^{(k)}_{\mathrm{ball}}(\scE_1,\scE_2)$ to additive error $\epsilon'$ and failure probability $\delta$.

    Fix $u,v\in\{1,2\}$. We estimate $\bar F^{(k)}_{\mathrm{ball}}(\scE_u,\scE_v)$. Draw
    \begin{equation}
        M_0
        =
        \left\lceil
            C_k (k!s)^{1/k} n^{1-1/k}
        \right\rceil
    \end{equation}
    paired samples from $(\scE_u,\scE_v)$, where $C_k>0$ is a sufficiently large constant depending only on $k$. Denote the ball labels by $(I_t,J_t)$ and the quantum registers by $A_tB_t$, for $t=1,\ldots,M_0$.

    Let
    \begin{equation}
\begin{gathered}
R:=\{a\in[n]:\#\{t:I_t=a\}\ge k\},
        \\
        C:=\{b\in[n]:\#\{t:J_t=b\}\ge k\}.
\end{gathered}
\end{equation}
    Write $G_R:=|R|$, $G_C:=|C|$, and $G:=\min\{G_R,G_C\}$. In the sparse regime
    \[
        k!s\lesssim n,
    \]
    we have $M_0/n=\scO_k(1)$. Standard occupancy estimates for uniform multinomial sampling give
    \begin{equation}
        \dsE[G_R]
        =
        \Theta_k\!\left(
            \frac{M_0^k}{k!n^{k-1}}
        \right)
        =
        \Theta_k(s),
        \qquad
        \dsE[G_C]
        =
        \Theta_k(s).
    \end{equation}
    Choosing $C_k$ sufficiently large and applying the Chernoff lower-tail bound for negatively associated occupancy indicators gives
    \begin{equation}
        \Pr(G<s)\le \frac{\delta}{12}.
        \label{eq: across enough row column collisions}
    \end{equation}

    On the event $G\ge s$, choose $s$ row labels uniformly without replacement from $R$ and $s$ column labels uniformly without replacement from $C$, then pair them by an independent uniformly random permutation. By exchangeability, the selected row labels are a uniform $s$-subset of $[n]$, and the selected column labels are an independent uniform $s$-subset of $[n]$.

    For the $r$-th selected row label $i_r$, choose one disjoint row block
    \[
        S_r=\{\alpha_{r,1},\ldots,\alpha_{r,k}\}
        \quad\text{with}\quad
        I_{\alpha_{r,1}}=\cdots=I_{\alpha_{r,k}}=i_r.
    \]
    For the paired selected column label $j_r$, choose one disjoint column block
    \[
        T_r=\{\beta_{r,1},\ldots,\beta_{r,k}\}
        \quad\text{with}\quad
        J_{\beta_{r,1}}=\cdots=J_{\beta_{r,k}}=j_r.
    \]
    The row blocks are mutually disjoint, and the column blocks are mutually disjoint. On the corresponding $2k$ registers, measure the observable
    \begin{equation}
        W_r
        =
        \bigotimes_{\ell=1}^k
        \operatorname{SWAP}
        \!\left(
            A_{\alpha_{r,\ell}},
            B_{\beta_{r,\ell}}
        \right),
    \end{equation}
    with projectors $(I\pm W_r)/2$. Let $Z_r\in\{-1,+1\}$ be the outcome.

    Conditioned on the selected ball labels $(i_r,j_r)$, the registers used in different $r$ are independent and disjoint. Moreover,
    \begin{equation}
        \dsE[Z_r\mid i_r,j_r]
        =
        \left(\mu_{i_rj_r}^{uv}\right)^k.
    \end{equation}
    Define
    \begin{equation}
        \widehat{\bar F}^{(k)}_{\mathrm{ball}}(\scE_u,\scE_v)
        :=
        \frac1s\sum_{r=1}^s Z_r
    \end{equation}
    on the event $G\ge s$, and define it arbitrarily on the event $G<s$.

    On $G\ge s$, decompose
    \begin{align}
&\widehat{\bar F}^{(k)}_{\mathrm{ball}}(\scE_u,\scE_v)
        -
        \bar F^{(k)}_{\mathrm{ball}}(\scE_u,\scE_v)
        \nonumber\\
    &=
        \frac1s\sum_{r=1}^s
        \left[
            Z_r
            -
            \left(\mu_{i_rj_r}^{uv}\right)^k
        \right] \nonumber\\
        &\quad+
        \left[
            \frac1s\sum_{r=1}^s
            \left(\mu_{i_rj_r}^{uv}\right)^k
            -
            \frac1{n^2}\sum_{a,b=1}^n
            \left(\mu_{ab}^{uv}\right)^k
        \right].
\end{align}
    The first term is measurement noise. Conditional on the selected labels, Hoeffding's inequality gives
    \begin{equation}
        \Pr\!\left[
            \Abs{
                \frac1s\sum_{r=1}^s
                \left(
                    Z_r
                    -
                    \left(\mu_{i_rj_r}^{uv}\right)^k
                \right)
            }
            >
            \frac{\epsilon'}8
            \,\middle|\,
            I,J
        \right]
        \le
        \frac{\delta}{12}
    \end{equation}
    for the above choice of $s$. The second term is label-selection noise. Since the selected row and column labels are sampled uniformly without replacement from balanced populations and the entries $(\mu_{ab}^{uv})^k$ lie in $[0,1]$, Serfling's inequality gives
    \begin{equation}
        \Pr\!\left[
            \Abs{
                \frac1s\sum_{r=1}^s
                \left(\mu_{i_rj_r}^{uv}\right)^k
                -
                \frac1{n^2}\sum_{a,b=1}^n
                \left(\mu_{ab}^{uv}\right)^k
            }
            >
            \frac{\epsilon'}8
        \right]
        \le
        \frac{\delta}{12}.
    \end{equation}
    Combining these two bounds with Eq.~(\ref{eq: across enough row column collisions}) and adjusting constants, we obtain
    \begin{equation}
        \Pr\!\left[
            \Abs{
                \widehat{\bar F}^{(k)}_{\mathrm{ball}}(\scE_u,\scE_v)
                -
                \bar F^{(k)}_{\mathrm{ball}}(\scE_u,\scE_v)
            }
            >
            \frac{\epsilon'}4
        \right]
        \le
        \frac{\delta}{3}.
        \label{eq: across one F term}
    \end{equation}

    Apply the same construction independently to $(u,v)=(1,1),(2,2),(1,2)$ and set
    \begin{equation}
        \widehat{\scD}^{(k)}_{\mathrm{ball}}
        :=
        \widehat{\bar F}^{(k)}_{\mathrm{ball}}(\scE_1,\scE_1)
        +
        \widehat{\bar F}^{(k)}_{\mathrm{ball}}(\scE_2,\scE_2)
        -
        2\widehat{\bar F}^{(k)}_{\mathrm{ball}}(\scE_1,\scE_2).
    \end{equation}
    By Eq.~(\ref{eq: across one F term}) and a union bound,
    \begin{equation}
        \Pr\!\left[
            \Abs{
                \widehat{\scD}^{(k)}_{\mathrm{ball}}
                -
                \scD^{(k)}_{\mathrm{ball}}(\scE_1,\scE_2)
            }
            >
            \epsilon'
        \right]
        \le
        \delta.
    \end{equation}
    Finally, Lemma~\ref{lem: coarsening error epsilon balls} gives
    \begin{equation}
        \Pr\!\left[
            \Abs{
                \widehat{\scD}^{(k)}_{\mathrm{ball}}
                -
                \scD^{(k)}(\scE_1,\scE_2)
            }
            >
            \epsilon
        \right]
        \le
        \delta.
    \end{equation}
    The total sample complexity is
    \begin{equation}
        M
        =
        3M_0
        =
        \scO\Bigg(
        \bigg(
            \frac{k!}{(\epsilon-32k\sqrt{\epsilon_b})^2}
            \log\frac1\delta
        \bigg)^{\frac1k}
        n^{1-\frac1k}
        \Bigg),
    \end{equation}
    in the sparse regime
    \[
        \frac{k!}{(\epsilon-32k\sqrt{\epsilon_b})^2}
        \log\frac1\delta
        \lesssim n.
    \]
    This proves the across-label bound.
\end{proof}

\subsubsection{Label-local part}
\label{appsub: MMD-k prior information label-local}

\begin{proof}
    Set
    \begin{equation}
        \epsilon'
        :=
        \epsilon-32k\sqrt{\epsilon_b}>0,
        \qquad
        s
        :=
        \left\lceil
            c(\epsilon')^{-2}\log\frac1\delta
        \right\rceil,
    \end{equation}
    where $c>0$ is a sufficiently large universal constant. By Lemma~\ref{lem: coarsening error epsilon balls}, it suffices to estimate
    $\scD^{(k)}_{\mathrm{ball}}(\scE_1,\scE_2)$ to additive error $\epsilon'$ and failure probability $\delta$.

    Fix $u,v\in\{1,2\}$. We estimate $\bar F^{(k)}_{\mathrm{ball}}(\scE_u,\scE_v)$ using label-local SWAP tests. There are
    \begin{equation}
        L:=n^2
    \end{equation}
    possible pair-of-ball labels $\ell=(a,b)\in[n]\times[n]$. In each trial, we sample one state from $\scE_u$ and one state from $\scE_v$, record the pair-of-ball label $\ell$, immediately perform a SWAP test, and record the unbiased SWAP outcome
    \begin{equation}
        R\in\{-1,+1\},
        \qquad
        \dsE[R\mid \ell=(a,b)]
        =
        \mu_{ab}^{uv}.
    \end{equation}
    Draw
    \begin{equation}
        M_0
        =
        \left\lceil
            C_k(k!s)^{1/k}L^{1-1/k}
        \right\rceil
        =
        \left\lceil
            C_k(k!s)^{1/k}n^{2-2/k}
        \right\rceil
    \end{equation}
    such samples, where $C_k>0$ is a sufficiently large constant depending only on $k$.

    Let $T_\ell$ be the number of times pair-of-ball label $\ell$ appears, and let
    \begin{equation}
        \mathcal G
        :=
        \{\ell:T_\ell\ge k\},
        \qquad
        m:=|\mathcal G|.
    \end{equation}
    In the sparse regime
    \[
        k!s\lesssim L=n^2,
    \]
    we have $M_0/L=\scO_k(1)$. Standard occupancy estimates for uniform multinomial sampling give
    \begin{equation}
        \dsE[m]
        =
        \Theta_k\!\left(
            \frac{M_0^k}{k!L^{k-1}}
        \right)
        =
        \Theta_k(s).
    \end{equation}
    Choosing $C_k$ sufficiently large and applying the Chernoff lower-tail bound for negatively associated occupancy indicators gives
    \begin{equation}
        \Pr(m<s)\le \frac{\delta}{12}.
        \label{eq: local enough pair labels}
    \end{equation}

    For each $\ell\in\mathcal G$, write the SWAP outcomes with label $\ell$ as
    \[
        R_{\ell,1},\ldots,R_{\ell,T_\ell}.
    \]
    Define the order-$k$ U-statistic kernel
    \begin{equation}
        Z_\ell
        :=
        \frac{1}{\binom{T_\ell}{k}}
        \sum_{1\le t_1<\cdots<t_k\le T_\ell}
        \prod_{r=1}^k R_{\ell,t_r}.
        \label{eq: local epsilon ball U statistic}
    \end{equation}
    Conditioned on $\ell=(a,b)$, the random variables $R_{\ell,t}$ are independent and have mean $\mu_{ab}^{uv}$. Hence
    \begin{equation}
        \dsE[Z_\ell\mid \ell=(a,b),T_\ell]
        =
        \left(\mu_{ab}^{uv}\right)^k.
        \label{eq: local Z ell mean}
    \end{equation}
    The label-local estimator for one $\bar F$ term is
    \begin{equation}
        \widehat{\bar F}^{(k)}_{\mathrm{ball}}(\scE_u,\scE_v)
        :=
        \frac1m\sum_{\ell\in\mathcal G} Z_\ell
    \end{equation}
    on the event $m>0$, and is defined arbitrarily otherwise. This is the same U-statistic estimator as Eq.~(\ref{eq: estimator of MMD-k}), with the pair-of-ball label replacing the original pair label.

    On the event $m\ge s$, decompose
    \begin{align}
&\widehat{\bar F}^{(k)}_{\mathrm{ball}}(\scE_u,\scE_v)
        -
        \bar F^{(k)}_{\mathrm{ball}}(\scE_u,\scE_v)
        \nonumber\\
    &=
        \frac1m\sum_{\ell\in\mathcal G}
        \left[
            Z_\ell
            -
            (\mu_\ell^{uv})^k
        \right] \nonumber\\
        &\quad+
        \left[
            \frac1m\sum_{\ell\in\mathcal G}
            (\mu_\ell^{uv})^k
            -
            \frac1L\sum_{\ell=1}^L
            (\mu_\ell^{uv})^k
        \right],
        \label{eq: local decomposition}
\end{align}
    where $\mu_\ell^{uv}=\mu_{ab}^{uv}$ for $\ell=(a,b)$.

    The first term in Eq.~(\ref{eq: local decomposition}) is SWAP-measurement noise. Conditional on the label transcript, the random variables $\{Z_\ell:\ell\in\mathcal G\}$ are independent, bounded in $[-1,1]$, and have means $(\mu_\ell^{uv})^k$. Therefore, by Hoeffding's inequality, on $m\ge s$,
    \begin{equation}
        \Pr\!\left[
            \Abs{
                \frac1m\sum_{\ell\in\mathcal G}
                \left(
                    Z_\ell
                    -
                    (\mu_\ell^{uv})^k
                \right)
            }
            >
            \frac{\epsilon'}8
            \,\middle|\,
            \mathcal G
        \right]
        \le
        \frac{\delta}{12}
    \end{equation}
    for the above choice of $s$.

    The second term in Eq.~(\ref{eq: local decomposition}) is label-selection noise. Conditional on $m$, the random set $\mathcal G$ is a uniform $m$-subset of the $L=n^2$ pair-of-ball labels by exchangeability of the multinomial occupancy vector. Since $(\mu_\ell^{uv})^k\in[0,1]$, Serfling's inequality gives, on $m\ge s$,
    \begin{equation}
        \Pr\!\left[
            \Abs{
                \frac1m\sum_{\ell\in\mathcal G}
                (\mu_\ell^{uv})^k
                -
                \frac1L\sum_{\ell=1}^L
                (\mu_\ell^{uv})^k
            }
            >
            \frac{\epsilon'}8
        \right]
        \le
        \frac{\delta}{12}.
    \end{equation}
    Combining these two bounds with Eq.~(\ref{eq: local enough pair labels}) and adjusting constants, we obtain
    \begin{equation}
        \Pr\!\left[
            \Abs{
                \widehat{\bar F}^{(k)}_{\mathrm{ball}}(\scE_u,\scE_v)
                -
                \bar F^{(k)}_{\mathrm{ball}}(\scE_u,\scE_v)
            }
            >
            \frac{\epsilon'}4
        \right]
        \le
        \frac{\delta}{3}.
        \label{eq: local one F term}
    \end{equation}

    Apply the same construction independently to $(u,v)=(1,1),(2,2),(1,2)$ and set
    \begin{equation}
        \widehat{\scD}^{(k)}_{\mathrm{ball}}
        :=
        \widehat{\bar F}^{(k)}_{\mathrm{ball}}(\scE_1,\scE_1)
        +
        \widehat{\bar F}^{(k)}_{\mathrm{ball}}(\scE_2,\scE_2)
        -
        2\widehat{\bar F}^{(k)}_{\mathrm{ball}}(\scE_1,\scE_2).
    \end{equation}
    By Eq.~(\ref{eq: local one F term}) and a union bound,
    \begin{equation}
        \Pr\!\left[
            \Abs{
                \widehat{\scD}^{(k)}_{\mathrm{ball}}
                -
                \scD^{(k)}_{\mathrm{ball}}(\scE_1,\scE_2)
            }
            >
            \epsilon'
        \right]
        \le
        \delta.
    \end{equation}
    Finally, Lemma~\ref{lem: coarsening error epsilon balls} gives
    \begin{equation}
        \Pr\!\left[
            \Abs{
                \widehat{\scD}^{(k)}_{\mathrm{ball}}
                -
                \scD^{(k)}(\scE_1,\scE_2)
            }
            >
            \epsilon
        \right]
        \le
        \delta.
    \end{equation}
    The total sample complexity is
    \begin{equation}
        M
        =
        3M_0
        =
        \scO\Bigg(
        \bigg(
            \frac{k!}{(\epsilon-32k\sqrt{\epsilon_b})^2}
            \log\frac1\delta
        \bigg)^{\frac1k}
        n^{2-\frac2k}
        \Bigg),
    \end{equation}
    in the sparse regime
    \[
        \frac{k!}{(\epsilon-32k\sqrt{\epsilon_b})^2}
        \log\frac1\delta
        \lesssim n^2.
    \]
    This proves the label-local bound.
\end{proof}

\section{Sample complexity of distance metrics on nonuniform ensembles}
\label{app: general case of sample complexity}

Here we consider nonuniform pure-state ensembles
\[
    \scE_1=\{(p_i,\ket{\psi_i})\}_{i=1}^{N_1},
    \qquad
    \scE_2=\{(q_j,\ket{\phi_j})\}_{j=1}^{N_2}.
\]
The main text focuses on uniform ensembles in order to isolate the scaling with the ensemble size. In the nonuniform case, there is in general no universal scaling that depends only on \(N_1,N_2\). If most of the probability mass is concentrated on a few labels, then rare labels need not affect a constant-accuracy estimate of MMD-\(k\). Thus, for MMD-\(k\), the relevant parameters are the weight profile and the corresponding \(k\)-collision probabilities.

For an ensemble \(\scE_a\), write its weights as
\[
    r^{(a)}=(r^{(a)}_1,\ldots,r^{(a)}_{N_a}),
\]
where \(r^{(1)}=p\), \(r^{(2)}=q\), and the associated states are denoted by
\[
    \ket{u_i^{(1)}}=\ket{\psi_i},
    \qquad
    \ket{u_j^{(2)}}=\ket{\phi_j}.
\]
For \(k\ge2\), define the \(k\)-collision mass and the corresponding collision scale by
\begin{equation}
    C_k(r):=\sum_i r_i^k,
    \qquad
    R_k(r):=C_k(r)^{-1/k}.
    \label{eq: nonuniform collision mass}
\end{equation}
For the uniform distribution on \(N\) labels, \(C_k=N^{1-k}\) and \(R_k=N^{1-1/k}\).

\begin{definition}[\textbf{Weight-regular nonuniform ensembles}]
    \label{def: weight regular nonuniform}
    A family of nonuniform ensembles is called weight-regular if there exist constants
    \(0<c_0\le C_0<\infty\), independent of \(N_1,N_2\), such that
    \begin{equation}
        \frac{c_0}{N_1}\le p_i\le \frac{C_0}{N_1},
        \qquad
        \frac{c_0}{N_2}\le q_j\le \frac{C_0}{N_2}
    \end{equation}
    for all labels \(i,j\).
\end{definition}

\subsection{General-measurement sample complexity of nonuniform MMD-\(k\)}

We first state the clean nonuniform analogue of Theorems~\ref{thm:mmdk-lower-fullygeneral} and~\ref{thm:mmdk-upper}. The theorem assumes weight regularity; without such a condition, the support sizes \(N_1,N_2\) alone do not determine the sample complexity.

\begin{theorem}[\textbf{MMD-\(k\) for weight-regular nonuniform ensembles}]
    \label{thm: nonuniform weight regular MMD general}
    Fix \(k\ge2\). Consider two weight-regular pure-state ensembles
    \[
        \scE_1=\{(p_i,\ket{\psi_i})\}_{i=1}^{N_1},
        \qquad
        \scE_2=\{(q_j,\ket{\phi_j})\}_{j=1}^{N_2}.
    \]
    Assume the weights are known. Then, in the general measurement model, MMD-\(k\) can be estimated to additive error \(\epsilon\) and failure probability \(\delta\) using
    \begin{widetext}
\begin{equation}
        M
        =
        \scO_{k,c_0,C_0}\!\left(
            \left(k!\epsilon^{-2}\log\frac1\delta\right)^{1/k}
            \max\{N_1^{1-1/k},N_2^{1-1/k}\}
        \right)
        \label{eq: nonuniform MMD upper regular}
    \end{equation}
\end{widetext}
    samples in the sparse-collision regime
    \[
        k!\epsilon^{-2}\log(1/\delta)\lesssim \min\{N_1,N_2\}.
    \]
    Conversely, for fixed \(k\), there exist constants \(\epsilon_0>0\), \(\delta_0\in(0,1/2)\), and \(c_k>0\), depending only on \(k,c_0,C_0\), and, for all sufficiently large \(\max\{N_1,N_2\}\), a finite hard family of weight-regular nonuniform pure ensembles such that any fully general protocol estimating MMD-\(k\) to additive error \(\epsilon_0\) and failure probability \(\delta_0\) must use
    \begin{equation}
        M
        \ge
        c_k\max\{N_1^{1-1/k},N_2^{1-1/k}\}.
        \label{eq: nonuniform MMD lower regular}
    \end{equation}
    Hence, when \(k,c_0,C_0,\epsilon,\delta\) are fixed independently of the support sizes, with \(0<\epsilon\le\epsilon_0\) and \(0<\delta<1/2\), the sparse-collision regime gives
    \begin{equation}
        M
        =
        \Theta_{k,c_0,C_0,\epsilon,\delta}\!\left(
            \max\{N_1^{1-1/k},N_2^{1-1/k}\}
        \right)
    \end{equation}
    in the scaling with the support sizes.
\end{theorem}

\begin{proof}
We first prove the upper bound. It suffices to estimate each term
\[
\begin{gathered}
\bar F^{(k)}(\scE_a,\scE_b)
    =
    \sum_{i=1}^{N_a}\sum_{j=1}^{N_b}
    r_i^{(a)}r_j^{(b)}
    \Abs{\braket{u_i^{(a)}}{u_j^{(b)}}}^{2k},
    \\
    (a,b)\in\{(1,1),(2,2),(1,2)\}.
\end{gathered}
\]
Because the weights are weight-regular and known, we can convert samples from \(\scE_a\) into uniform-label samples over its support with only constant overhead. Namely, sample from \(\scE_a\), obtaining label \(i\), and accept the sample with probability
\begin{equation}
    \frac{r_{\min}^{(a)}}{r_i^{(a)}},
    \qquad
    r_{\min}^{(a)}:=\min_i r_i^{(a)}.
\label{eq:app-nonuniform-rejection-acceptance}
\end{equation}
The accepted label is uniform on \([N_a]\), since
\[
    \Pr(i\mid\mathrm{accepted})
    \propto
    r_i^{(a)}\frac{r_{\min}^{(a)}}{r_i^{(a)}}
    =
    r_{\min}^{(a)}.
\]
The total acceptance probability is \(N_ar_{\min}^{(a)}\ge c_0\), so this rejection step costs only a constant factor in samples.

After rejection, the labels in a batch for \((\scE_a,\scE_b)\) are uniform over
\([N_a]\times[N_b]\). Define
\begin{equation}
\begin{gathered}
X_{ij}^{(ab)}
    :=
    \Abs{\braket{u_i^{(a)}}{u_j^{(b)}}}^2,
    \\
    Y_{ij}^{(ab)}
    :=
    N_aN_b\,r_i^{(a)}r_j^{(b)}\left(X_{ij}^{(ab)}\right)^k.
\end{gathered}
\label{eq:app-nonuniform-weighted-pair-observable}
\end{equation}
By weight regularity,
\[
    0\le Y_{ij}^{(ab)}\le C_0^2 .
\]
Moreover,
\[
    \bar F^{(k)}(\scE_a,\scE_b)
    =
    \frac1{N_aN_b}\sum_{i=1}^{N_a}\sum_{j=1}^{N_b}Y_{ij}^{(ab)}.
\]
We now apply the across-label SWAP construction from the proof of Theorem~\ref{thm:mmdk-upper} to the accepted uniform-label samples. The only difference is that the measured \(k\)-copy SWAP outcome for a selected pair \((i,j)\) is multiplied by the deterministic weight factor
\[
    N_aN_b\,r_i^{(a)}r_j^{(b)},
\]
which remains bounded by \(C_0^2\). Therefore the same row/column collision and concentration argument applies with constants depending only on \(C_0\) and \(k\). For one term \(\bar F^{(k)}(\scE_a,\scE_b)\), the number of accepted paired samples needed is
\[
    \scO_{k,C_0}\!\left(
        \left(k!\epsilon^{-2}\log\frac1\delta\right)^{1/k}
        \max\{N_a^{1-1/k},N_b^{1-1/k}\}
    \right).
\]
The rejection step introduces only a constant factor depending on \(c_0\). Applying the estimate independently to the three terms
\[
    \bar F^{(k)}(\scE_1,\scE_1),
    \qquad
    \bar F^{(k)}(\scE_2,\scE_2),
    \qquad
    \bar F^{(k)}(\scE_1,\scE_2),
\]
and taking a union bound proves Eq.~(\ref{eq: nonuniform MMD upper regular}).

We next prove the lower bound using the equatorial-qubit construction from Appendix~\ref{app: proof lower general}, with the physical sampling weights retained explicitly. Assume first that \(N_2\ge N_1\), and fix any two weight vectors \(p,q\) satisfying the stated regularity conditions. Define
\[
    |\psi_\theta\rangle:=\frac{|0\rangle+e^{i\theta}|1\rangle}{\sqrt2},
    \qquad \rho_\theta:=|\psi_\theta\rangle\langle\psi_\theta|,
\]
and choose distinct offsets
\[
\begin{gathered}
\alpha_i:=\frac{\pi i}{3k(N_1+1)},\quad i\in[N_1],
    \\
    \beta_j:=\frac{\pi j}{3k(N_2+1)},\quad j\in[N_2].
\end{gathered}
\]
Keep the first ensemble fixed as
\begin{equation}
    \mathcal E_1:=\{(p_i,\rho_{\alpha_i})\}_{i=1}^{N_1}.
\label{eq:app-nonuniform-hard-reference-ensemble}
\end{equation}
Under hypothesis \(H_b\), \(b\in\{0,1\}\), draw the hidden indices \(\ell_j\) independently and uniformly from \(\{0,\ldots,k-1\}\), once and for all, and set
\begin{equation}
    \sigma_{b,j}:=\rho_{\beta_j+(2\ell_j+b)\pi/k},
    \qquad
    \mathcal E_{2,b}:=\{(q_j,\sigma_{b,j})\}_{j=1}^{N_2}.
\label{eq:app-nonuniform-hard-comparison-ensemble}
\end{equation}
Both hard priors have finite support. Every realized ensemble has distinct pure components, since the nonzero differences between its offsets have magnitude less than \(\pi/(3k)\), whereas equality of two comparison states would require an integer multiple of \(2\pi/k\). The known physical weights are \(p_i,q_j\) under both hypotheses.

For label \(j\), let
\[
    \Gamma_{b,j}^{(s)}:=\mathbb E_b[\sigma_{b,j}^{\otimes s}],
    \qquad \Gamma_{b,j}^{(0)}:=1.
\]
The finite Fourier-sum calculation in Appendix~\ref{app: proof lower general} gives
\begin{equation}
    \Gamma_{0,j}^{(s)}=\Gamma_{1,j}^{(s)}\quad(0\le s<k),
    \qquad
    \Gamma_{b,j}^{(k)}=G_k+(-1)^bA_j,
\label{eq:app-nonuniform-hard-tensor-moments}
\end{equation}
where
\begin{widetext}
\[
    G_k:=2^{-k}\!\sum_{\substack{x,y\in\{0,1\}^k\\|x|=|y|}}|x\rangle\langle y|,
    \qquad
    A_j:=2^{-k}\!\left(e^{ik\beta_j}|1^k\rangle\langle0^k|
                    +e^{-ik\beta_j}|0^k\rangle\langle1^k|\right).
\]
\end{widetext}
These are equalities of quantum tensor moments, so they apply directly to arbitrary measurements.

To verify separation of the weighted MMD-\(k\) values, write
\[
\begin{gathered}
R:=\sum_i p_i\rho_{\alpha_i}^{\otimes k},
    \\
    \widehat M_b:=\sum_jq_j\sigma_{b,j}^{\otimes k},
    \\
    \overline A_q:=\sum_jq_jA_j,
\end{gathered}
\]
\[
\begin{gathered}
\overline M_b:=\mathbb E_b\widehat M_b=G_k+(-1)^b\overline A_q,
    \\
    D_b:=\|R-\widehat M_b\|_{\mathrm{HS}}^2,
    \\
    \overline D_b:=\|R-\overline M_b\|_{\mathrm{HS}}^2.
\end{gathered}
\]
Set \(a_p:=\sum_i p_i e^{ik\alpha_i}\) and \(a_q:=\sum_jq_j e^{ik\beta_j}\). Expanding the two squared norms and using the two nonzero entries of \(\overline A_q\) yields
\begin{align}
\overline D_1-\overline D_0
    &=4\operatorname{Tr}[(R-G_k)\overline A_q]\notag\\
    &=2^{3-2k}\operatorname{Re}(a_p\overline{a_q})\notag\\
    &=2^{3-2k}\sum_{i,j}p_iq_j\cos(k(\alpha_i-\beta_j))\notag\\
    &\ge 2^{2-2k}=:g_k>0.
\label{eq:app-nonuniform-mmd-reference-gap}
\end{align}
The last inequality uses \(|k(\alpha_i-\beta_j)|<\pi/3\) and \(\sum_{i,j}p_iq_j=1\); thus the reference-value gap does not shrink with either support size.

Independence of the hidden states gives
\[
\begin{aligned}
\mathbb E_b\|\widehat M_b-\overline M_b\|_{\mathrm{HS}}^2
    &=\sum_jq_j^2\left(1-\operatorname{Tr}[(\Gamma_{b,j}^{(k)})^2]\right)\\
    &\le\sum_jq_j^2\le\frac{C_0}{N_2}.
\end{aligned}
\]
Since all moment operators are density operators,
\[
    |D_b-\overline D_b|
    \le 4\|\widehat M_b-\overline M_b\|_{\mathrm{HS}}.
\]
Consequently,
\begin{equation}
    \Pr_b(|D_b-\overline D_b|>g_k/8)
    \le \frac{1024C_0}{N_2g_k^2}=:\zeta_{N_2},
\label{eq:app-nonuniform-mmd-gap-concentration}
\end{equation}
which is at most \(1/16\) for all sufficiently large \(N_2\).

Now presample \(M\) labeled states from each ensemble, as in Appendix~\ref{app: proof lower general}; arbitrary adaptive protocols using at most \(M\) samples from each ensemble can be simulated from these streams. The law of the complete numerical label transcript \(L=((I_t,J_t))_{t=1}^M\) is the same under both hypotheses, namely
\[
    \lambda(L)=\prod_{t=1}^M p_{I_t}q_{J_t}.
\]
Let \(B_j(L):=|\{t:J_t=j\}|\). Conditional on \(L\), the prior-averaged quantum input is, up to a common permutation of registers,
\[
    \left(\bigotimes_{t=1}^M\rho_{\alpha_{I_t}}\right)
    \otimes\left(\bigotimes_{j=1}^{N_2}\Gamma_{b,j}^{(B_j(L))}\right).
\]
The tensor-moment identities make these inputs identical whenever every \(B_j<k\). Therefore the output laws \(P_0,P_1\) of any joint or adaptive measurement satisfy
\begin{equation}
\begin{aligned}
\operatorname{TV}(P_0,P_1)
    &\le\Pr\!\left[\max_j B_j\ge k\right]\\
    &\le\binom{M}{k}C_k(q)
    \le\frac{M^k}{k!}\,C_0^{k-1}N_2^{1-k}.
\end{aligned}
\label{eq:app-nonuniform-output-tv-bound}
\end{equation}
Here \(C_k(q)\le(\max_jq_j)^{k-1}\sum_jq_j\le C_0^{k-1}N_2^{1-k}\).

Set \(\epsilon_0:=g_k/8\) and \(\delta_0:=1/16\). If an estimator is \((\epsilon_0,\delta_0)\)-accurate for every ensemble pair in the two finite prior supports, thresholding its output at \((\overline D_0+\overline D_1)/2\) gives testing errors at most \(\delta_0+\zeta_{N_2}\le1/8\) under each hypothesis. This averages the pointwise estimation guarantee over the original priors without conditioning them, preserving their label-wise independence. Hence
\[
    \frac34\le\operatorname{TV}(P_0,P_1)
    \le\frac{M^k}{k!}\,C_0^{k-1}N_2^{1-k},
\]
and therefore
\begin{equation}
    M\ge\left(\frac{3k!}{4C_0^{k-1}}\right)^{1/k}N_2^{1-1/k}.
\label{eq:app-nonuniform-mmd-explicit-lower-bound}
\end{equation}
For \(N_1\ge N_2\), interchange the roles of the two ensembles and place the hidden states on the first ensemble's labels. This proves Eq.~(\ref{eq: nonuniform MMD lower regular}), with constants independent of the support sizes, already in dimension two. The same support-size scaling holds when \(\epsilon\le\epsilon_0\) and \(\delta\le\delta_0\). More generally, for any fixed \(\delta<1/2\), take the larger support size sufficiently large that its concentration error is at most \((1-2\delta)/4\). The same threshold test then gives \(\operatorname{TV}(P_0,P_1)\ge(1-2\delta)/2>0\), yielding the stated scaling with a constant that may also depend on \(\delta\).
\end{proof}

For arbitrary weights, the collision argument remains valid under explicit quantum tensor-moment matching. A distance gap is an additional requirement for turning it into an estimation lower bound.

\begin{proposition}[\textbf{Collision-scale obstruction for arbitrary weights}]
    \label{prop: arbitrary nonuniform collision lower}
    Fix \(k\ge2\) and a known weight vector \(r\). Under each hypothesis \(H_b\), draw the pure state \(\sigma_{b,i}\) of each label \(i\) independently, once and for all, with physical sampling weight \(r_i\). Keep the other ensemble fixed and common to the two hypotheses. Suppose that, for every label,
    \begin{equation}
        \mathbb E_0[\sigma_{0,i}^{\otimes s}]
        =\mathbb E_1[\sigma_{1,i}^{\otimes s}]
        \qquad(0\le s<k).
\label{eq:app-arbitrary-weights-moment-matching}
\end{equation}
    Then any fully general measurement protocol using at most \(M\) samples from each ensemble has output laws satisfying
    \begin{equation}
        \operatorname{TV}(P_0,P_1)
        \le
        \binom{M}{k}C_k(r),
        \qquad
        C_k(r)=\sum_i r_i^k.
        \label{eq: arbitrary nonuniform TV collision}
    \end{equation}
    Suppose additionally that the actual MMD-\(k\) value \(D_b\) obeys
    \begin{equation}
        \Pr_b(|D_b-d_b|>\eta)\le\zeta,
        \qquad d_1-d_0=\Delta>0,
\label{eq:app-arbitrary-weights-gap-assumption}
\end{equation}
    for deterministic reference values \(d_0,d_1\). If \(\epsilon+\eta<\Delta/2\) and \(\delta+\zeta<1/2\), an estimator that is \((\epsilon,\delta)\)-accurate uniformly over the hard family requires
    \begin{equation}
        M\ge\left(\frac{k!(1-2\delta-2\zeta)}{C_k(r)}\right)^{1/k}.
    \end{equation}
    In particular, when these gap and error parameters are fixed independently of the support size, this yields
    \begin{equation}
        M=\Omega_k\!\left(C_k(r)^{-1/k}\right)
         =\Omega_k(R_k(r)),
\label{eq:app-arbitrary-weights-collision-scale}
\end{equation}
    with the hidden constant also depending on the fixed error parameters.
\end{proposition}

\begin{proof}
Presampling reduces the protocol to a measurement on two complete streams of length \(M\). The label law is common to both hypotheses. Conditional on the labels, independence factors the averaged hidden-state input into the operators \(\mathbb E_b[\sigma_{b,i}^{\otimes B_i}]\), where \(B_i\) is the number of occurrences of hidden label \(i\). The assumed tensor-moment identities make the two inputs equal whenever \(\max_iB_i<k\). Thus any measurement's output distributions can differ only through this collision event, giving
\[
\begin{aligned}
\operatorname{TV}(P_0,P_1)
    &\le\Pr\!\left[\max_iB_i\ge k\right]\\
    &\le\mathbb E\sum_i\binom{B_i}{k}
    =\binom{M}{k}\sum_i r_i^k.
\end{aligned}
\]
The equality counts all \(k\)-tuples of sampling positions; each has probability \(\sum_i r_i^k\) of carrying one common label. This proves Eq.~(\ref{eq: arbitrary nonuniform TV collision}).

Under the gap assumptions, thresholding an accurate estimator at \((d_0+d_1)/2\) has error at most \(\delta+\zeta\) under each hypothesis. Averaging the uniform estimation guarantee over the unconditioned priors therefore gives
\[
    1-2\delta-2\zeta
    \le\operatorname{TV}(P_0,P_1)
    \le\frac{M^k}{k!}C_k(r),
\]
which yields the claimed lower bound.
\end{proof}

Thus \(R_k(r)\) is the collision scale for hard families with the stated tensor-moment matching and separation properties. The weight-regular construction above verifies both properties. For a concentrated weight profile, the concentration and distance gap must be checked separately; in particular, \(R_k(r)\) can be \(\Theta(1)\), so the support size alone cannot give a growing lower bound.

\subsection{Label-local SWAP-test Wasserstein baseline for nonuniform ensembles}

We now record the nonuniform extension of the label-local SWAP-test Wasserstein estimator. This result concerns the SWAP-test transcript model of Definition~\ref{def: sample complexity}, not the fully general measurement model.

\begin{theorem}[\textbf{Label-local SWAP-test sample complexity of Wasserstein distance between nonuniform ensembles}]
    \label{thm: nonuniform wasserstein swap}
    Considering two pure quantum ensembles
    \[
        \scE_1=\{(p_i,\ket{\psi_i})\}_{i=1}^{N_1},
        \qquad
        \scE_2=\{(q_j,\ket{\phi_j})\}_{j=1}^{N_2},
    \]
    define
    \[
        \omega_{\min}:=\min_{i,j}p_iq_j .
    \]
    The sample complexity to estimate the Wasserstein distance between the two ensembles using the same label-local SWAP-test estimation process as in Section~\ref{secsub: sc of Wass}, together with empirical estimates of the weights \(p_i,q_j\), is
    \begin{widetext}
\begin{equation}
        M
        =
        \scO\Bigg(
            \frac{1}{\omega_{\min}}
            \left(
                \frac{1}{\epsilon^2}\log\frac{N_1N_2}{\delta}
                +
                \frac{N_1N_2}{\delta}
            \right)
            +
            \frac{N_1+N_2+\log(1/\delta)}{\epsilon^2}
        \Bigg),
        \label{eq: nonuniform wasserstein local bound}
    \end{equation}
\end{widetext}
    where \(\epsilon\) is the additive error and \(\delta\) is the failure probability.
\end{theorem}

\begin{proof}
Let
\[
    C_{ij}=1-\Abs{\braket{\psi_i}{\phi_j}}^2,
    \qquad
    \widehat C_{ij}=1-\widehat X_{ij},
\]
where \(\widehat X_{ij}\) is the empirical average of the unbiased SWAP-test outcomes \(R\in\{-1,+1\}\) with pair label \((i,j)\). For fixed weights \(p,q\), a sufficient condition for
\[
    \Abs{W(\widehat C;p,q)-W(C;p,q)}\le \epsilon
\]
is
\[
    \max_{i,j}\Abs{\widehat C_{ij}-C_{ij}}\le \epsilon .
\]
Indeed, if \(\|\widehat C-C\|_\infty\le\epsilon\), then for every feasible coupling \(P\),
\[
    \Abs{\langle P,\widehat C-C\rangle}\le \epsilon
\]
because \(\sum_{i,j}P_{ij}=1\). Comparing the optimal values for \(C\) and \(\widehat C\) gives the claim.

Suppose a pair label \((i,j)\) has appeared \(T_{ij}\) times. Hoeffding's inequality gives
\[
    \Pr\!\left(
        \Abs{\widehat X_{ij}-X_{ij}}\ge \eta
    \right)
    \le
    2\exp(-T_{ij}\eta^2/2),
\]
and hence the same bound holds for \(\Abs{\widehat C_{ij}-C_{ij}}\). Therefore, by a union bound over the \(N_1N_2\) pair labels, it suffices to have
\begin{equation}
    T_{\min}:=\min_{i,j}T_{ij}
    \ge
    \frac{2}{\eta^2}\log\frac{2N_1N_2}{\delta}.
\label{eq:app-nonuniform-wasserstein-pair-count}
\end{equation}
Since each pair label \((i,j)\) appears with probability \(p_iq_j\ge\omega_{\min}\), a Chernoff lower-tail bound and a union bound imply that \(T_{\min}\) exceeds the above threshold with probability at least \(1-\delta\) once
\[
    M
    =
    \scO\!\left(
        \frac1{\omega_{\min}}
        \left[
            \frac1{\eta^2}\log\frac{N_1N_2}{\delta}
            +
            \log\frac{N_1N_2}{\delta}
        \right]
    \right).
\]
This is bounded by the first term in Eq.~(\ref{eq: nonuniform wasserstein local bound}).

It remains to account for unknown weights. Let \(\widehat p,\widehat q\) be the empirical label frequencies. The dual formulation of Wasserstein distance with cost \(0\le C_{ij}\le1\) gives
\begin{equation}
    \Abs{W(C;\widehat p,\widehat q)-W(C;p,q)}
    \le
    \|\widehat p-p\|_1+\|\widehat q-q\|_1.
\label{eq:app-nonuniform-wasserstein-weight-stability}
\end{equation}
Standard multinomial concentration implies
\[
    \|\widehat p-p\|_1+\|\widehat q-q\|_1\le \eta
\]
with probability at least \(1-\delta\) using
\[
    \scO\!\left(
        \frac{N_1+N_2+\log(1/\delta)}{\eta^2}
    \right)
\]
additional samples. Taking \(\eta\) to be a constant multiple of \(\epsilon\) and applying the triangle inequality proves Eq.~(\ref{eq: nonuniform wasserstein local bound}).
\end{proof}

\subsection{Label-local SWAP-test MMD-\(k\) baseline for nonuniform ensembles}

Finally, we record the nonuniform extension of the label-local SWAP-test estimator for MMD-\(k\). This is a useful experimentally simple baseline, but it is not the general-measurement sample complexity: the across-label coherent protocol in Theorem~\ref{thm: nonuniform weight regular MMD general} can be parametrically better.

Fix \((a,b)\in\{(1,1),(1,2),(2,2)\}\). Let
\[
\begin{gathered}
\Omega_{ab}:=[N_a]\times[N_b],
    \\
    w_{ij}^{(ab)}:=r_i^{(a)}r_j^{(b)},
    \\
    X_{ij}^{(ab)}
    :=
    \Abs{\braket{u_i^{(a)}}{u_j^{(b)}}}^2 .
\end{gathered}
\]
Then
\[
    \bar F^{(k)}(\scE_a,\scE_b)
    =
    \sum_{(i,j)\in\Omega_{ab}}
    w_{ij}^{(ab)}
    \left(X_{ij}^{(ab)}\right)^k .
\]
Assume the weights are known. In the label-local SWAP-test model, each trial draws a pair label \(\ell=(i,j)\sim w^{(ab)}\), performs an ordinary SWAP test, and returns
\[
    R_t\in\{-1,+1\},
    \qquad
    \dsE[R_t\mid \ell_t=\ell]=X_\ell^{(ab)} .
\]
Let \(T_\ell\) be the number of times label \(\ell\) appears. For \(T_\ell\ge k\), define the within-label \(U\)-statistic
\begin{equation}
    Z_\ell
    :=
    \binom{T_\ell}{k}^{-1}
    \sum_{1\le t_1<\cdots<t_k\le T_\ell}
    \prod_{r=1}^k R_{\ell,t_r}.
    \label{eq: nonuniform local U statistic}
\end{equation}
Then
\[
    \dsE[Z_\ell\mid T_\ell,\ell]
    =
    \left(X_\ell^{(ab)}\right)^k .
\]
The corresponding importance-corrected label-local estimator is
\begin{equation}
    \widehat{\bar F}^{(k)}_{\rm loc}(\scE_a,\scE_b)
    :=
    \binom{M}{k}^{-1}
    \sum_{\ell\in\Omega_{ab}:T_\ell\ge k}
    \binom{T_\ell}{k}
    \frac{Z_\ell}{(w_\ell^{(ab)})^{k-1}} .
    \label{eq: nonuniform local estimator}
\end{equation}
Equivalently, Eq.~(\ref{eq: nonuniform local estimator}) averages over all \(k\)-tuples of trials whose pair labels coincide, with the factor \((w_\ell^{(ab)})^{1-k}\) correcting the collision bias.

For $(a,b)\in\{(1,1),(1,2),(2,2)\}$ define
\begin{equation}
    S_{ab}(k)
    :=
    \sum_{\ell\in\Omega_{ab}}
    \left(w^{(ab)}_\ell\right)^{2-k}.
    \label{eq: Sab_def}
\end{equation}
For product-form weights, \(S_{ab}(k)\) factorizes:
\begin{equation}
\begin{gathered}
S_{12}(k)=
    \left(\sum_{i=1}^{N_1} p_i^{\,2-k}\right)
    \left(\sum_{j=1}^{N_2} q_j^{\,2-k}\right),
    \\
    S_{11}(k)=
    \left(\sum_{i=1}^{N_1} p_i^{\,2-k}\right)^2,
    \\
    S_{22}(k)=
    \left(\sum_{j=1}^{N_2} q_j^{\,2-k}\right)^2 .
    \label{eq: Sab_factor}
\end{gathered}
\end{equation}
For two uniform ensembles with the same number of states, \(S_{ab}(k)\) reduces to \(N^{2k-2}\).

\begin{proposition}[\textbf{Label-local nonuniform MMD-\(k\) baseline}]
    \label{prop: nonuniform local SWAP baseline}
    Fix \(k\ge2\). For \((a,b)\in\{(1,1),(1,2),(2,2)\}\), the estimator Eq.~(\ref{eq: nonuniform local estimator}) is unbiased:
    \begin{equation}
        \dsE\!\left[
            \widehat{\bar F}^{(k)}_{\rm loc}(\scE_a,\scE_b)
        \right]
        =
        \bar F^{(k)}(\scE_a,\scE_b).
\label{eq:app-nonuniform-local-estimator-unbiasedness}
\end{equation}
    Moreover, for fixed \(k\),
    \begin{equation}
\begin{gathered}
\operatorname{Var}\!\left[
            \widehat{\bar F}^{(k)}_{\rm loc}(\scE_a,\scE_b)
        \right]
        \le
        C_k
        \sum_{r=1}^k
        \frac{S_{ab,r}}{M^r},
        \\
        S_{ab,r}:=
        \sum_{\ell\in\Omega_{ab}}
        \left(w_\ell^{(ab)}\right)^{2-r},
        \label{eq: nonuniform local variance}
\end{gathered}
\end{equation}
    where \(C_k\) depends only on \(k\). Since \(S_{ab,r}^{1/r}\le S_{ab}(k)^{1/k}\) for \(1\le r\le k\), constant additive accuracy and constant success probability are achieved using
    \begin{equation}
        M_{ab}
        =
        \scO_k\!\left(
            S_{ab}(k)^{1/k}
        \right)
        \label{eq: nonuniform local scale k collision}
    \end{equation}
    samples for one \(\bar F^{(k)}\) term. Consequently,
    \begin{equation}
        M
        =
        \scO_k\!\left(
            \max\{
                S_{11}(k)^{1/k},
                S_{12}(k)^{1/k},
                S_{22}(k)^{1/k}
            \}
        \right)
        \label{eq: nonuniform local MMD scale}
    \end{equation}
    samples suffice to estimate \(\scD^{(k)}(\scE_1,\scE_2)\) to constant additive accuracy with constant success probability in the label-local SWAP-test model.
\end{proposition}

\begin{proof}
Unbiasedness follows by conditioning on a \(k\)-tuple of trials. For a fixed \(k\)-tuple,
\[
    \Pr(\ell_1=\cdots=\ell_k=\ell)=(w_\ell^{(ab)})^k,
\]
and, conditional on this event,
\[
    \dsE\!\left[\prod_{r=1}^k R_r\right]
    =
    \left(X_\ell^{(ab)}\right)^k .
\]
Multiplying by \((w_\ell^{(ab)})^{1-k}\) and summing over \(\ell\) gives
\[
    \dsE\!\left[
        \widehat{\bar F}^{(k)}_{\rm loc}(\scE_a,\scE_b)
    \right]
    =
    \sum_\ell w_\ell^{(ab)}
    \left(X_\ell^{(ab)}\right)^k
    =
    \bar F^{(k)}(\scE_a,\scE_b).
\]

For the variance, write Eq.~(\ref{eq: nonuniform local estimator}) as the order-\(k\) U-statistic
\[
    \widehat{\bar F}^{(k)}_{\rm loc}
    =
    \binom{M}{k}^{-1}
    \sum_{A\subset[M],\,|A|=k} h_A .
\]
If two \(k\)-subsets \(A,B\) are disjoint, then \(h_A\) and \(h_B\) are independent. If they overlap in exactly \(r\) indices, then the product \(h_Ah_B\) can be nonzero only when all labels in \(A\cup B\) are the same label \(\ell\). Since \(|A\cup B|=2k-r\),
\begin{equation}
\begin{aligned}
\dsE[|h_Ah_B|]
    &\le
    \sum_{\ell\in\Omega_{ab}}
    (w_\ell^{(ab)})^{2k-r}
    (w_\ell^{(ab)})^{-2k+2}\\
    &=
    \sum_{\ell\in\Omega_{ab}}
    (w_\ell^{(ab)})^{2-r}
    =
    S_{ab,r}.
\end{aligned}
\label{eq:app-nonuniform-local-overlap-moment}
\end{equation}
The number of pairs of \(k\)-subsets with intersection size \(r\) is at most \(C_kM^{2k-r}\). Dividing by \(\binom{M}{k}^2=\Theta_k(M^{2k})\) yields
\[
    \operatorname{Var}\!\left[
        \widehat{\bar F}^{(k)}_{\rm loc}
    \right]
    \le
    C_k
    \sum_{r=1}^k
    \frac{S_{ab,r}}{M^r}.
\]
Finally, write
\[
    S_{ab,r}
    =
    \sum_{\ell}
    (w_\ell^{(ab)})^2
    \left(\frac1{w_\ell^{(ab)}}\right)^r.
\]
Since \(\sum_\ell (w_\ell^{(ab)})^2\le 1\), the \(L_r\)-norms with respect to this sub-probability measure are nondecreasing in \(r\), and hence
\[
    S_{ab,r}^{1/r}\le S_{ab,k}^{1/k}=S_{ab}(k)^{1/k}.
\]
Thus \(M=\scO_k(S_{ab}(k)^{1/k})\) makes the variance bounded by a constant. Chebyshev's inequality gives constant additive accuracy and constant success probability. Repeating the construction independently for the three terms \((1,1),(2,2),(1,2)\) and taking a union bound gives Eq.~(\ref{eq: nonuniform local MMD scale}).
\end{proof}

For two uniform ensembles with \(N_1=N_2=N\), Eq.~(\ref{eq: Sab_factor}) gives
\begin{equation}
    S_{ab}(k)=N^{2k-2},
    \qquad
    M_{\rm loc}=\scO_k(N^{2-2/k}),
\label{eq:app-uniform-local-baseline-recovery}
\end{equation}
recovering the label-local SWAP-test scaling. Under weight regularity, the same \(N^{2-2/k}\) scaling holds for the label-local baseline, whereas the coherent across-label protocol achieves the smaller general-measurement scaling \(N^{1-1/k}\).

\section{Auxiliary technical lemmas}
Some technical results needed are collected in this section.
\subsection{Preliminary results on Binomial distributions}
For any integer $m$, we write $[m] := \{1, \dots, m\}$. Fix \(n\in\mathbb N\) and a finite state space \(\{x_1,\dots,x_{n}\}\). Let
\[
  Z_1,\dots,Z_M \in [n] 
\]
be i.i.d.\ samples with
\(\Pr(Z_j = \ell) = 1/n\) for each $\ell\in [n]$ and $j\in [M]$. Define for $\ell\in[n]$,
\begin{equation}
  T_\ell := \sum_{j=1}^M \mathbf 1\{Z_j = \ell \}
\label{eq:app-occupancy-label-count}
\end{equation}
so that \(\sum_{\ell=1}^n T_\ell = M\).
Fix an integer \(k\geq 1\), and define
\begin{equation}
  m := \sum_{\ell=1}^n \mathbf 1\{T_\ell \geq k\}.
\label{eq:app-occupancy-collision-count}
\end{equation}
Thus \(m\) is the number of states that appear at least $k$ times among the $M$ samples. By symmetry, each \(T_\ell\) follows a binomial distribution
\[
  T_\ell \sim \mathrm{Bin}\!\left(M,\frac1n\right),
  \qquad
  \mu = \dsE[T_\ell] = \frac{M}{n}.
\]
We have, 
\begin{align*}
  \dsE [m]
  &= \sum_{\ell=1}^n \dsE\big[ \mathbf 1\{T_\ell \geq k\} \big]
   = \sum_{\ell=1}^n \Pr(T_\ell \geq k) \\
  &= n \,\Pr\!\left(\mathrm{Bin}\!\left(M,\frac1n\right)\geq k\right)\\
  &= n \sum_{j=k}^M \binom{M}{j}\left(\frac1n\right)^j\left(1-\frac1n\right)^{M-j} .
\end{align*}


\subsection{Bound on Expectation $\dsE[m]$}
For $k \geq 1$, we have
\begin{align*}
  \dsE [m]
  & = n \sum_{j=k}^{M}
        \binom{M}{j}
        \left(\frac{1}{n}\right)^j
        \left(1-\frac{1}{n}\right)^{M-j}\\
  & = n \binom{M}{k}
        \left(\frac{1}{n}\right)^k
        \left(1-\frac{1}{n}\right)^{M-k}(1 + R_n),   
\end{align*}
where the remainder
\begin{align*}
  R_n
  := \sum_{j=k+1}^{M}
       \frac{\binom{M}{j}}{\binom{M}{k}}
       \left(\frac1n\right)^{j-k}
       \left(1-\frac1n\right)^{k-j}.
\end{align*}
First note that
\begin{align*}
  \binom{M}{k}
  = \frac{M(M-1)\cdots(M-k+1)}{k!}
  = \Theta \left( \frac{M^k}{k!}\right),
\end{align*}
as well as
\begin{align*}
\left(1-\frac{1}{n}\right)^{M-k}
  &= \exp\left(
      (M-k)\log\left(1-\frac{1}{n}\right)
    \right)\\
  &= \Theta \left( e^{-M/n} \right).
\end{align*}
So we have 
\begin{align*}
  n \binom{M}{n}
    \left(\frac{1}{n}\right)^k
    \left(1-\frac{1}{n}\right)^{M-k}
    = \Theta \left( \frac{M^k e^{-M/n}}{k!\,n^{k-1}} \right).
\end{align*}

Now we deal with the remainder. For \(j\ge k+1\), note that $\frac{\binom{M}{j}}{\binom{M}{k}}\leq \frac{M^{j-k}}{(j-k)!}$. 
Thus,
\begin{align*}
|R_n|&\leq
  \sum_{j=k+1}^{M}
    \frac{M^{j-k}}{(j-k)!}
    \left(\frac1n\right)^{j-k}\\
  &\leq \sum_{r=1}^{\infty}
    \frac{1}{r!} \left( \frac{M}{n} \right)^r\\
  &= e^{M/n}-1 = \Theta \left( \frac{M}{n} \right).
\end{align*}
Combining all the analysis above, we obtain
\begin{align}\label{eqn_ES_bound_order}
\dsE [m]
  = \Theta \left(\frac{M^k e^{-M/n}}{k!\,n^{k-1}} \right) . 
\end{align}
With the condition $M\le cn$, where $c$ is a constant, we have
\begin{equation}
    \dsE [m]
  = \Theta \left(\frac{M^k}{k!\,n^{k-1}} \right).\label{eq: Em M<=cN2}
\end{equation}

\subsection{Negative Association and Chernoff-type inequality}
We recall the notion of negative association \cite{Dubhashi_Ranjan_Negative_Dependence96}.

\begin{definition}[Negative Association]
Let $X := (X_1,\dots,X_n)$ be a vector of random variables. The random variables $X$ are \emph{negatively associated} (NA) if for every two disjoint index sets 
$I, J \subseteq [n]$,
\begin{align}
&\dsE\!\left[ f(X_i,\, i\in I)\, g(X_j,\, j\in J) \right]
  \notag\\
  &\;\le\;
  \dsE\!\left[f(X_i,\, i\in I)\right]\,
  \dsE\!\left[g(X_j,\, j\in J)\right]
\label{eq:app-negative-association-definition}
\end{align}
for all functions $f : \mathbb{R}^{|I|} \to \mathbb{R}$ and 
$g : \mathbb{R}^{|J|} \to \mathbb{R}$ that are both non-decreasing or both non-increasing.
\end{definition}
We will use several standard properties of NA
(see, e.g., \cite{Dubhashi_Ranjan_Negative_Dependence96}).

\begin{lemma}[Basic properties of NA]\label{lem:NA-closure}

\begin{itemize}\quad
  \item If \((X_1,\dots,X_n)\) is NA and
    \(f_i:\mathbb R\to\mathbb R\) are coordinatewise nondecreasing
    functions, then \((f_1(X_1),\dots,f_n(X_n))\) is NA.
  \item If \((X_1,\dots,X_n)\) and \((X'_1,\dots,X'_n)\) are independent
    NA families, then \((X_1+X'_1,\dots,X_n+X'_n)\) is NA.
   \item If \((X_1,\dots,X_n)\) is
    NA, then for any non-decreasing functions $f_i,\,i\in[n]$, 
    \begin{equation}\dsE \prod_{i\in[n]} f_i(X_i) \leq \prod_{i\in[n]}\dsE  f_i(X_i). 
\label{eq:app-negative-association-product-bound}
\end{equation}
\end{itemize}

\end{lemma}

Recall that $(T_1,\dots,T_n)$ is NA by \cite{Dubhashi_Ranjan_Negative_Dependence96} and \(Y_\ell = \mathbf 1\{T_\ell \geq K\}\).
Each \(Y_\ell\) is a coordinatewise nondecreasing function of \(T_\ell\), and
does not depend on other coordinates \(T_{\ell'}\) with \(\ell'\neq \ell\).
Thus we can apply Lemma~\ref{lem:NA-closure} coordinatewise and obtain that the family $(Y_1,\dots,Y_d)$ is negatively associated.

We now derive Chernoff-type bounds for
$S = \sum_{\ell=1}^d Y_\ell$, using the
negative association of the $Y_\ell$'s.

\begin{lemma}[MGF bound for NA Bernoulli sums]\label{lem:mgf-NA}
Let $Y_1,\dots,Y_d$ be negatively associated random variables taking
values in $\{0,1\}$.
Let $S = \sum_{\ell=1}^d Y_\ell$ and $\mu := \dsE S = \sum_\ell \dsE Y_\ell$.
Then, for all $\theta > 0$,
\begin{align}
  \dsE [ e^{\theta S} ]
  \leq \exp\big( (e^{\theta}-1)\mu \big).
\label{eq:app-negative-association-mgf}
\end{align}
\end{lemma}

\begin{proof}
Since the $Y_\ell$'s are NA and $y\mapsto e^{\theta y}$ is nondecreasing for
$\theta>0$, the definition of NA implies
\begin{align*}
  \dsE \Big[\prod_{\ell=1}^d e^{\theta Y_\ell}\Big]
  \le \prod_{\ell=1}^d \dsE [e^{\theta Y_\ell} ].
\end{align*}
The left-hand side is precisely \(\dsE e^{\theta S}\).
For each \(\ell\), since \(Y_\ell\in\{0,1\}\),
\[
  \dsE [ e^{\theta Y_\ell} ]
  = (1-p_\ell)  + p_\ell e^{\theta}
  = 1 + p_\ell(e^{\theta}-1),
  \qquad p_\ell := \dsE [Y_\ell ].
\]
Thus
\[
  \prod_{\ell=1}^d \dsE [ e^{\theta Y_\ell} ]
  = \prod_{\ell=1}^d \big(1 + p_\ell(e^{\theta}-1)\big).
\]
Using \(\log(1+u)\le u\) for \(u>-1\), we get
\begin{align*}
\prod_{\ell=1}^d \dsE [ e^{\theta Y_\ell} ]
  &= \exp \left( \sum_{\ell=1}^d \log(1 + p_\ell(e^{\theta}-1)\big) \right)\\
  &\leq \exp \left( \sum_{\ell=1}^d p_\ell(e^{\theta}-1) \right)\\
  &= e^{(e^{\theta}-1)\mu}.
\end{align*}
Combination with the NA property completes the proof.
\end{proof}

We are now ready for Chernoff-type tail bounds.

\begin{proposition}[Chernoff bounds]\label{prop:Chernoff-S}
Let \(S = \sum_{\ell=1}^n Y_\ell\) with \(Y_\ell = \mathbf 1\{T_\ell\geq k\}\) as above,
and write \(\mu := \dsE [S] \).
Then the following hold:
\begin{enumerate}[label=(\roman*)]
  \item For any \(\epsilon>0\),
  \begin{align}
    \Pr(S \ge (1+\epsilon)\mu)
    \leq \exp\big(- \frac{\epsilon^2}{2+\epsilon} \mu\big).       
  \label{eq:app-negative-association-chernoff-upper}
\end{align}
  \item For any \(\epsilon\in(0,1)\),
  \begin{align}
    \Pr(S \le (1-\epsilon)\mu)
    \leq \exp\big(-\mu\epsilon^2/2\big).  
  \label{eq:app-negative-association-chernoff-lower}
\end{align}
\end{enumerate}
As a result, for $\epsilon\in ( 0, 1)$,
\begin{align}
     \Pr( |S - \mu| \ge \epsilon \mu)
    \leq 2\exp\big(- \mu\epsilon^2/3 \big).       
\label{eq:app-negative-association-chernoff-two-sided}
\end{align}
\end{proposition}

\begin{proof}
The proofs follow the standard Chernoff arguments, using the moment generating function bound in Lemma~\ref{lem:mgf-NA} instead of independence.

For (i), fix \(\epsilon>0\) and set \(a = (1+\epsilon)\mu\).
By Markov's inequality and Lemma~\ref{lem:mgf-NA},
for any \(\theta>0\),
\begin{align*}
\Pr(S \geq a)
  &\le e^{-\theta a}\,\dsE [e^{\theta S}]\\
  &\le e^{-\theta a} \exp\big( (e^{\theta}-1)\mu \big)\\
  &= \exp\big( -\theta a + (e^{\theta}-1)\mu \big).
\end{align*}
Taking $\theta = \log ( 1 + \epsilon )$ yields that 
\begin{align*}
&\log \Pr(S \geq (1+\epsilon)\mu )
   \\
  &\leq \mu(\epsilon - (1+\epsilon) \log(1+\epsilon) ) \\
  &\leq -\frac{\epsilon^2}{ 2 + \epsilon }\mu.
\end{align*}
Here for the last inequality, we use the fact that $\log( 1+x ) \geq \frac{ x }{ 1 + x/2}$ for all $x>0$.

The proof for (ii) follows the same procedure but takes $\theta = \log(1-\epsilon)$.
\end{proof}

\subsection{The minimum number of samples}\label{appsub: minM to reach T}
In this subsection, we consider the number of total samples $M$ needed to achieve $T_{\min} \ge t$, that is, with a high probability, the number of samples for every label is no less than $t$. Here we consider the general case with nonuniform probability distribution.

Let each sample fall into label $\ell\in[n]$ with probability $p_\ell>0$, $\sum_{\ell=1}^n p_\ell=1$. After drawing $M$ i.i.d.\ samples, the count vector $(T_1,\dots,T_n)$ is multinomial, and for each $\ell$ the marginal satisfies
\begin{equation*}
    T_\ell \sim \mathrm{Bin}(M,p_\ell),\qquad \mu_\ell:=\dsE[T_\ell]=Mp_\ell.
\end{equation*}
We aim to ensure $\Pr(T_{\min} < t)\le \delta$, where $T_{\min}:=\min_{\ell\in[n]}T_\ell$. By the union bound,
\begin{equation*}
    \Pr(T_{\min} < t)\le \sum_{\ell=1}^n \Pr(T_\ell < t).
\end{equation*}
Fix $\ell$ and assume $\mu_\ell \ge t$. Let $\eta_\ell\in[0,1]$ be defined by $t=(1-\eta_\ell)\mu_\ell$, i.e.
\begin{equation*}
    \eta_\ell = 1-\frac{t}{\mu_\ell}=1-\frac{t}{Mp_\ell}.
\end{equation*}
Applying the Chernoff lower-tail bound as Proposition \ref{prop:Chernoff-S} for a binomial random variable,
\begin{equation*}
    \Pr\!\big(T_\ell \le (1-\eta_\ell)\mu_\ell\big)\le \exp\!\Big(-\frac{\eta_\ell^2\mu_\ell}{2}\Big),
\end{equation*}
we obtain
\begin{equation}
    \Pr(T_\ell < t)\le \exp\!\Big(-\frac{(Mp_\ell-t)^2}{2Mp_\ell}\Big).
\end{equation}
Therefore,
\begin{equation}
    \Pr(T_{\min} < t)\le \sum_{\ell=1}^n \exp\!\Big(-\frac{(Mp_\ell-t)^2}{2Mp_\ell}\Big).
    \label{eq:nonunif_Tmin_union_chernoff}
\end{equation}
A convenient sufficient condition is enforce
$\Pr(T_\ell<t)\le \delta_\ell$ for $\{\delta_\ell\}_{\ell=1}^n$ such that $\delta_\ell\in(0,1)$ and $\sum_{\ell=1}^n \delta_\ell \le \delta$. Namely, it suffices that
\begin{equation*}
    \exp\!\Big(-\frac{(Mp_\ell-t)^2}{2Mp_\ell}\Big)\le \delta_\ell
    \quad\Longleftrightarrow\quad
    \frac{(Mp_\ell-t)^2}{2Mp_\ell}\ge \log\frac{1}{\delta_\ell}.
\end{equation*}
Let $L_\ell:=\log\frac{1}{\delta_\ell}$ and set $\mu_\ell=Mp_\ell$. The inequality becomes
\begin{equation*}
    \frac{(\mu_\ell-t)^2}{2\mu_\ell}\ge L_\ell
    \quad\Longleftrightarrow\quad
    \mu_\ell^2-2(t+L_\ell)\mu_\ell+t^2\ge 0,
\end{equation*}
which holds whenever $\mu_\ell\ge (t+L_\ell)+\sqrt{(t+L_\ell)^2-t^2}=(t+L_\ell)+\sqrt{L_\ell^2+2tL_\ell}$. Hence it suffices that, for every $\ell$,
\begin{equation*}
    Mp_\ell \ge t+L_\ell+\sqrt{L_\ell^2+2tL_\ell}.
\end{equation*}
Equivalently,
\begin{equation}
\begin{gathered}
M \ge \max_{\ell\in[n]}\frac{t+L_\ell+\sqrt{L_\ell^2+2tL_\ell}}{p_\ell},
    \\ L_\ell=\log\frac{1}{\delta_\ell},\quad \sum_{\ell=1}^n \delta_\ell\le \delta.
    \label{eq:nonunif_Tmin_M_bound}
\end{gathered}
\end{equation}
In particular, taking $\delta_\ell=\delta/n$ yields $L_\ell=L:=\log(n/\delta)$ and
\begin{equation}
    M \ge \frac{t+L+\sqrt{L^2+2tL}}{p_{\min}},\qquad p_{\min}:=\min_{\ell\in[n]}p_\ell. \label{eq: M needed to achieve Tmin}
\end{equation}
Since $\sqrt{L^2+2tL} <t+L$, we always use $M\ge \frac{t+L}{p_{\min}}$.

Eq.~(\ref{eq: M needed to achieve Tmin}) gives an upper bound to achieve $T_{\min} \ge t$ (a sufficient condition), we also consider the lower bound.

Let $\ell^\star\in[n]$ be a label attaining $p_{\min}$, i.e., $p_{\ell^\star}=p_{\min}$.
Since the event $\{T_{\min}\ge t\}$ implies $\{T_{\ell^\star}\ge t\}$, we have
\begin{equation}
\begin{gathered}
\Pr(T_{\min}\ge t)\ \le\ \Pr(T_{\ell^\star}\ge t),
    \\
    \Pr(T_{\min}<t)\ \ge\ \Pr(T_{\ell^\star}<t).
\end{gathered}
\end{equation}
In particular, if $\Pr(T_{\min}<t)\le \delta$, then necessarily $\Pr(T_{\ell^\star}<t)\le \delta$.
Write $\mu_{\min}:=\dsE[T_{\ell^\star}]=Mp_{\min}$ and define $\eta\ge 0$ by
$t=(1+\eta)\mu_{\min}$ (equivalently $\eta=\frac{t}{Mp_{\min}}-1$).
When $\mu_{\min}\le t$ (i.e., $\eta\ge 0$), applying the  Chernoff upper-tail bound as Proposition \ref{prop:Chernoff-S} for a binomial random variable gives
\begin{equation}
\begin{aligned}
\Pr(T_{\ell^\star}\ge t)
    &\ =\ \Pr\!\big(T_{\ell^\star}\ge (1+\eta)\mu_{\min}\big)\\
    &\ \le\ \exp\!\Big(-\frac{\eta^2\mu_{\min}}{2+\eta}\Big)\\
    &\ =\ \exp\!\Big(-\frac{(t-Mp_{\min})^2}{Mp_{\min}+t}\Big).
\end{aligned}
\end{equation}
Therefore, a necessary condition for $\Pr(T_{\min}<t)\le \delta$ is
\begin{equation}
    1-\delta\ \le\ \Pr(T_{\ell^\star}\ge t)
    \ \le\ \exp\!\Big(-\frac{(t-Mp_{\min})^2}{Mp_{\min}+t}\Big),
\end{equation}
which implies
\begin{equation}
    \frac{(t-Mp_{\min})^2}{Mp_{\min}+t}\ \le\ \log\frac{1}{1-\delta}.
\end{equation}
Let $u:=Mp_{\min}$. Rearranging yields the quadratic inequality
\begin{equation}
\begin{aligned}
&(t-u)^2\ \le\ (u+t)\log\frac{1}{1-\delta}
    \\
    &\Longleftrightarrow\quad
    u^2-(2t+\Lambda)u+(t^2-\Lambda t)\ \le\ 0,
    \\
    &\qquad \Lambda:=\log\frac{1}{1-\delta}.
\end{aligned}
\end{equation}
Hence it is necessary that
\begin{equation}
    u \ \ge\ (t+\tfrac{\Lambda}{2})-\frac{1}{2}\sqrt{\Lambda^2+8t\Lambda},
\end{equation}
and thus
\begin{equation}
    M\ \ge\ \frac{(t+\tfrac{\Lambda}{2})-\frac{1}{2}\sqrt{\Lambda^2+8t\Lambda}}{p_{\min}},
    \qquad \Lambda=\log\frac{1}{1-\delta}.
    \label{eq:nonunif_Tmin_lower}
\end{equation}
In particular, for any constant $\delta\in(0,1)$, $\Lambda=\Theta(1)$, and Eq.~(\ref{eq:nonunif_Tmin_lower})
implies the scaling lower bound
\begin{equation}
    M\ =\ \Omega\!\left(\frac{t}{p_{\min}}\right).
    \label{eq:nonunif_Tmin_lower_simplified}
\end{equation}

A second lower bound comes from the coupon-collector obstruction. Let
\begin{equation*}
    S:=\{\ell\in[n]:\,p_\ell\le 2p_{\min}\},\qquad s:=|S|.
\end{equation*}
If $T_{\min}\ge 1$, then in particular no label in $S$ is empty, i.e.,
\begin{equation*}
    \{T_{\min}\ge 1\}\ \subseteq\ \bigcap_{\ell\in S}\{T_\ell\ge 1\}
    \ =\ \left\{\sum_{\ell\in S}\mathbf{1}\{T_\ell=0\}=0\right\}.
\end{equation*}
Define the number of empty labels in $S$ by
\begin{equation*}
    Z_S\ :=\ \sum_{\ell\in S}\mathbf{1}\{T_\ell=0\}.
\end{equation*}
Then $\Pr(T_{\min}\ge 1)\le \Pr(Z_S=0)$, so a necessary condition for $\Pr(T_{\min}\ge 1)\ge 1-\delta$
is
\begin{equation}
    \Pr(Z_S=0)\ \ge\ 1-\delta.
    \label{eq:Zs_need}
\end{equation}

For any $\ell\in S$, since $p_\ell\le 2p_{\min}$, we have
\begin{equation}
\begin{aligned}
\Pr(T_\ell=0)&=(1-p_\ell)^M \ \ge\ (1-2p_{\min})^M\\
    &\ \ge\ \exp(-4Mp_{\min}),
    \label{eq:empty_prob_lb_rig}
\end{aligned}
\end{equation}
where the last inequality uses $1-x\ge e^{-2x}$ for $x\in[0,1/2]$ (and we may assume $p_{\min}\le 1/4$;
otherwise $p_{\min}=\Theta(1)$ and the claimed scaling lower bound is trivial).
Therefore,
\begin{equation}
    \dsE[Z_S]
    \ =\ \sum_{\ell\in S}\Pr(T_\ell=0)
    \ \ge\ s\,e^{-4Mp_{\min}}.
    \label{eq:EZs_lb_rig}
\end{equation}

Next we upper bound $\Pr(Z_S=0)$ in terms of $\dsE[Z_S]$.
For multinomial occupancy, the indicators $\{\mathbf{1}\{T_\ell=0\}\}_{\ell\in S}$
are negatively associated, and hence
\begin{equation}
\begin{aligned}
\Pr(Z_S=0)
    &\ =\ \Pr\Big(\bigcap_{\ell\in S}\{T_\ell\ge 1\}\Big)\\
    &\ \le\ \prod_{\ell\in S}\Pr(T_\ell\ge 1)\\
    &\ =\ \prod_{\ell\in S}\big(1-\Pr(T_\ell=0)\big).
    \label{eq:neg_assoc_step}
\end{aligned}
\end{equation}
Using $1-x\le e^{-x}$ for all $x\in[0,1]$, we further obtain
\begin{equation}
    \Pr(Z_S=0)
    \ \le\ \exp\!\Big(-\sum_{\ell\in S}\Pr(T_\ell=0)\Big)
    \ =\ \exp\!\big(-\dsE[Z_S]\big).
    \label{eq:PZ0_le_expEZ}
\end{equation}
Combining \eqref{eq:Zs_need} and \eqref{eq:PZ0_le_expEZ} yields the necessary condition
\begin{equation}
\begin{aligned}
&1-\delta\ \le\ \Pr(Z_S=0)\ \le\ \exp\!\big(-\dsE[Z_S]\big)
    \\
    &\Longrightarrow\quad
    \dsE[Z_S]\ \le\ \log\frac{1}{1-\delta}.
    \label{eq:EZs_necessary}
\end{aligned}
\end{equation}
Finally, plugging \eqref{eq:EZs_lb_rig} into \eqref{eq:EZs_necessary} gives
\begin{equation*}
    s\,e^{-4Mp_{\min}}
    \ \le\ \log\frac{1}{1-\delta}.
\end{equation*}
Equivalently,
\begin{equation}
    M\ \ge\ \frac{1}{4p_{\min}}
    \left(\log s-\log\log\frac{1}{1-\delta}\right).
    \label{eq:nonunif_Tmin_lower_coupon_rig}
\end{equation}
In particular, for $\delta\in(0,1/2)$ we have $\log\frac{1}{1-\delta}=\Theta(\delta)$ and hence
$\log\log\frac{1}{1-\delta}=O(\log\log(1/\delta))$, so \eqref{eq:nonunif_Tmin_lower_coupon_rig}
implies the scaling lower bound
\begin{equation}
    M\ =\ \Omega\!\left(\frac{\log(s/\delta)}{p_{\min}}\right).
    \label{eq:nonunif_Tmin_lower_coupon}
\end{equation}

Combining Eqs.~(\ref{eq:nonunif_Tmin_lower_simplified}) and (\ref{eq:nonunif_Tmin_lower_coupon})
gives the overall necessary scaling
\begin{equation}
    M\ =\ \Omega\!\left(\frac{t+\log(s/\delta)}{p_{\min}}\right).
\label{eq:app-occupancy-combined-necessary-scaling}
\end{equation}
In the uniform case $p_\ell=1/n$ (hence $p_{\min}=1/n$ and $s=n$), this reduces to
$M=\Omega\!\big(n(t+\log(n/\delta))\big)$.

\subsection{Moment-matching construction and pair-label realization}
\label{Appendix_worst_case_construction}

This subsection provides the scalar moment-matching and cross-fidelity-table
realization used in the pair-label construction of
Appendix~\ref{appsub: MMD-k proof of lower bound}.
Lemma~\ref{lem:moment-matching-interval} gives shifted-Legendre distributions
on an interval; the label-local SWAP-test construction applies it on
$[0,\alpha/N]$ and uses Lemma~\ref{lem:realize-bounded-fidelity-table}
to realize the resulting independent pair-label cross-fidelity table.
The general-measurement proofs in Appendices~\ref{app: proof lower general}
and~\ref{app:wasserstein-general-lower} use their own finite quantum-state
constructions.

\begin{lemma}[Moment-matching measures on an interval]
\label{lem:moment-matching-interval}
Fix $k\ge 1$. There exist constants $c_k,C_k>0$, depending only on $k$, such
that the following holds. For every $a>0$ and every $0<\Delta\le c_k a^k$,
there exist two probability distributions $\mu_0,\mu_1$ supported on $[0,a]$
such that
\begin{equation}
    \int_0^a x^r\, d\mu_0(x)
    =
    \int_0^a x^r\, d\mu_1(x),
    \qquad r=0,1,\dots,k-1,
\label{eq:app-interval-moment-matching}
\end{equation}
and
\begin{equation}
    \int_0^a x^k\, d\mu_1(x)
    -
    \int_0^a x^k\, d\mu_0(x)
    =
    \Delta.
\label{eq:app-interval-prescribed-moment-gap}
\end{equation}
Moreover, for every $j\ge k$,
\begin{equation}
    \left|
    \int_0^a x^j\, d\mu_1(x)
    -
    \int_0^a x^j\, d\mu_0(x)
    \right|
    \le
    C_k\Delta a^{j-k}.
\label{eq:app-interval-higher-moment-bound}
\end{equation}
\end{lemma}

\begin{proof}
Define $\mu_0$ to be the uniform distribution on $[0,a]$, i.e., it has density
\begin{equation*}
    f_0(x) := \frac{1}{a}\,\mathbf{1}_{[0,a]}(x).
\end{equation*}
Let $P_k(\cdot)$ be the Legendre polynomial of degree $k$ on $[-1,1]$, and define its shift to $[0,a]$ by
\begin{equation*}
    \widetilde{P}_k(x) := P_k\!\left(\frac{2x}{a}-1\right),\qquad x\in[0,a].
\end{equation*}
For any parameter $\lambda\in\mathbb{R}$ satisfying
\begin{equation*}
    |\lambda|\le \frac{1}{\sup_{t\in[-1,1]}|P_k(t)|},
\end{equation*}
define $\mu_\lambda$ to be the distribution on $[0,a]$ with density
\begin{equation}
    f_\lambda(x) := \frac{1}{a}\Bigl(1+\lambda\,\widetilde{P}_k(x)\Bigr)\mathbf{1}_{[0,a]}(x).
\label{eq:app-legendre-perturbed-density}
\end{equation}
Then $f_\lambda(x)\ge 0$ on $[0,a]$, hence $\mu_\lambda$ is a valid probability measure supported on $[0,a]$.
Moreover, by orthogonality of Legendre polynomials,
\begin{equation*}
\begin{gathered}
\int_0^a x^r\, d\mu_0(x) \;=\; \int_0^a x^r\, d\mu_\lambda(x)
    \\ \text{for all } r=0,1,\dots,k-1.
\end{gathered}
\end{equation*}
For the $k$-th moment, define
\begin{equation}
    \kappa_k
    :=
    \left|
    \int_0^1 y^k P_k(2y-1)\,dy
    \right|
    =
    \frac{(k!)^2}{(2k+1)!}.
\label{eq:app-legendre-moment-normalization}
\end{equation}
Then
\begin{equation*}
\left|
\int_0^a x^k\, d\mu_0(x)
-
\int_0^a x^k\, d\mu_\lambda(x)
\right|
=
|\lambda|\,\kappa_k a^k.
\end{equation*}
Choosing the sign of $\lambda$ and setting $|\lambda|=\Delta/(\kappa_k a^k)$ gives the desired $k$-th moment gap, provided $\Delta\le c_k a^k$ for a sufficiently small constant $c_k>0$.

Finally, for every $j\ge k$,
\begin{equation*}
\left|
\int_0^a x^j\, d\mu_0(x)
-
\int_0^a x^j\, d\mu_\lambda(x)
\right|
\le
C_k|\lambda|a^j
=
C_k\Delta a^{j-k}.
\end{equation*}
Taking $\mu_1=\mu_\lambda$ completes the proof.
\end{proof}

To construct the two pairs of ensembles mentioned in \ref{appsub: MMD-k proof of lower bound}, we first show that $X_\ell\overset{\text{i.i.d.}}{\sim}\mu_0$ and $X'_\ell\overset{\text{i.i.d.}}{\sim}\mu_1$ can be constructed by two specific pairs of ensembles, as long as $\mu_0,\mu_1 \subset [0,\alpha/N]$ for $\alpha\in(0,1)$.
\begin{lemma}[Realizability of an entrywise-bounded cross-\allowbreak fidelity table]
\label{lem:realize-bounded-fidelity-table}
Let $N\ge 1$ and let $X=(X_{ij})\in[0,1]^{N\times N}$ satisfy
\begin{equation}
    0\le X_{ij} < \frac{1}{N}\qquad \text{for all } i,j\in\{1,\dots,N\}.
\end{equation}
Then there exist two ensembles of $N$ pure states $\{\ket{\psi_i}\}_{i=1}^N$ and $\{\ket{\phi_j}\}_{j=1}^N$
in a Hilbert space of dimension $2N$ such that
\begin{equation}
    \Abs{\braket{\psi_i}{\phi_j}}^2 = X_{ij}\qquad \text{for all } i,j\in\{1,\dots,N\}.
\label{eq:app-cross-fidelity-realization}
\end{equation}
Moreover, both ensembles can be chosen to contain $N$ pairwise distinct states.
\end{lemma}

\begin{proof}
Fix an orthonormal basis $\{e_1,\dots,e_N,f_1,\dots,f_N\}$ of $\mathbb{C}^{2N}$.
Define
\begin{equation*}
    \ket{\psi_i} := \ket{e_i},\qquad i=1,\dots,N.
\end{equation*}
Since $X_{ij}<1/N$ for all $i,j$, for each fixed column $j$ we have
\begin{equation*}
    \sum_{i=1}^N X_{ij} < \sum_{i=1}^N \frac{1}{N} = 1,
\end{equation*}
so the slack term $1-\sum_{i=1}^N X_{ij}$ is strictly positive. Define, for each $j=1,\dots,N$,
\begin{equation}
\label{eq:phi-construction}
    \ket{\phi_j}
    := \sum_{i=1}^N \sqrt{X_{ij}}\;\ket{e_i}
    + \sqrt{\,1-\sum_{i=1}^N X_{ij}\,}\;\ket{f_j}.
\end{equation}
Then
\begin{equation*}
    \norm{\phi_j}^2
    = \sum_{i=1}^N X_{ij} + \Bigl(1-\sum_{i=1}^N X_{ij}\Bigr)
    = 1,
\end{equation*}
so each $\ket{\phi_j}$ is a unit vector. Moreover, for all $i,j$,
\begin{equation*}
    \braket{\psi_i}{\phi_j} = \braket{e_i}{\phi_j} = \sqrt{X_{ij}},
\end{equation*}
hence $\Abs{\braket{\psi_i}{\phi_j}}^2=X_{ij}$ as claimed.

Finally, the states $\{\ket{\psi_i}\}_{i=1}^N$ are orthonormal and therefore pairwise distinct.
Also, since $1-\sum_{i=1}^N X_{ij}>0$, the component of $\ket{\phi_j}$ along $\ket{f_j}$ in
\eqref{eq:phi-construction} is nonzero, and $\braket{f_j}{\phi_{j'}}=0$ for all $j'\neq j$.
Thus $\ket{\phi_j}\neq \ket{\phi_{j'}}$ whenever $j\neq j'$, i.e., $\{\ket{\phi_j}\}_{j=1}^N$ are pairwise distinct.
\end{proof}

For the label-local SWAP-test hard instance, fix integers $N\ge 1$ and $k\ge 1$, choose a constant $\alpha\in(0,1)$, and set $a:=\alpha/N$. Applying Lemma~\ref{lem:moment-matching-interval} with interval length $a$ and $\Delta=2\epsilon$, whenever $2\epsilon\le c_k a^k$, gives two probability distributions $\mu_0,\mu_1$ supported on $[0,a]=[0,\alpha/N]$ such that
\begin{equation*}
    \int_0^a x^r\, d\mu_0(x)
    =
    \int_0^a x^r\, d\mu_1(x),
    \qquad r=0,1,\dots,k-1,
\end{equation*}
and
\begin{equation}
    \left|
    \int_0^a x^k\, d\mu_0(x)
    -
    \int_0^a x^k\, d\mu_1(x)
    \right|
    =
    2\epsilon.
\label{eq:app-pair-label-prescribed-gap}
\end{equation}
Moreover, for every $j\ge k$,
\begin{equation}
\begin{aligned}
\left|
    \int_0^a x^j\, d\mu_0(x)
    -
    \int_0^a x^j\, d\mu_1(x)
    \right|
    &\le
    C_k\epsilon a^{j-k}\\
    &\le
    C_{k,\alpha}\epsilon N^{-(j-k)}.
\end{aligned}
\label{eq:app-pair-label-higher-moment-scaling}
\end{equation}
Drawing $X_{ij}\overset{\mathrm{i.i.d.}}{\sim}\mu_b$ under $H_b$ and applying Lemma~\ref{lem:realize-bounded-fidelity-table} gives the pair-label hard instance used in Appendix~\ref{appsub: MMD-k proof of lower bound}.

\end{document}